\documentclass[journal,onecolumn,draftclsnofoot,12pt]{IEEEtran}

\IEEEoverridecommandlockouts

\makeatletter
\newcommand*\titleheader[1]{\gdef\@titleheader{#1}}
\AtBeginDocument{%
  \let\st@red@title\@title
  \def\@title{%
    \bgroup\normalfont\large\centering\@titleheader\par\egroup
    \vskip1.5em\st@red@title}
}
\makeatother

\usepackage{epsfig}
\usepackage{url}
\usepackage{upgreek}
\usepackage{bm}
\usepackage{cite}
\usepackage{subcaption}
\usepackage{amsmath}
\usepackage{amsfonts}
\usepackage{amssymb}
\usepackage{graphicx}
\usepackage{array}
\usepackage{nicefrac}
\usepackage{textcomp}
\usepackage{xcolor}
\usepackage{tikz}
\usepackage{bm}
\usepackage{algorithmicx}
\usepackage{algorithm}
\usepackage[noend]{algpseudocode}
\usepackage{varwidth}
\usepackage{multicol}
\usepackage{mathtools}
\usepackage{cuted}
\usepackage{stfloats}
\usepackage{xr}
\usepackage{sansmath,etoolbox}

\appto{\sffamily}{\sansmath}
\appto{\rmfamily}{\unsansmath}

\def\sfA{\mathsf{A}}
\def\sfB{\mathsf{B}}
\def\sfC{\mathsf{C}}
\def\sfD{\mathsf{D}}

\def\sfH{\mathsf{H}}
\def\sfI{\mathsf{I}}

\def\sfP{\mathsf{P}}
\def\sfQ{\mathsf{Q}}
\def\sfR{\mathsf{R}}
\def\sfS{\mathsf{S}}
\def\sfT{\mathsf{T}}
\def\sfU{\mathsf{U}}
\def\sfV{\mathsf{V}}
\def\sfW{\mathsf{W}}
\def\sfX{\mathsf{X}}
\def\sfY{\mathsf{Y}}

\def\sfb{{\mathsf{b}}}
\def\sfc{{\mathsf{c}}}

\def\sfee{{\mathsf{e}}}

\def\sfh{{\mathsf{h}}}

\def\sfp{{\mathsf{p}}}
\def\sfq{{\mathsf{q}}}

\def\sfs{{\mathsf{s}}}
\def\sft{{\mathsf{t}}}
\def\sfu{{\mathsf{u}}}
\def\sfv{{\mathsf{v}}}
\def\sfw{{\mathsf{w}}}
\def\sfx{{\mathsf{x}}}
\def\sfy{{\mathsf{y}}}

\def\sf0{{\mathsf{0}}}

\usepackage{scalerel}
\usetikzlibrary{svg.path}

\definecolor{orcidlogocol}{HTML}{A6CE39}
\tikzset{
  orcidlogo/.pic={
    \fill[orcidlogocol] svg{M256,128c0,70.7-57.3,128-128,128C57.3,256,0,198.7,0,128C0,57.3,57.3,0,128,0C198.7,0,256,57.3,256,128z};
    \fill[white] svg{M86.3,186.2H70.9V79.1h15.4v48.4V186.2z}
                 svg{M108.9,79.1h41.6c39.6,0,57,28.3,57,53.6c0,27.5-21.5,53.6-56.8,53.6h-41.8V79.1z M124.3,172.4h24.5c34.9,0,42.9-26.5,42.9-39.7c0-21.5-13.7-39.7-43.7-39.7h-23.7V172.4z}
                 svg{M88.7,56.8c0,5.5-4.5,10.1-10.1,10.1c-5.6,0-10.1-4.6-10.1-10.1c0-5.6,4.5-10.1,10.1-10.1C84.2,46.7,88.7,51.3,88.7,56.8z};
  }
}

\newcommand\orcidicon[1]{\href{https://orcid.org/#1}{\mbox{\scalerel*{
\begin{tikzpicture}[yscale=-1,transform shape]
\pic{orcidlogo};
\end{tikzpicture}
}{|}}}}
\usepackage{hyperref} 

\newcommand{\ji}{\text{j}}
\newcommand{\D}{\text{d}}
\newcommand{\U}{\text{u}}
\newcommand{\R}{\text{Re}}
\newcommand{\I}{\text{Im}}

\newcommand{\T}{\text{T}}
\newcommand{\maxx}{\text{max}}

\newcommand{\Her}{\text{H}}
\newcommand{\diag}{\text{diag}}
\newcommand{\bdiag}{\text{blkDiag}}
\newcommand{\ndiag}{\text{nd}}
\newcommand{\tr}{\text{tr}}
\newcommand{\thh}{\text{th}}
\newcommand{\NBS}{N_\text{BS}}
\newcommand{{\BS}}{\text{BS}}
\newcommand{\vecc}{\text{vec}}
\newcommand{\SC}{N_\text{SC}}
\newcommand{\G}{N_\mathcal{G}}
\newcommand{\SG}{N_\mathcal{S}}
\newcommand{\PA}{\text{PA}}
\newcommand{\QPA}{\bar{\mathbf{Q}}_\PA}
\newcommand{\CP}{N_\text{CP}}
\newcommand{\ext}{\text{ext}}
\newcommand{\kn}{{k,n}}
\newcommand{\inn}{{i,n}}
\newcommand{\ijj}{{i,j}}
\newcommand{\CDL}{R^\text{DL}}
\newcommand{\CUL}{R^\text{UL}}
\newcommand*{\rom}[1]{\expandafter\@slowromancap\romannumeral #1@}

\DeclareMathOperator*{\argmax}{arg\,max}
\DeclareMathOperator*{\argmin}{arg\,min}

\newtheorem{theorem}{Theorem}[section]
\newtheorem{corollary}{Corollary}[theorem]
\newtheorem{lemma}[theorem]{Lemma}
\newenvironment{proof}[1]{\par\noindent\underline{Proof:}\space#1}{\hfill $\blacksquare$}

\def\BibTeX{{\rm B\kern-.05em{\sc i\kern-.025em b}\kern-.08em
    T\kern-.1667em\lower.7ex\hbox{E}\kern-.125emX}}
\begin{document}

\title{Multi-user Downlink Beamforming using Uplink Downlink Duality  with CEQs for Frequency Selective Channels}


\author{Khurram Usman~Mazher$^{\textsuperscript{\orcidicon{0000-0001-8760-391X}}}$,~\IEEEmembership{Student Member,~IEEE,} Amine~Mezghani$^{\textsuperscript{\orcidicon{0000-0002-7625-9436}}}$,~\IEEEmembership{Member,~IEEE,}
        and Robert W.~Heath Jr.$^{\textsuperscript{\orcidicon{0000-0002-4666-5628}}}$,~\IEEEmembership{Fellow,~IEEE}
\thanks{K. U. Mazher (khurram.usman@utexas.edu) is with the Wireless Networking and Communications Group, The University of Texas at Austin, Austin, TX 78712 USA.}
\thanks{A. Mezghani (amine.mezghani@umanitoba.ca) is with the Department of Electrical and Computer Engineering at the University of Manitoba, Winnipeg, MB R3T 2N2, Canada.}
\thanks{R. W. Heath Jr. (rwheathjr@ncsu.edu) is with 6GNC, Department of Electrical and Computer Engineering at the North Carolina State University, Raleigh, NC 27695 USA.}
}

\maketitle

\vspace{-1.2cm}
\begin{abstract}
High-resolution fully digital transceivers are infeasible at millimeter-wave (mmWave) due to their increased power consumption, cost, and hardware complexity. The use of low-resolution converters is one possible solution to realize fully digital architectures at mmWave. In this paper, we consider a setting in which a fully digital base station with constant envelope quantized (CEQ) digital-to-analog converters on each radio frequency chain communicates with multiple single antenna users with individual signal-to-quantization-plus-interference-plus-noise ratio (SQINR) constraints over frequency selective channels. We first establish uplink downlink duality for the system with CEQ hardware constraints and OFDM-based transmission considered in this paper. Based on the uplink downlink duality principle, we present a solution to the multi-user multi-carrier beamforming and power allocation problem that maximizes the minimum SQINR over all users and sub-carriers. We then present a \emph{per sub-carrier} version of the originally proposed solution that decouples all sub-carriers of the OFDM waveform resulting in smaller sub-problems that can be solved in a parallel manner. Our numerical results based on 3GPP channel models generated from \emph{Quadriga} demonstrate improvements in terms of ergodic sum rate and ergodic minimum rate over state-of-the-art linear solutions. We also show improved performance over non-linear solutions in terms of the coded bit error rate with the increased flexibility of assigning individual user SQINRs built into the proposed framework.
\end{abstract}


\vspace{-.2cm}
\begin{IEEEkeywords}
\vspace{-.3cm}
CEQ ADC, CEQ DAC, optimized dithering, uplink downlink duality, OFDM
\end{IEEEkeywords}

\section{Introduction}\label{sec:intro}


Low-resolution analog to digital converters (ADCs) and digital to analog converters (DACs) are the key to power-efficient fully digital massive multiple-input-multiple-output (MIMO) transceivers operating at large bandwidths \cite{SP,NYU,powerStuder}. In addition to a reduction in the power consumption and cost of the DACs, other components in the radio frequency (RF) chain can be tailored to low-resolution DACs (such as power amplifiers and baseband processing) making the architecture even more efficient. The distortion resulting from the low-resolution DACs can be compensated by oversampling in time/space and advanced signal processing algorithms. In this paper, we consider the setting where multiple single antenna users are communicating with a fully digital base station (BS) with constant envelope quantizer (CEQ) DACs on each RF chain using the orthogonal frequency division multiplexing (OFDM) waveform in the downlink (DL). We design the frequency domain precoders and power allocation by maximizing the minimum signal-to-quantization-plus-interference-plus-noise ratio (SQINR) across all users and sub-carriers. This optimization criterion has not been considered before for CEQ OFDM based DL transmissions.

\vspace{-0.cm}
\subsection{Prior work}
Prior work on multi-user (MU) MIMO-OFDM DL under low resolution DAC constraints can be grouped into linear methods \cite{StuderOFDMconf,StuderOFDM,amine2009,amine2016} and non-linear methods \cite{StuderSQUID,MSM1,MSM2,MSM3,magic,magic2,access,oob,oob2}. In the linear framework, information symbols are mapped to the antennas using a precoding matrix designed based on an optimization criterion. Most of the existing prior work on linear methods (under low resolution constraints) is limited to using precoders designed for the ideal $\infty$-resolution setting. The low-resolution DAC constraint is enforced by quantizing the signal before transmission. Prior work on MRT, ZF, and minimizing the mean square error (MMSE) based linear precoding has shown that large sum rates are achievable for MU-MIMO-OFDM DL despite the extreme distortion caused by 1-bit quanitzation \cite{StuderOFDMconf,StuderOFDM,amine2009,amine2016}. Further improvement in terms of the achieved SQINR and uncoded bit error rate (BER) for ZF precoding (for flat fading channels) was demonstrated by adding optimized dithering to the transmit signal before quantization \cite{amodh}. Nevertheless, linear methods \cite{StuderOFDMconf,StuderOFDM,amine2009,amine2016} have a significant performance gap from non-linear methods \cite{StuderSQUID,MSM1,MSM2,MSM3,magic,magic2,access,oob,oob2} particularly for large number of active users. One exception to this observation is our prior work in \cite{pw,pw2} where we demonstrated performance comparable to non-linear methods for flat fading channels.

Non-linear methods directly map each set of information symbols to the quantized transmit signal by solving a relaxed version of an NP-hard problem (due to CEQ constraints) based on some optimization criterion (such as MMSE). Non-linear methods for MU-MIMO-OFDM DL with low-resolution quantized phase DAC constraints based on the semi-definite relaxation and squared $\ell_\infty$-norm relaxation of the symbol error rate (SER) were proposed in \cite{StuderSQUID}. A different approach based on maximizing the safety margin (MSM) of the received symbols (drawn from a phase shift keying (PSK) constellation) from the decision boundaries was demonstrated in \cite{MSM1,MSM3} for frequency selective channels with 1-bit DAC constraints. That work was later generalized to CEQs \cite{MSM2}. Another non-linear method for CEQ MU-MIMO-OFDM DL approximated the solution to the MSE problem formulated in time domain using a greedy coordinate descent algorithm \cite{magic}. A slightly different version of that algorithm where the greedy minimization is replaced a round robin minimization has been reported recently \cite{magic2}. Another efficient solution to the MSE problem based on cyclic coordinate descent was proposed for constant envelope MU-MIMO-OFDM \cite{access}. A slightly different but closely related solution based on Gibbs sampling optimized a linear combination of the MSE and out-of-band (OOB) radiated power and showed a reduction in the OOB power by about 10 dB at the expense of reduced throughput \cite{oob,oob2}. The performance of the non-linear methods in \cite{StuderSQUID,MSM1,MSM3,MSM2,magic,magic2,access,oob,oob2} is comparable in terms of coded and uncoded BER. Some non-linear methods, however, result in a significant computational cost for systems with larger dimensions due to exponential increase in their complexity with the system dimensionality \cite{Studer}. The focus of research in this direction \cite{StuderSQUID,MSM1,MSM3,MSM2,magic,magic2,access,oob,oob2} has been to solve the NP-hard problem using various relaxations and approximations without sacrificing on the performance. Another important aspect is that each non-linear method needs to solve an optimization problem for every channel use during the coherence time. Lastly, most of the non-linear methods \cite{StuderSQUID,MSM1,MSM3,MSM2,magic,magic2,access,oob,oob2} have hyperparameters that need to be appropriately chosen according to the operating conditions.

The prior work on CEQ MU-MIMO-OFDM DL (for both linear and non-linear methods) have primarily focused on the MSE, the SER and the MSM optimization criterion. In this paper, we introduce a per-user and per-subcarrier target SQINR framework and propose a linear precoding solution based on maximizing the minimum (max-min) SQINR over all users and sub-carriers. 

\vspace{-0.3cm}
\subsection{Contributions}
In this paper, we provide a linear precoding based solution to the MU-MIMO-OFDM DL precoding problem under CEQ hardware constraints at the BS. The BS communicates with multiple single antennas users, with individual SQINR constraints, over frequency selective channels using OFDM. We linearize the resulting non-linear system using the \emph{Bussgang} decomposition \cite{StuderOFDM} to derive the proposed linear solution. The main contributions of this paper can be summarized~as:

\begin{itemize}
\item We establish UL-DL duality for the MU-MIMO-OFDM setting with CEQ ADC/DAC~constraints under an uncorrelated quantization noise assumption. This is different from the $\infty$-resolution setting because of the introduction of quantization noise into the DL/UL SQINR expressions. Furthermore, this is different from the flat fading case because the quantization noise depends on all the sub-carriers destroying the orthogonality inherent in OFDM.


\item Making use of the UL-DL duality result, we propose an alternating minimization solution to the MU-MIMO-OFDM DL beamforming (BF) problem with CEQ DAC constraints based on the max-min SQINR criterion. The solution jointly optimizes the power allocated to each user across all sub-carriers and the frequency domain DL BF matrix and does not have any hyper-parameters that need to be tuned. We give theoretical justification and demonstrate through numerical experiments that the bigger problem involving all sub-carriers can be broken down into smaller decoupled problems for each sub-carrier in the large system~limit. 

\item We introduce optimized dithering by adding dummy users in the system which operate in the null space of the true system users. Optimized dithering ensures that the quantization noise resulting from CEQs is uncorrelated, particularly when the number of users is small.

\item We demonstrate the superiority of the proposed solution over other linear and non-linear precoding solutions \cite{StuderOFDMconf,StuderOFDM,StuderSQUID,magic2} in terms of the ergodic sum rate, ergodic minimum rate, and coded BER using numerical experiments carried out over 3GPP channel models.



\end{itemize}

The work in this paper generalizes our prior work \cite{pw,pw2} to CEQs and frequency selective channels. Our prior work \cite{pw,pw2}, limited to frequency flat channels and 1-bit DACs, was an important step towards the development of the proposed framework. It was, however, not directly applicable to large bandwidth signals being transmitted over frequency selective channels which is the most probable use case for low-resolution DACs equipped fully digital architectures. The ideas presented in our prior work were also limited to 1-bit quantization. This paper generalizes the UL-DL duality principle proved in \cite{pw2} to frequency selective channels and CEQs. With this generalization, the alternating minimization algorithm proposed in \cite{pw2} is applicable to the frequency selective setting with a few minor changes in the structure of the involved matrices. This, however, results in a large dimensional problem comprising of all sub-carriers of all users. In this paper, we argue that this bigger problem can in fact be broken down into smaller sub-problems for each sub-carrier and verify this using numerical experiments. Lastly, we highlight a few important aspects of the precoding under hardware constraints problem that seem to have been neglected in the prior work. The results from prior work \cite{StuderSQUID,MSM1,MSM3,MSM2,magic,magic2,access,oob,oob2} obtained on independent and identically distributed (IID) Rayleigh fading channels show that all linear precoding strategies hit a floor at a certain SNR/transmit power and are significantly outperformed by non-linear precoding strategies. We demonstrate through our results that the non-linear precoding strategies also floor out at a certain SNR/transmit power for realistic channel models considered in this paper and are in fact outperformed by the proposed solution over a wide range of parameters.

The rest of this paper is organized as follows. In Section \ref{sec:systemModel}, we describe the OFDM system model and formulate the MU-DL-BF problem. In Section \ref{sec:duality}, we establish the UL-DL duality principle for frequency selective channels under CEQ constraints using the uncorrelated quantization noise approximation. In Section \ref{sec:solution}, we provide the details of the joint power allocation and beamforming optimization algorithm for the MU-DL-BF problem. We present numerical results in Section \ref{sec:results} before concluding the paper with directions for future work in Section \ref{sec:conc}.

\emph{Notation:} $\bm{B}$ is a matrix, $\bm{b}$ is a vector and $b$ is a scalar. $\bm{B}_i$ and $\bm{B}_{ij}$ denotes the $i^{\thh}$ row and $i^{\thh} , j^{\thh}$ entry of the matrix $\bm{B}$. $\bm{b}_i$ denote the $i^{\thh}$ entry of $\bm{b}$. The operator $(\cdot)^\T$, $(\cdot)^\Her$, and $(\cdot)^\ast$ denote the transpose, conjugate transpose and conjugate of a matrix/vector. $\diag (\bm{B})$ denotes a diagonal matrix containing only the diagonal elements of $\bm{B}$. $\bm{B}^\ndiag = \bm{B} - \diag (\bm{B})$ denotes the matrix $\bm{B}$ with its diagonal set to 0. $\tr(\bm{B})$, $\| \bm{B} \|_F$, and $\lambda_\maxx (\bm{B})$ denote the trace, Frobenius norm, and dominant eigenvalue of the matrix $\bm{B}$. $\vecc(\bm{B})$ represents the vectorization operation applied to $\bm{B}$. $\bdiag( \bm{B}_1, \dots \bm{B}_m )$ denotes a block-diagonal matrix with the matrices $\bm{B}_1, \dots \bm{B}_m$ on its diagonal. $\mathbf{F}_N$ is the FFT matrix of size $N \times N$ normalized by $\frac{1}{\sqrt{N}}$. $\bm{I}_N$ represents the identity matrix of size $N \times N$. The vector $\bm{1}_N$ ($\bm{0}_N$) denotes a vector of all ones (zeros) of length $N$. The matrix $\bm{R}_{\bm{b}}$ denotes the covariance matrix of the signal $\text{b}$. $\|\bm{b}\|_p$ is the $p$-norm of $\bm{b}$. $\bm{e}_k$ denotes the canonical basis vector with a 1 at the $k^\thh$ index and zeros elsewhere. $\R(a)$ and $\I(a)$ denote the real and imaginary parts of $a$. The notations $|\cdot| , {(\cdot)}^k$ and $ \angle(\cdot)$ denote the absolute value, $k^{\thh}$ power and phase operation applied to a scalar or element-wise to a vector/matrix. $\mathcal{N}( \bm{\mu} , \bm{\Sigma} )$ denotes a complex Gaussian multi-variate distribution with mean  $\bm{\mu}$ and covariance $\bm{\Sigma}$. 

\section{System model} \label{sec:systemModel}
We consider a DL scenario where a single BS with $\NBS$ antennas and RF chains, each equipped with a $b$-bit CEQ DAC, communicates with $K$ single antenna users using the OFDM waveform with $\SC$ sub-carriers, as illustrated in Fig. \ref{fig:system}. The quantization operation of the $b$-bit CEQ is represented by $\mathcal{Q}_b$. The $b$-bit CEQ takes values in the set $\mathcal{X}_b = \{ e^{\ji( \pi + 2\pi m )/2^b} \}$ for $m \in \{ 0, \dots, 2^b -1\}$ and can be efficiently implemented using polar amplifier based transmitter structures \cite{polar}. For the $\infty$-bit CEQ, $\mathcal{X}_\infty$ equals the complex unit-magnitude circle and the CEQ returns the unit-norm normalized version of its input. Note that $b = 2$ corresponds to the 1-bit DAC setting. We let $\mathcal{S}$ ( $\mathcal{G}$) denote the set of occupied (guard) subcarriers with $|\mathcal{S}| = \SG$ ($|\mathcal{G}| = \G$) and $\SC = \SG + \G$. Each OFDM symbol has a cyclic prefix (CP) of length $\CP$ to ensure that the linear convolution over the wireless channel can be replaced by circular convolution. For ease of exposition, all scalars/vectors/matrices defined in the frequency domain are in the san serif font (e.g. $\bm{\sfX}$). Similarly, vectorized version of matrices and block-diagonal matrices have a $\bar{(\cdot)}$ on top of them. Furthermore, we describe the system model and other development under the assumption that $\mathcal{G} = \emptyset$, i.e. all $\SC$ sub-carriers are active for ease of exposition. The analysis and algorithm development presented in this work can be mapped to any arbitrary $\mathcal{S}$. 
 
 \begin{figure}[t]
    	\begin{center}
    		\includegraphics[width=.99\textwidth,clip,keepaspectratio]{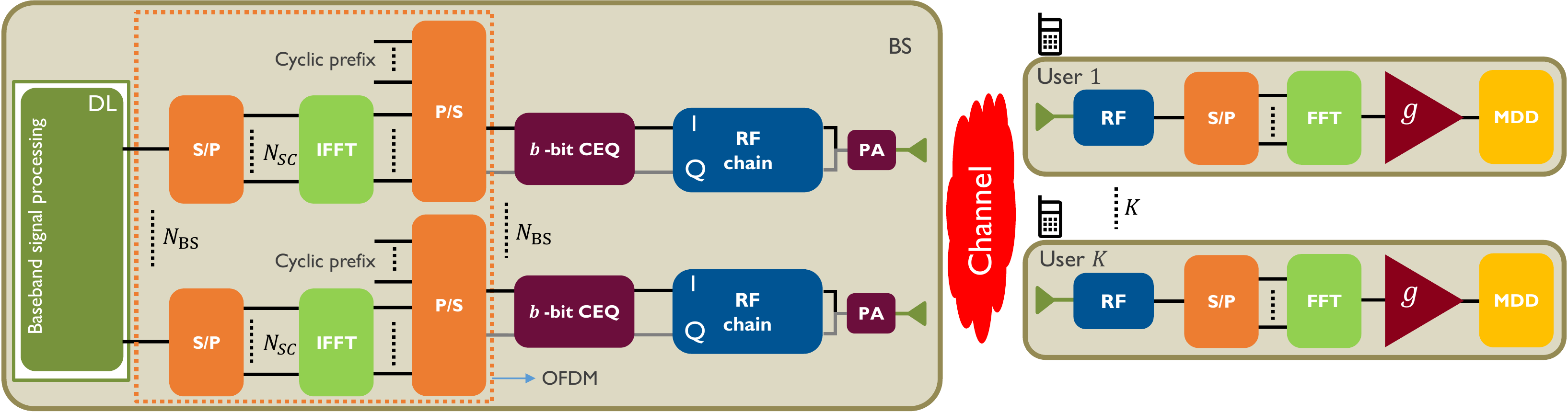}
    	\end{center}
\vspace{-.5cm}
    	\caption{Functional block diagram of the system model where a fully digital BS with ${\NBS}$ antennas and a $b$-bit CEQ ADC/DAC on each RF chain communicates with $K$ single antenna users using the OFDM waveform with $\SC$ sub-carriers.}
    	 \label{fig:system}
\vspace{-.7cm}
\end{figure}

 
 The $\ell^\thh$-tap of the $L$-tap channel from the BS to all $K$ users is denoted by $\mathbf{{H}}_{\ell} \in \mathbb{C}^{\NBS \times K}$. The frequency domain channel on the $n^\thh$ sub-carrier, $\bm{\sfH}_n \in \mathbb{C}^{\NBS \times K}$, is obtained by taking an $\SC$-size FFT of ${\mathbf{H}}_{\ell}$ over the channel tap dimension. The $k^\thh$ column of $\bm{\sfH}_n$, $\bm{\sfh}_\kn$, denotes the frequency domain channel of the $k^\thh$ user on the $n^\thh$ sub-carrier. The BS sends the IID  $\mathcal{N}( 0, 1 )$ signal $\sfs_\kn$ to the $k^\thh$ users on the $n^\thh$ sub-carrier for $1 \leq k \leq K$ and $1 \leq n \leq \SC$. With $\bm{\sfs}_n = [\sfs_{1,n}, \dots \sfs_\kn, \dots \sfs_{K,n}]^\T \in \mathbb{C}^{K}$ denoting the symbols on the $n^\thh$ sub-carrier, the information symbol matrix for all $\SC$ sub-carriers is given by $\bm{\sfS} =  [\bm{\sfs}_1, \dots \bm{\sfs}_{\SC}] \in \mathbb{C}^{K \times \SC}$. During the DL stage, the symbols $\bm{\sfs}_n$ are mapped to the antenna array using the BF matrix $\bm{\sfT}_n = [\bm{\sft}_{1,n} , \dots , \bm{\sft}_{K,n}] \in \mathbb{C}^{\NBS \times K}$, where $\| \bm{\sft}_{k,n} \|_2 = 1$. The BS has a total transmit power constraint of $P_{\BS}$~Watts.
 
We also describe the corresponding UL scenario (with $b$-bit ADCs at the BS) where the $K$ users send information symbols to the BS. The $b$-bit ADCs are not a requirement and can be thought of as a mathematical construct for the purpose of this paper. During the UL stage, the symbols $\bm{\sfs}_n$ are resolved at the BS using the BF matrix $\bm{\sfU}_n = [\bm{\sfu}_{1,n} , \dots , \bm{\sfu}_{K,n}] \in \mathbb{C}^{\NBS \times K}$. The $K$ users transmit under a \emph{sum} power constraint of $P_{\BS}$ Watts  equal to the total power constraint of the BS during the DL stage. We conclude this Section by introducing the proposed max-min problem formulation and the small angle approximation which will be used in Section \ref{sec:duality} for proving UL-DL duality under CEQ constraints.

\vspace{-0.3cm}
\subsection{Downlink SQINR for CEQ DACs} \label{sec:DL}
With $\bm{\sfq}_n = [\sfq_{1,n} , \dots , \sfq_{K,n}]^\T$ denoting the DL power allocation vector over the $n^\thh$ sub-carrier for all $K$ users, let $\bm{\sfq} = [\bm{\sfq}_1^\T, \dots \bm{\sfq}_{\SC}^\T]^\T$. With $\bm{\sfQ}_n \triangleq \diag( \sqrt{ \bm{\sfq}_n} ) \in \mathbb{C}^{K \times K}$, let $\bm{\sfx}_n = \bm{\sfT}_n \bm{\sfQ}_n \bm{\sfs}_n  \in \mathbb{C}^{{\NBS}} $ denote the beamformed signal of the $n^\thh$ sub-carrier at the BS just prior to cyclic prefix addition and conversion to time domain. Defining $\bm{\sfX}_\D = [\bm{\sfx}_1, \dots \bm{\sfx}_{\SC}] \in \mathbb{C}^{{\NBS} \times \SC}$, the $\infty$-resolution time domain signal matrix (which comprises one OFDM symbol without explicitly accounting for the CP) after the IFFT operation is given by ${\mathbf{X}}_\D = [{\mathbf{x}}_1, \dots {\mathbf{x}}_{\SC}] = \bm{\sfX}_\D \mathbf{F}_{\SC}^{\Her} \in \mathbb{C}^{{\NBS} \times \SC}$. This signal is sent over the wireless channel and received at the $K$ users corrupted by additive IID $\mathcal{N}( 0 , \sigma^2 )$ noise ${\mathbf{W}}_\D = [{\mathbf{w}}_1, \dots {\mathbf{w}}_{\SC}]  \in \mathbb{C}^{K \times \SC}$. The received time domain signal (after discarding the CP) ${\mathbf{Y}}_\D = [{\mathbf{y}}_1, \dots {\mathbf{y}}_{\SC}]  \in \mathbb{C}^{K \times \SC}$ is converted to the frequency domain using the FFT operation as $\bm{\sfY}_\D = [\bm{\sfy}_1,\dots \bm{\sfy}_{\SC}] = {\mathbf{Y}}_\D \mathbf{F}_{\SC}$. The noise statistics of the frequency domain noise $\bm{\sfW}_\D = [\bm{\sfw}_1,\dots \bm{\sfw}_{\SC}] = {\mathbf{W}}_\D \mathbf{F}_{\SC}$ are preserved under the unitary FFT operation. After this transformation, the signal received and sent on the $n^\thh$ sub-carrier are related by $\bm{\sfy}_n = \bm{\sfH}_n^\T \bm{\sfx}_n + \bm{\sfw}_n$. Define $\bar{\bm{\sfH}} = \bdiag(\bm{\sfH}_1,\dots \bm{\sfH}_{\SC})$, $\bar{\bm{\sfQ}} = \bdiag(\bm{\sfQ}_1,\dots \bm{\sfQ}_{\SC})$, $\bar{\bm{\sfT}} = \bdiag(\bm{\sfT}_1,\dots \bm{\sfT}_{\SC})$, $\bar{\bm{\sfy}}_\D = \vecc(\bm{\sfY}_\D)$, $\bar{\bm{\sfx}}_\D = \vecc(\bm{\sfX}_\D)$, $\bar{{\mathbf{x}}}_\D = \vecc({\mathbf{X}}_\D)$, $\bar{\bm{\sfw}}_\D = \vecc(\bm{\sfW}_\D)$ and $\bar{\bm{\sfs}} = \vecc( \bm{\sfS} )$. Using the matrix identity $\vecc(\bm{\sfA\sfB\sfC}) = \left(\bm{\sfC}^\T \otimes \bm{\sfA}\right) \vecc(\bm{\sfB})$, it can be shown that 
\begin{equation} \label{eq:downlinkSig_r1}
\begin{split}
\bar{\bm{\sfy}}_\D &= \bar{\bm{\sfH}}^\T \bar{\bm{\sfx}}_\D + \bar{\bm{\sfw}}_\D  = \bar{\bm{\sfH}}^\T \left( \mathbf{F}_{\SC} \otimes \bm{\sfI}_{{\NBS}} \right)  \underbrace{ \left( \mathbf{F}_{\SC}^{\Her} \otimes \bm{\sfI}_{{\NBS}} \right) \bar{\bm{\sfT}} \bar{\bm{\sfQ}} \bar{\bm{\sfs}} }_{\text{time domain signal } \bar{{\mathbf{x}}}_\D} + \bar{\bm{\sfw}}_\D = \bar{\bm{\sfH}}^\T \bar{\bm{\sfT}} \bar{\bm{\sfQ}} \bar{\bm{\sfs}} + \bar{\bm{\sfw}}_\D .
\end{split}
\end{equation}
The block-diagonal nature of the matrices in (\ref{eq:downlinkSig_r1}), which corresponds to the OFDM input-output system model with $\infty$-resolution DACs and has been used in prior work \cite{StuderOFDM}, emphasizes that each sub-carrier can be treated independently of the others.



Next we introduce CEQ in the system model and linearize the resulting non-linearity using \emph{Bussgang} decomposition. The vectorized time domain signal after the CEQ operation is given by $\bar{{\mathbf{z}}}_\D \triangleq \vecc( {\mathbf{Z}}_\D ) = \vecc( \mathcal{Q}_b ( {\mathbf{X}}_\D ) ) = \mathcal{Q}_b( \bar{{\mathbf{x}}}_\D )$. Using (\ref{eq:downlinkSig_r1}), the received symbols are now given by $\bar{\bm{\sfy}}_\D = \bar{\bm{\sfH}}^\T \left( \mathbf{F}_{\SC} \otimes \bm{\sfI}_{{\NBS}} \right) \mathcal{Q}_b \left( \left( \mathbf{F}^\Her_{\SC} \otimes \bm{\sfI}_{{\NBS}} \right) \bar{\bm{\sfT}} \bar{\bm{\sfQ}} \bar{\bm{\sfs}} \right) + \bar{\bm{\sfw}}_\D$. The orthogonal nature of the sub-carriers is destroyed due to the CEQ operation since the FFT and IFFT operations can not cancel each other out. The \emph{Bussgang} theorem \cite{SE} can be used to decompose the signal into a useful linear part and an uncorrelated distortion $\bar{\pmb{\eta}}_\text{d}$ with the covariance  $\mathbf{R}_{\bar{\pmb{\eta}}_\text{d}}$. Using (\ref{eq:downlinkSig_r1}), the covariance matrix of the DL signal $\bar{{\mathbf{x}}}_\D$ before CEQ quantization is given by $\mathbf{R}_{\bar{{\mathbf{x}}}_\D} = \left( \mathbf{F}^\Her_{\SC} \otimes \bm{\sfI}_{{\NBS}} \right) \bar{\bm{\sfT}} \bar{\bm{\sfQ}} \bar{\bm{\sfQ}}^\Her \bar{\bm{\sfT}}^\Her \left( \mathbf{F}^\Her_{\SC} \otimes \bm{\sfI}_{{\NBS}} \right)^\Her$. The Bussgang gain is defined as \cite{SE}
\begin{equation} \label{eq:downlinkSig_r8}
\begin{split}
\bar{\mathbf{A}}_\D &=  \zeta_b \diag \left( \mathbf{R}_{\bar{{\mathbf{x}}}_\D} \right)^{-\frac{1}{2}} \overset{(a)}{=} \bm{\sfI}_{\SC} \otimes \underbrace{\zeta_b \diag \left( \frac{1}{\SC} \sum_{n = 1}^{\SC} \bm{\sfT}_n {\bm{\sfQ}_n} \bm{\sfQ}_n^\Her \bm{\sfT}_n^\Her \right)^{-\frac{1}{2}}}_{\mathbf{A}_\D},
\end{split}
\end{equation}
where $(a)$ follows from the block-diagonal structure of the matrices involved in $\mathbf{R}_{\bar{{\mathbf{x}}}_\D}$. Here $\zeta_b$ is a constant that depends on the statistical properties of the $b$-bit CEQ and is given by $\zeta_b = \frac{2^b}{2\sqrt{pi}}\text{sin}\left( \frac{\pi}{2^b} \right)$ \cite{statistical}. Note that $\zeta_b =  \sqrt{\frac{2}{\pi}}$ for $b = 2$ and $\zeta_b =  \sqrt{\frac{\pi}{4}}$ for $b = \infty$. The signal after the CEQ DAC can be rewritten as $\bar{{\mathbf{z}}}_\D = \bar{\mathbf{A}}_\D \left( \mathbf{F}^\Her_{\SC} \otimes \bm{\sfI}_{{\NBS}} \right) \bar{\bm{\sfT}} \bar{\bm{\sfQ}} \bar{\bm{\sfs}} + \bar{\pmb{\eta}}_\text{d}$.



Power allocation done across users and sub-carriers before the CEQ DAC, captured by $\bar{\bm{\sfQ}}$, will be wiped out due to the quantization operation (due to multiplication by $\bar{\mathbf{A}}_\D$). Power allocation has to be done again on a \emph{per-antenna} basis in analog after the CEQ DAC operation. This is achieved by multiplication with the non-negative diagonal matrix $\QPA \in \mathbb{C}^{\NBS \SC \times \NBS \SC}$. It can be seen from the form of $\diag \left( \mathbf{R}_{\bar{{\mathbf{x}}}_\D} \right)$ in (\ref{eq:downlinkSig_r8}) that the resulting per-antenna power (for the $\infty-$resolution case) across the $\NBS$ antennas for each time instant $n \in \{1, \dots \SC \}$ is dependent on the beamformers and power allocation for all $\SC$ sub-carriers. Furthermore, because of the Kronecker structure, the resulting per-antenna power is the same for $n \in \{1, \dots \SC \}$. Hence we take the matrix $\QPA$ to have a Kronecker structure as well of the the form $\mathbf{I}_{\SC} \otimes \mathbf{Q}_\PA$, where $\mathbf{Q}_\PA \in \mathbb{C}^{\NBS \times \NBS}$ is the per-antenna power allocation at any time instant $n$ within the OFDM symbol. The total per-antenna power allocation at each time instant is constrained to be equal to the DL transmit power given by forcing $P_{\BS} = \text{tr}(\mathbf{Q}_\PA \mathbf{Q}_\PA^\Her)$. The linearized signal model~of the symbols received at the $K$ users after incorporating the per-antenna power allocation is given~by
\begin{equation} \label{eq:downlinkSig_r10}
\begin{split}
\bar{\bm{\sfy}}_\D = \bar{\bm{\sfH}}^\T \QPA \bar{\mathbf{A}}_\D \bar{\bm{\sfT}} \bar{\bm{\sfQ}} \bar{\bm{\sfs}} + \bar{\bm{\sfH}}^\T \left( \mathbf{F}_{\SC} \otimes \bm{\sfI}_{{\NBS}} \right) \QPA \bar{\pmb{\eta}}_\text{d} + \bar{\bm{\sfw}}_\D .
\end{split}
\end{equation}
We note here that the first $K$ columns of the block-diagonal matrix $\bar{\bm{\sfH}}$ correspond to the first sub-carrier, the next $K$ columns to the second sub-carrier and so on. Let $\bar{\bm{\sfh}}_\kn$ / $\bar{\bm{\sft}}_\kn$ denote the $((n-1)K + k)^\thh$ column of $\bar{\bm{\sfH}}$ / $\bar{\bm{\sfT}}$ which corresponds to the channel / beamformer of the $k^\thh$ user on the $n^\thh$ sub-carrier. Similarly let ${{\sfq}}_\kn$ denotes the $((n-1)K + k)^\thh$ entry of ${\bm{\sfq}}$ which corresponds to the power allocated to the $k^\thh$ user on the $n^\thh$ sub-carrier. With $\bar{\bm{\sfR}}_\kn = \bar{\bm{\sfh}}_\kn \bar{\bm{\sfh}}_\kn^\Her$, the DL SQINR for the $k^\thh$ user at the $n^\thh$ sub-carrier, $\gamma_\kn^{\text{DL}}(\bar{\bm{\sfT}} ,  \QPA , \bm{\sfq})$, is given by
\begin{equation}
\begin{split}
& \gamma_\kn^{\text{DL}}(\bar{\bm{\sfT}} ,  \QPA , \bm{\sfq}) = \\
& \frac{ \sfq_\kn \bar{\bm{\sft}}_\kn^\T \bar{\mathbf{A}}_\D \QPA \bar{\bm{\sfR}}_\kn \QPA^\Her \bar{\mathbf{A}}_\text{d}^\Her  \bar{\bm{\sft}}_\kn^*}{  \sum_{ \substack{i = 1\\ i \neq k}}^K \underbrace{ \sfq_\inn \bar{\bm{\sft}}_\inn^\T \bar{\mathbf{A}}_\D \QPA \bar{\bm{\sfR}}_\kn \QPA^\Her \bar{\mathbf{A}}_\D^\Her  \bar{\bm{\sft}}_\inn^* }_{\text{MUI}} + \underbrace{ \sigma^2 }_{\text{IID}} + \underbrace{ \tr \left( \left( \mathbf{F}_{\SC} \otimes \bm{\sfI}_{{\NBS}} \right) \QPA \mathbf{R}_{\bar{\pmb{\eta}}_\D} \QPA^\Her \left( \mathbf{F}_{\SC} \otimes \bm{\sfI}_{{\NBS}} \right)^\Her \bar{\bm{\sfR}}_\kn^\ast \right)}_{\text{QN}} } .
\label{eq:downlinkSINR_covar}
\end{split}
\end{equation}
The MUI term in the denominator of (\ref{eq:downlinkSINR_covar}) contains only the $(K-1)$ interfering terms of the $n^\thh$ sub-carrier whereas the quantization noise term has contributions from all $\SC$ sub-carriers.
The DL SQINR for the $k^\thh$ user on the $n^\thh$ sub-carrier in (\ref{eq:downlinkSINR_covar}) is thus a function of the beamformer matrix $\bar{\bm{\sfT}}$, the per-antenna power allocation matrix $\QPA$, and the power allocation vector $\bm{\sfq}$. 


\vspace{-0.3cm}
\subsection{Uplink SQINR for CEQ ADCs} \label{sec:UL}
Next we develop the system model for an OFDM based UL. Let $\bm{\sfp} = [\bm{\sfp}_1^\T, \dots \bm{\sfp}_{\SC}^\T]^\T$,  with $\bm{\sfp}_n = [\sfp_{1,n} , \dots , \sfp_{K,n}]^\T$ denoting the UL power allocation vector over the $n^\thh$ sub-carrier. With $\bm{\sfP}_n \triangleq \diag( \sqrt{\bm{\sfp}_n} ) \in \mathbb{C}^{K \times K}$, let $\bm{\sfx}_n = \bm{\sfP}_n \bm{\sfs}_n \in \mathbb{C}^{K}$ denote the signal on the $n^\thh$ sub-carrier at the $K$ users just prior to time domain conversion. Defining $\bm{\sfX}_\U = [\bm{\sfx}_1, \dots \bm{\sfx}_{\SC}] \in \mathbb{C}^{K \times \SC}$, the time domain signal matrix (for one OFDM symbol without the CP) after the IFFT operation is given by ${\mathbf{X}}_\U = [{\mathbf{x}}_1,\dots {\mathbf{x}}_{\SC}] = \bm{\sfX}_\U \mathbf{F}_{\SC}^\Her \in \mathbb{C}^{K \times \SC}$. The time domain signal received at the BS (after discarding the CP) denoted by ${\mathbf{V}}_\U = [{\mathbf{v}}_1, \dots {\mathbf{v}}_{\SC}] \in \mathbb{C}^{\NBS \times \SC}$ is perturbed by additive IID $\mathcal{N}( 0 , \sigma^2 )$ noise ${\mathbf{W}}_\U = [{\mathbf{w}}_1, \dots {\mathbf{w}}_{\SC}] \in \mathbb{C}^{\NBS \times \SC}$. The BS transforms the time domain signal into the frequency domain using the FFT operation denoted by $\bm{\sfV}_\U = [\bm{\sfv}_1,\dots \bm{\sfv}_{\SC}] = {\mathbf{V}}_\U \mathbf{F}_{\SC}$. The frequency domain noise $\bm{\sfW}_\U = {\mathbf{W}}_\U \mathbf{F}_{\SC}$ has the same statistics as ${\mathbf{W}}_\U$. At this point, the signals sent and received on the $n^\thh$ sub-carrier are related by $\bm{\sfv}_n = \bm{\sfH}_n \bm{\sfx}_n + \bm{\sfw}_n$ which follows from the well-known orthogonality of OFDM sub-carriers. The $K$ symbols on the $n^\thh$ sub-carrier are then resolved by the beamformer $\bm{\sfU}_n$ given by $\bm{\sfy}_n = \bm{\sfU}_n^\T \bm{\sfv}_n$. The received frequency domain symbols on all active sub-carriers can be collected in the matrix $\bm{\sfY}_\U = [\bm{\sfy}_1,\dots \bm{\sfy}_{\SC}]$. Define $\bar{\bm{\sfP}} = \bdiag(\bm{\sfP}_1,\dots \bm{\sfP}_{\SC})$, $\bar{\bm{\sfU}} = \bdiag(\bm{\sfU}_1,\dots \bm{\sfU}_{\SC})$, $\bar{\bm{\sfy}}_\U = \vecc(\bm{\sfY}_\U)$, $\bar{\bm{\sfv}}_\U = \vecc(\bm{\sfV}_\U)$, $\bar{{\mathbf{v}}}_\U = \vecc({\mathbf{V}}_\U)$ and $\bar{\bm{\sfw}}_\U = \vecc(\bm{\sfW}_\U)$. Using the vectorization identity from Section \ref{sec:DL}, $\bar{\bm{\sfy}}_\U$ is
\begin{equation} \label{eq:uplinkSig_r1}
\begin{split}
\bar{\bm{\sfy}}_\U &= \bar{\bm{\sfU}}^\T \bar{\bm{\sfv}}_\U = \bar{\bm{\sfU}}^\T \left( \mathbf{F}_{\SC} \otimes \bm{\sfI}_{{\NBS}} \right)  \underbrace{ \left( \mathbf{F}_{\SC}^{\Her} \otimes \bm{\sfI}_{{\NBS}} \right) \left( \bar{\bm{\sfH}}  \bar{\bm{\sfP}} \bar{\bm{\sfs}} + \bar{\bm{\sfw}}_\U \right) }_{\text{time domain signal } \bar{{\mathbf{v}}}_\U} = \bar{\bm{\sfU}}^\T \bar{\bm{\sfH}} \bar{\bm{\sfP}} \bar{\bm{\sfs}} + \bar{\bm{\sfU}}^\T \bar{\bm{\sfw}}_\U.
\end{split}
\end{equation}
Like its DL counterpart, it can be seen from the block-diagonal structure of (\ref{eq:uplinkSig_r1}) that each sub-carrier can be treated independently under the $\infty$-resolution assumption.



\sloppy Now we introduce CEQ quantization to the UL model and linearize it using the Bussgang decomposition. Using (\ref{eq:uplinkSig_r1}), the final symbols are given by $\bar{\bm{\sfy}}_\U = \bar{\bm{\sfU}}^\T \left( \mathbf{F}_{\SC} \otimes \bm{\sfI}_{{\NBS}} \right) \mathcal{Q}_b \left( \left( \mathbf{F}_{\SC}^\Her \otimes \bm{\sfI}_{{\NBS}} \right) \left( \bar{\bm{\sfH}}  \bar{\bm{\sfP}} \bar{\bm{\sfs}} + \bar{\bm{\sfw}}_\U \right) \right)$. The covariance matrix of the UL signal $\bar{{\mathbf{v}}}_\U$ before CEQ quantization is given by $\mathbf{R}_{\bar{{\mathbf{v}}}_\U} = \left( \mathbf{F}_{\SC}^\Her \otimes \bm{\sfI}_{{\NBS}} \right) \bar{\bm{\sfH}}  \bar{\bm{\sfP}}  \bar{\bm{\sfP}}^\Her \bar{\bm{\sfH}}^\Her \left( \mathbf{F}_{\SC}^\Her \otimes \bm{\sfI}_{{\NBS}} \right)^\Her + \sigma^2 \bm{\sfI}_{\SC \NBS}$. The Bussgang gain is defined as
\begin{equation} \label{eq:uplinkSig_r6}
\begin{split}
\bar{\mathbf{A}}_\U &= \zeta_b \diag \left( \mathbf{R}_{\bar{{\mathbf{v}}}_\U} \right)^{-\frac{1}{2}} \overset{(a)}{=} \bm{\sfI}_{\SC} \otimes \underbrace{\zeta_b \diag \left( \frac{1}{\SC} \sum_{n = 1}^{\SC} \left( \bm{\sfH}_n {\bm{\sfP}_n} \bm{\sfP}_n^\Her \bm{\sfH}_n^\Her + \sigma^2 \bm{\sfI}_{\NBS} \right) \right)^{-\frac{1}{2}}}_{\mathbf{A}_\U},
\end{split}
\end{equation}
where $(a)$ follows from the block-diagonal structure of the matrices involved in $\mathbf{R}_{\bar{{\mathbf{v}}}_\U}$.  Like its DL counterpart, the UL signal can be decomposed into a linear signal part and an uncorrelated distortion $\bar{\pmb{\eta}}_\U$ with covariance $ \mathbf{R}_{\bar{\pmb{\eta}}_\U}$ using the Bussgang decomposition. The linearized signal model of the symbols received at the BS during UL is given by
\begin{equation} \label{eq:uplinkSig_r7}
\begin{split}
\bar{\bm{\sfy}}_\U &= \bar{\bm{\sfU}}^\T \bar{\mathbf{A}}_\U \bar{\bm{\sfH}}  \bar{\bm{\sfP}} \bar{\bm{\sfs}} + \bar{\bm{\sfU}}^\T \bar{\mathbf{A}}_\U \bar{\bm{\sfw}}_\U + \bar{\bm{\sfU}}^\T \left( \mathbf{F}_{\SC} \otimes \bm{\sfI}_{{\NBS}} \right) \bar{\pmb{\eta}}_\U .
\end{split}
\end{equation}
Let $\bar{\bm{\sfu}}_\kn$ denotes the $((n-1)K + k)^\thh$ column of $\bar{\bm{\sfU}}$ and $\sfp_\kn$ denote the $((n-1)K + k)^\thh$ entry of $\bm{\sfp}$. The UL SQINR for the $k^\thh$ user at the $n^\thh$ sub-carrier, $\gamma_\kn^{\text{UL}}(\bar{\bm{\sfu}}_\kn , \bm{\sfp})$, is given by
\begin{equation}
\gamma_\kn^{\text{UL}}(\bar{\bm{\sfu}}_\kn , \bm{\sfp}) = \frac{ \sfp_\kn \bar{\bm{\sfu}}_\kn^\T \bar{\mathbf{A}}_\U  \bar{\bm{\sfR}}_\kn \bar{\mathbf{A}}_\U^\Her  \bar{\bm{\sfu}}_\kn^*}{ \bar{\bm{\sfu}}_\kn^\T \bigg(  \sum_{ \substack{i = 1\\ i \neq k }}^K \underbrace{ \sfp_\inn \bar{\mathbf{A}}_\U \bar{\bm{\sfR}}_\inn \bar{\mathbf{A}}_\U^\Her }_{\text{MUI}} + \underbrace{ \sigma^2 \bar{\mathbf{A}}_\U \bar{\mathbf{A}}_\U^\Her }_{\text{IID}} + \underbrace{ \left( \mathbf{F}_{\SC} \otimes \bm{\sfI}_{{\NBS}} \right) \mathbf{R}_{ \bar{\pmb{\eta}}_\U} \left( \mathbf{F}_{\SC} \otimes \bm{\sfI}_{{\NBS}} \right)^\Her }_{\text{QN}} \bigg) \bar{\bm{\sfu}}_\kn^*} .
\label{eq:uplinkSINR_covar}
\end{equation}
We note here that the UL SQINR in (\ref{eq:uplinkSINR_covar}) for the $k^\thh$-user depends only on the power allocation vector $\bm{\sfp}$ and the combiner for the $k^\thh$-user $\bar{\bm{\sfu}}_\kn$. This is in contrast to the DL SQINR in (\ref{eq:downlinkSINR_covar}) which depends on the power allocation vector $\bm{\sfq}$ and the beamformer matrix $\bar{\bm{\sfT}}$ of all $K$ users. We will make use of this observation in Section \ref{sec:joint} to replace the MU-MIMO-OFDM DL problem by its equivalent UL counterpart by making use of the UL-DL duality proved in Section~\ref{sec:duality}.

\vspace{-0.5cm}
\subsection{Problem formulation}\label{sec:problem}
In this paper, we maximize the minimum of the achieved DL SQINR to target SQINR ratio of all users and sub-carriers over the choice of the BF matrix $\bar{\bm{\sfT}}$ and power allocation vector $\bm{\sfq}$. This optimization criterion has not been considered before for MU-MIMO-OFDM DL precoding under CEQ hardware constraints. This formulation bridges the performance gap between linear and non-linear methods and provides more flexibility in terms of the ability to allocate individual user targets. With $\gamma_\kn$ denoting the target SQINR for the $k^\thh$ user on the $n^\thh$ sub-carrier, the MU-MIMO-OFDM DL precoding problem with individual SQINR constraints is given by
\begin{equation} \label{eq:DL_SINR}
\begin{aligned}
\CDL_\text{opt}(P_{\BS}) &= \max_{\bar{\bm{\sfT}},\bm{\sfq}} \min_{\substack{ {1\leq k \leq K} \\ {1 \leq n \leq \SC}}}  \frac{\gamma_\kn^{\text{DL}}(\bar{\bm{\sfT}} , \QPA , \bm{\sfq})}{\gamma_\kn} \\
& \textrm{s.t. }  \|\bm{\sfq}\|_1 \leq P_{\BS} \SC\\
& ||\bm{\sft}_\kn||_2 = 1, \quad 1\leq k \leq K, 1 \leq n \leq \SC. \\
\end{aligned}
\end{equation}
Dropping the maximization over $\bar{\bm{\sfT}}$ in (\ref{eq:DL_SINR}) results in the power allocation problem where the minimum achieved to target SQINR ratio has to be maximized only over all admissible power allocation vectors for a fixed DL-BF matrix $\bar{\bm{\sfT}}^\star$ 
\begin{equation} \label{eq:DL_SINR2}
\begin{aligned}
\CDL_\text{opt}(P_{\BS}, \bar{\bm{\sfT}}^\star) &= \max_{\bm{\sfq}}\min_{\substack{ {1\leq k \leq K} \\ {1 \leq n \leq \SC}}} \frac{\gamma_\kn^{\text{DL}}(\bar{\bm{\sfT}}^\star , \QPA , \bm{\sfq})}{\gamma_\kn} \\
& \textrm{s.t. } \|\bm{\sfq}\|_1 \leq P_{\BS} \SC.\\
\end{aligned}
\end{equation}
The corresponding MU-MIMO-OFDM UL problems are given in the same manner as (\ref{eq:DL_SINR}) and (\ref{eq:DL_SINR2})  with $\gamma_\kn^{\text{DL}}$ replaced by $\gamma_\kn^{\text{UL}}$. The solution to these problems for $\infty$-resolution converters are obtained by exploiting the UL-DL duality principle to cast the MU-DL-BF problem in terms of the easier-to-solve MU-UL-BF problem \cite{MUDL}. The UL-DL duality principle does not hold for the system with CEQ constraints due to the quantization noise. The quantization noise matrices $ \mathbf{R}_{\bar{\pmb{\eta}}_\D}$ in (\ref{eq:downlinkSINR_covar}) and $ \mathbf{R}_{\bar{\pmb{\eta}}_\U}$ in (\ref{eq:uplinkSINR_covar}) depend on the DL beamforming matrix and channel realization respectively and prevent a straightforward extension of the UL-DL duality principle. We show in Section \ref{sec:duality} that the UL-DL duality principle can be extended to CEQ constraints and an OFDM signal model under certain conditions and then use that result to find the solution to (\ref{eq:DL_SINR}). Towards that end, we next introduce an approximation that will later be used for proving the duality principle.

\vspace{-0.6cm}
\subsection{Small angle approximation} \label{sec:uqn}
With $\bar{{\mathbf{x}}}_\D$ and $\bar{{\mathbf{z}}}_\D$ denoting the DL time-domain signal before and after CEQ, the covariance matrix of the uncorrelated distortion $\bar{\pmb{\eta}}_\text{d} = \bar{{\mathbf{z}}}_\D - \bar{\mathbf{A}}_\D  \bar{{\mathbf{x}}}_\D$ resulting from the Bussgang~decomposition~is 
\begin{equation} \label{eq:upn2}
\mathbf{R}_{\bar{\pmb{\eta}}_\D} =  \mathbf{R}_{ \bar{{\mathbf{z}}}_\D }  - \bar{\mathbf{A}}_\D  \mathbf{R}_{ \bar{{\mathbf{x}}}_\D } \bar{\mathbf{A}}_\D^\Her.
\end{equation}
Using the definition of the Bussgang gain $\zeta_b$ from Section \ref{sec:DL}, define $\hat{\mathbf{X}}_{\D} =  \R \left( \frac{1}{\zeta_b^2} \bar{\mathbf{A}}_\D  \mathbf{R}_{ \bar{{\mathbf{x}}}_\D } \bar{\mathbf{A}}_\D^\Her \right)$ and $\hat{\mathbf{Y}}_{\D} =  \I \left( \frac{1}{\zeta_b^2} \bar{\mathbf{A}}_\D  \mathbf{R}_{ \bar{{\mathbf{x}}}_\D } \bar{\mathbf{A}}_\D^\Her \right)$. The diagonal entries of $\hat{\mathbf{X}}_{\D}$ and $\hat{\mathbf{Y}}_{\D}$ are equal to 1 and 0. It is known \cite{statistical} that the correlation matrix of the $b$-bit CEQ signal $\bar{{\mathbf{z}}}_\D$  is given by

\begin{equation} \label{eq:upn3}
\mathbf{R}_{ \bar{{\mathbf{z}}}_\D } = \frac{2^b}{\pi} \text{sin}^2\left( \frac{\pi}{2^b} \right) \sum_{\Delta b = 0}^{2^{b-1}-1} e^{\ji(2\Delta b \pi / 2^b)} \text{sin}^{-1}\left( \R \left( \left( \hat{\mathbf{X}}_{\D}  + \ji \hat{\mathbf{Y}}_{\D} \right) e^{-\ji(2\Delta b \pi / 2^b)}  \right) \right),
\end{equation}
where the $\text{sin}^{-1}$ operation is applied element-wise on its matrix argument. For $b = \infty$, 
\begin{equation} \label{eq:upn4}
\mathbf{R}_{ \bar{{\mathbf{z}}}_\D }  = \frac{1}{2} \int_0^\pi e^{\ji \phi} \text{sin}^{-1}\left( \R \left( \left( \hat{\mathbf{X}}_{\D}  + \ji \hat{\mathbf{Y}}_{\D} \right) e^{-\ji \phi}  \right) \right) \D \phi.
\end{equation}
It can be verified from (\ref{eq:upn3}) and (\ref{eq:upn4}) that $\diag\left( \mathbf{R}_{ \bar{{\mathbf{z}}}_\D } \right) = \mathbf{I}_{\SC \NBS}$ for all values of $b$. This also makes intuitive sense since the diagonal corresponds to the variance of $b$-bit CEQ entries with unit norm. We now approximate the non-linear $\text{sin}^{-1}(\cdot)$ function using the first-order Taylor expansion $\text{sin}^{-1}(x) = x + o(x^3)$ for $\mathbf{R}_{ \bar{{\mathbf{z}}}_\D }^\ndiag$. The approximation is justified because the off-diagonal entries of $\mathbf{R}_{ \bar{{\mathbf{x}}}_\D }$ will be forced to be small by adding optimized dithering to the signal before quantization as will be explained in Section \ref{sec:dummy}. Under this approximation, $\mathbf{R}_{ \bar{{\mathbf{z}}}_\D }^\ndiag$ for $b \neq \infty$ is
\begin{equation} \label{eq:upn5}
\mathbf{R}_{ \bar{{\mathbf{z}}}_\D }^\ndiag = \underbrace{ \frac{2^b}{\pi} \text{sin}^2\left( \frac{\pi}{2^b} \right) \sum_{\Delta b = 0}^{2^{b-1}-1} \text{cos}^2\left(2\Delta b \pi / 2^b\right) }_{ \bar{\zeta_b} } \hat{\mathbf{X}}_{\D}^\ndiag +   \ji  \underbrace{  \frac{2^b}{\pi} \text{sin}^2\left( \frac{\pi}{2^b} \right) \sum_{\Delta b = 0}^{2^{b-1}-1} \text{sin}^2\left(2\Delta b \pi / 2^b\right) }_{ \bar{\zeta_b} } \hat{\mathbf{Y}}_{\D}^\ndiag.
\end{equation}
For $b = \infty$, $\bar{\zeta}_b = \frac{\pi}{4}$ using (\ref{eq:upn4}). The off-diagonal entries of  $\mathbf{R}_{\bar{\pmb{\eta}}_\D}$ are then given by
\begin{equation} \label{eq:upn5}
\begin{aligned}
\mathbf{R}_{\bar{\pmb{\eta}}_\D}^\ndiag = \bar{\zeta_b} \left( \hat{\mathbf{X}}_{\D}^\ndiag + \ji \hat{\mathbf{Y}}_{\D}^\ndiag \right) - \zeta_b^2 \left( \hat{\mathbf{X}}_{\D}^\ndiag + \ji \hat{\mathbf{Y}}_{\D}^\ndiag \right). \\
\end{aligned}
\end{equation}
It can be verified from the definition of $\bar{\zeta_b}$ and $\zeta_b^2$ (defined in Section \ref{sec:DL}) that they are equal for all values of $b$. Hence, all off-diagonal entries of $\mathbf{R}_{\bar{\pmb{\eta}}_\D}$, i.e. $\mathbf{R}_{\bar{\pmb{\eta}}_\D}^\ndiag$ in (\ref{eq:upn5}), are zero. We are left with the diagonal part of the matrix and hence $\mathbf{R}_{\bar{\pmb{\eta}}_\D} = \left( 1 - \zeta_b^2 \right)\mathbf{I}$. The $\text{sin}^{-1}(x) \approx x$ approximation thus makes the quantization noise $\bar{\pmb{\eta}}_\D$ uncorrelated. The same approximation can also be applied to the UL quantization noise $\bar{\pmb{\eta}}_\U$. It will be shown in Section \ref{sec:duality} that the resulting uncorrelated quantization noise is crucial for proving UL DL duality under CEQ ADC and DAC constraints. 



\section{UL-DL duality with hardware constraints} \label{sec:duality}

In this section, we generalize the UL-DL duality principle to the MU-MIMO-OFDM system with CEQ DACs/ADCs. We show that the same SQINR constraints can be achieved in both DL and UL by appropriately relating the linear beamforming and combining matrices and separately optimizing the DL/UL power allocation vectors under the same sum power constraint. This result is summarized in Theorem \ref{theorem:theorem1}.

\begin{theorem} \label{theorem:theorem1}
Consider a BS equipped with CEQ DACs communicating with $K$ users over $\SC$ sub-carriers with target SQINR values $[\gamma_{1,1} \dots \gamma_\kn \dots \gamma_{K,\SC} ]$ using the beamforming matrix $\bar{\bm{\sfT}}$, DL power allocation vector $\bm{\sfq}$ and per-antenna power allocation matrix $\QPA = \bm{\sfI}_{\SC} \otimes { \diag \left( \frac{1}{\SC} \sum_{n = 1}^{\SC} \bm{\sfT}_n {\bm{\sfQ}_n} \bm{\sfQ}_n^\Her \bm{\sfT}_n^\Her \right)^{\frac{1}{2}}}$. The same set of SQINR values can be achieved in the UL under CEQ ADC constraints by letting $\bar{\bm{\sft}}_\kn = \bar{\mathbf{A}}_\text{u} \bar{\bm{\sfu}}_\kn / \| \bar{\mathbf{A}}_\text{u} \bar{\bm{\sfu}}_\kn \|_2$ and $\bm{\sfp} =  \frac{ \sigma^2}{\zeta_b^{2}} \left(  \bm{\sfI}_{K\SC} -  \bar{\bm{\sfD}}(\bar{\bm{\sfT}})  \bar{\bm{\Psi}}^\T(\bar{\bm{\sfT}}) - \bar{\bm{\sfD}}(\bar{\bm{\sfT}})  \bar{\bm{\Phi}}^\T(\bar{\bm{\sfT}}) \right)^{-1} \bar{\bm{\sfD}}(\bar{\bm{\sfT}}) \bm{1}_{K\SC}$ for the same sum power constraint in the UL stage as the total BS transmit power in the DL stage.
\end{theorem}

\begin{proof}
The proof closely follows the proof of \cite[Theorem 3.1]{pw2} for frequency flat channels under 1-bit hardware constraints. We present a brief version of that proof in Sections \ref{sec:dualityDL} and \ref{sec:dualityUL} for completeness after making the changes required for CEQ DACs and the OFDM signal model. Our goal is to show that the same SQINR values can be achieved in the DL and UL for the same total power by simplifying the linearized DL and UL SQINRs from (\ref{eq:downlinkSINR_covar}) and (\ref{eq:uplinkSINR_covar}) using the small angle approximation.
\end{proof}

\vspace{-0.5cm}
\subsection{Downlink SQINR} \label{sec:dualityDL}
The choice of $\QPA = \bm{\sfI}_{\SC} \otimes { \diag \left( \frac{1}{\SC} \sum_{n = 1}^{\SC} \bm{\sfT}_n {\bm{\sfQ}_n} \bm{\sfQ}_n^\Her \bm{\sfT}_n^\Her \right)^{\frac{1}{2}}}$ makes the per-antenna power allocation after the quantization operation equal to the $\infty-$resolution DAC setting. We equate the target DL SQINR for user $k$ at the $n^\thh$ sub-carrier, $\gamma_\kn$, to the achieved DL SQINR $\gamma_\kn^{\text{DL}}(\bar{\bm{\sfT}} ,  \QPA , \bm{\sfq})$ (under the small angle approximation) from (\ref{eq:downlinkSINR_covar}) and choose the DL power allocation vector $\bm{\sfq}$ that achieves these target SQINRs as
\begin{equation}
\gamma_\kn = \frac{ \sfq_\kn \bar{\bm{\sft}}_\kn^\T \bar{\bm{\sfR}}_\kn \bar{\bm{\sft}}_\kn^*}{  \sum_{ \substack{i = 1\\ i \neq k}}^K \sfq_\inn \bar{\bm{\sft}}_\inn^\T  \bar{\bm{\sfR}}_\kn \bar{\bm{\sft}}_\inn^* + \frac{1}{\zeta_b^2} \sigma^2 + \left( \frac{1}{\zeta_b^2} - 1 \right) \tr \left( \left( \bm{\sfI}_{\SC} \otimes \diag \left( \frac{1}{\SC} \sum_{n = 1}^{\SC} \bm{\sfT}_n {\bm{\sfQ}_n} \bm{\sfQ}_n^\Her \bm{\sfT}_n^\Her \right) \right) \bar{\bm{\sfR}}_\kn^\ast \right)}.
\label{eq:downlinkSINR_covar1}
\vspace{-0.2cm}
\end{equation}
We note here that all the variables in (\ref{eq:downlinkSINR_covar1}) with a $(\bar{\cdot})$ have a block diagonal structure, with only the $n^\thh$-block making a non-zero contribution to any multiplications/additions involving these variables. For example, only the $n^\thh$ block on the diagonal of $\bar{\bm{\sfR}}_\kn$ (${\bm{\sfR}}_\kn = {\bm{\sfh}}_\kn {\bm{\sfh}}_\kn^\Her \in \mathbb{C}^{\NBS \times \NBS}$) makes a non-zero contribution to any of the terms involving it. Using this observation and the matrix identities $\tr \left( \bm{\sfA} \diag(\bm{\sfB}) \right) = \tr \left( \bm{\sfB} \diag(\bm{\sfA}) \right)$ and $\tr \left( \bm{\sfA\sfB\sfC} \right) = \tr \left( \bm{\sfB\sfC\sfA} \right) = \tr \left( \bm{\sfC\sfA\sfB} \right)$, the $K \SC$ equalities in (\ref{eq:downlinkSINR_covar1}) are simplified to
\begin{equation}
\gamma_\kn = \frac{ \sfq_\kn {\bm{\sft}}_\kn^\T {\bm{\sfR}}_\kn {\bm{\sft}}_\kn^*}{  \sum_{ \substack{i = 1\\ i \neq k}}^K \sfq_\inn {\bm{\sft}}_\inn^\T  {\bm{\sfR}}_\kn {\bm{\sft}}_\inn^* + \frac{1}{\zeta_b^2} \sigma^2 + \left( \frac{1}{\zeta_b^2} - 1 \right) \tr \left(  \frac{1}{\SC}  \sum_{i = 1}^{K} \sum_{j = 1}^{\SC}  \sfq_\ijj \bm{\sft}_\ijj^\T \diag \left( {\bm{\sfR}}_\kn^\ast \right)  \bm{\sft}_\ijj^* \right)}.
\label{eq:downlinkSINR_covar2}
\vspace{-0.3cm}
\end{equation}
The DL SQINR formulation $\gamma_\kn^{\text{DL}}(\bar{\bm{\sfT}} ,  \QPA , \bm{\sfq})$ in Section \ref{sec:DL} given by (\ref{eq:downlinkSINR_covar}) is equivalently given by the simplified expression $\gamma_\kn^{\text{DL}}(\bar{\bm{\sfT}} , \bm{\sfq})$ on the right hand side (RHS) of (\ref{eq:downlinkSINR_covar2}) under the small angle approximation and $\QPA = \bm{\sfI}_{\SC} \otimes { \diag \left( \frac{1}{\SC} \sum_{n = 1}^{\SC} \bm{\sfT}_n {\bm{\sfQ}_n} \bm{\sfQ}_n^\Her \bm{\sfT}_n^\Her \right)^{\frac{1}{2}}}$. Next, we define the $K \SC \times K \SC$ diagonal SQINR~matrix $\bar{\bm{\sfD}}(\bar{\bm{\sfT}}) = \bdiag\left(\bm{\sfD}_1(\bm{\sfT}_1),\dots \bm{\sfD}_{\SC}(\bm{\sfT}_{\SC})\right) $, where $\bm{\sfD}_n(\bm{\sfT}_n) = \diag \left( \frac{\gamma_{1,n}}{( \bm{\sft}_{1,n}^\T {\bm{\sfR}}_{1,n} {\bm{\sft}}_{1,n}^*)}, \dots \frac{\gamma_{K,n}}{( \bm{\sft}_{K,n}^\T {\bm{\sfR}}_{K,n} {\bm{\sft}}_{K,n}^*)} \right)$.
We also define the $K \times K$ MUI coupling matrix for the $n^\thh$ sub-carrier, $\bm{\Psi}_n(\bm{\sfT}_n)$, as
\begin{equation}
  \bm{\Psi}_n(\bm{\sfT}_n) = \begin{bmatrix}
     0 & {\bm{\sft}}_{2,n}^\T {\bm{\sfR}}_{1,n}  {\bm{\sft}}_{2,n}^*  & \dots & {\bm{\sft}}_{K,n}^\T {\bm{\sfR}}_{1,n}  {\bm{\sft}}_{K,n}^* \\ 
     
    {\bm{\sft}}_{1,n}^\T {\bm{\sfR}}_{2,n} {\bm{\sft}}_{1,n}^* & 0  & \dots & {\bm{\sft}}_{K,n}^\T {\bm{\sfR}}_{2,n} {\bm{\sft}}_{K,n}^*\\ 
    
    \vdots & \ddots & \ddots &\vdots\\ 
    
     {\bm{\sft}}_{1,n}^\T {\bm{\sfR}}_{K,n} {\bm{\sft}}_{1,n}^* & {\bm{\sft}}_{2,n}^\T {\bm{\sfR}}_{K,n} {\bm{\sft}}_{2,n}^*  & \dots & 0
      \end{bmatrix}.
  \label{eq:PsiMatrix2}
\end{equation}
The $K \SC \times K \SC$ MUI coupling matrix for all sub-carriers is defined as $\bar{\bm{\Psi}}(\bar{\bm{\sfT}}) = \bdiag\left(\bm{\Psi}_1(\bm{\sfT}_1),\dots \bm{\Psi}_{\SC}(\bm{\sfT}_{\SC})\right)$.
With $ \tilde{\bm{\sfR}}_\kn = \frac{1}{\SC}\left( \frac{1}{\zeta_b^2} - 1 \right) \text{diag} \left(\bm{\sfR}_\kn \right)$, we also define the $K \SC \times K \SC$ quantization coupling matrix $\bar{\bm{\Phi}}(\bar{\bm{\sfT}})$ as
\begin{equation}
 \bar{\bm{\Phi}}(\bar{\bm{\sfT}}) = \begin{bmatrix}
     {\bm{\sft}}_{1,1}^\T \tilde{\bm{\sfR}}_{1,1}  {\bm{\sft}}_{1,1}^*  & \dots & {\bm{\sft}}_{K,1}^\T \tilde{\bm{\sfR}}_{1,1}  {\bm{\sft}}_{K,1}^* &   {\bm{\sft}}_{1,2}^\T \tilde{\bm{\sfR}}_{1,1}  {\bm{\sft}}_{1,2}^* & \dots &  {\bm{\sft}}_{K,\SC}^\T \tilde{\bm{\sfR}}_{1,1}  {\bm{\sft}}_{K,\SC}^* \\ 
     
    {\bm{\sft}}_{1,1}^\T \tilde{\bm{\sfR}}_{2,1} {\bm{\sft}}_{1,1}^*  & \dots & {\bm{\sft}}_{K,1}^\T \tilde{\bm{\sfR}}_{2,1} {\bm{\sft}}_{K,1}^* & {\bm{\sft}}_{1,2}^\T \tilde{\bm{\sfR}}_{2,1}  {\bm{\sft}}_{1,2}^* & \dots &  {\bm{\sft}}_{K,\SC}^\T \tilde{\bm{\sfR}}_{2,1}  {\bm{\sft}}_{K,\SC}^*\\ 
    
    \vdots & \vdots & \vdots &\vdots &\vdots &\vdots\\ 
    
      {\bm{\sft}}_{1,1}^\T \tilde{\bm{\sfR}}_{K,\SC} {\bm{\sft}}_{1,1}^*  & \dots & {\bm{\sft}}_{K,1}^\T \tilde{\bm{\sfR}}_{K,\SC} {\bm{\sft}}_{K,1}^* & {\bm{\sft}}_{1,2}^\T \tilde{\bm{\sfR}}_{K,\SC}  {\bm{\sft}}_{1,2}^* & \dots &  {\bm{\sft}}_{K,\SC}^\T \tilde{\bm{\sfR}}_{K,\SC}  {\bm{\sft}}_{K,\SC}^*\\ 
      \end{bmatrix}.
  \label{eq:PhiMatrix}
\end{equation}
Using $\bar{\bm{\sfD}}(\bar{\bm{\sfT}})$, $\bar{\bm{\Psi}}(\bar{\bm{\sfT}})$ and $\bar{\bm{\Phi}}(\bar{\bm{\sfT}})$, the $K\SC$ equations in (\ref{eq:downlinkSINR_covar2}) can be written as
\begin{equation} \label{eq:downlinkSINR_covar3}
\bm{\sfq} = \bar{\bm{\sfD}}(\bar{\bm{\sfT}})  \bar{\bm{\Psi}}(\bar{\bm{\sfT}}) \bm{\sfq} + \bar{\bm{\sfD}}(\bar{\bm{\sfT}})  \bar{\bm{\Phi}}(\bar{\bm{\sfT}})\bm{\sfq}  + \frac{\sigma^2}{\zeta_b^2} \bar{\bm{\sfD}}(\bar{\bm{\sfT}}) \bm{1}_{K\SC}. \end{equation}
Using (\ref{eq:downlinkSINR_covar3}), the DL power allocation vector $\bm{\sfq}$ can be written as
\begin{equation} \label{eq:downlinkSINR_covar4}
\bm{\sfq} =  \frac{ \sigma^2}{\zeta_b^2} \left(  \bm{\sfI}_{K\SC} -  \bar{\bm{\sfD}}(\bar{\bm{\sfT}})  \bar{\bm{\Psi}}(\bar{\bm{\sfT}}) - \bar{\bm{\sfD}}(\bar{\bm{\sfT}})  \bar{\bm{\Phi}}(\bar{\bm{\sfT}}) \right)^{-1} \bar{\bm{\sfD}}(\bar{\bm{\sfT}}) \bm{1}_{K\SC}. \end{equation}
Lemma \ref{lemma:lemma1} shows that the matrix inverse in (\ref{eq:downlinkSINR_covar4}) exists for any \emph{feasible} target SQINR set $\{ \gamma_\kn\}$.

\begin{lemma} \label{lemma:lemma1}
For any \emph{feasible} target DL SQINR set $\{ \gamma_\kn\}$, the matrix $\left(  \bm{\sfI}_{K\SC} -  \bar{\bm{\sfD}}(\bar{\bm{\sfT}})  \bar{\bm{\Psi}}(\bar{\bm{\sfT}}) - \bar{\bm{\sfD}}(\bar{\bm{\sfT}})  \bar{\bm{\Phi}}(\bar{\bm{\sfT}}) \right)$ is invertible.
\end{lemma}

\begin{proof}
The proof follows from the proof of \cite[Lemma 3.2 and Lemma 3.3]{pw2} by replacing ${\bm{\sfD}}({\bm{\sfT}})  {\bm{\Psi}}({\bm{\sfT}})$ with $\bar{\bm{\sfD}}(\bar{\bm{\sfT}})  \bar{\bm{\Psi}}(\bar{\bm{\sfT}}) + \bar{\bm{\sfD}}(\bar{\bm{\sfT}})  \bar{\bm{\Phi}}(\bar{\bm{\sfT}})$.
\end{proof}
\\ The choice of DL power allocation vector $\bm{\sfq}$ in (\ref{eq:downlinkSINR_covar4}) thus achieves the target SQINRs $\{ \gamma_\kn\}$ for the given beamforming matrix $\bar{\bm{\sfT}}$ .

\vspace{-0.3cm}
\subsection{Uplink SQINR}  \label{sec:dualityUL}
Now we show that the same target SQINRs $\{\gamma_\kn\}$ can be achieved in the UL by choosing the UL power allocation vector $\bm{\sfp}$ under the sum power constraint equal to the BS DL power. Next, observe that $\bar{\bm{\sfu}}_\kn^\T \bm{\sfI} \bar{\bm{\sfu}}_\kn^*$ can be replaced by $\bar{\bm{\sfu}}_\kn^\T \bar{\mathbf{A}}_\text{u} \bar{\mathbf{A}}_\text{u}^{-2}  \bar{\mathbf{A}}_\text{u} \bar{\bm{\sfu}}_\kn^*$ in (\ref{eq:uplinkSINR_covar}) with $\bar{\mathbf{A}}_\text{u}$ defined in (\ref{eq:uplinkSig_r6}). Relating the DL beamformers and UL combiners as $\bar{\bm{\sft}}_\kn = \bar{\mathbf{A}}_\text{u} \bar{\bm{\sfu}}_\kn / \| \bar{\mathbf{A}}_\text{u} \bar{\bm{\sfu}}_\kn \|_2$ with $\bm{\sft}_\kn^\T \bm{\sft}_\kn^* = 1$ and recalling the observation made in Section \ref{sec:dualityDL} about non-zero contribution from only the $n^\thh$ block of all block-diagonal variables, the  $K\SC$ target SQINRs can be equated to the achieved UL SQINR (\ref{eq:uplinkSINR_covar}) under the small angle approximation as
\begin{equation} \label{eq:uplinkSINR_covar4}
\gamma_\kn = \frac{\sfp_\kn \bm{\sft}_\kn^\T \bm{\sfR}_\kn \bm{\sft}_\kn^*}{  \sum_{ \substack{i = 1\\ i \neq k}}^K \sfp_\inn \bm{\sft}_\kn^\T \bm{\sfR}_\inn \bm{\sft}_\kn^* + \frac{1}{\zeta_b^2} \sigma^2 + \left( \frac{1}{\zeta_b^2} - 1 \right) \tr \left( \frac{1}{\SC} \sum_{i = 1}^K  \sum_{j = 1}^{\SC} \sfp_\ijj \bm{\sft}_\kn^\T \diag( \bm{\sfR}_\ijj ) \bm{\sft}_\kn^* \right) }.
\end{equation}
The UL SQINR $\gamma_\kn^{\text{UL}}(\bar{\bm{\sfu}}_\kn , \bm{\sfp})$ in (\ref{eq:uplinkSINR_covar}) is equivalently given by the expression $\gamma_\kn^{\text{UL}}(\bm{\sft}_\kn ,  \bm{\sfp})$ on the right hand side (RHS) of (\ref{eq:uplinkSINR_covar4}). Using $\bar{\bm{\sfD}}(\bar{\bm{\sfT}})$, $\bar{\bm{\Psi}}(\bar{\bm{\sfT}})$ and $\bar{\bm{\Phi}}(\bar{\bm{\sfT}})$, the $K \SC$ equations in (\ref{eq:uplinkSINR_covar4}) can be rearranged in matrix form as
\begin{equation} \label{eq:uplinkSINR_covar5}
\bm{\sfp} = \bar{\bm{D}}(\bar{\bm{\sfT}})  \bar{\bm{\Psi}}^\T(\bar{\bm{\sfT}}) \bm{\sfp} + \bar{\bm{\sfD}}(\bar{\bm{\sfT}})  \bar{\bm{\Phi}}^\T(\bar{\bm{\sfT}})\bm{\sfp}  + \frac{\sigma^2}{\zeta_b^2} \bar{\bm{\sfD}}(\bar{\bm{\sfT}}) \bm{1}_{K\SC}.
\end{equation}
The UL power allocation vector $\bm{\sfp}$ is given by
\begin{equation} \label{eq:uplinkSINR_covar6}
\bm{\sfp} =  \frac{\sigma^2}{\zeta_b^2} \left(  \bm{\sfI}_{K\SC} -  \bar{\bm{\sfD}}(\bar{\bm{\sfT}})  \bar{\bm{\Psi}}^\T(\bar{\bm{\sfT}}) - \bar{\bm{\sfD}}(\bar{\bm{\sfT}})  \bar{\bm{\Phi}}^\T(\bar{\bm{\sfT}}) \right)^{-1} \bar{\bm{\sfD}}(\bar{\bm{\sfT}}) \bm{1}_{K\SC}. 
\end{equation}
The existence of the matrix inverse in (\ref{eq:uplinkSINR_covar6}) follows from Lemma \ref{lemma:lemma1}. Ignoring the scalar factor $\frac{\sigma^2}{\zeta_b^2}$, the total UL power allocation is given by
\begin{equation} \label{eq:duality}
\begin{split}
\|\bm{\sfp}\|_1 &=   \bm{1}_{K\SC}^\T \left(  \bm{\sfI}_{K\SC} -  \bar{\bm{\sfD}}(\bar{\bm{\sfT}})  \bar{\bm{\Psi}}^\T(\bar{\bm{\sfT}}) - \bar{\bm{\sfD}}(\bar{\bm{\sfT}})  \bar{\bm{\Phi}}^\T(\bar{\bm{\sfT}}) \right)^{-1} \bar{\bm{\sfD}}(\bar{\bm{\sfT}}) \bm{1}_{K\SC} \\
& \overset{(a)}{=} \left( \left(  \bm{\sfI}_{K\SC} -  \bar{\bm{\sfD}}(\bar{\bm{\sfT}})  \bar{\bm{\Psi}}(\bar{\bm{\sfT}}) - \bar{\bm{\sfD}}(\bar{\bm{\sfT}})  \bar{\bm{\Phi}}(\bar{\bm{\sfT}}) \right)^{-1} \bar{\bm{\sfD}}(\bar{\bm{\sfT}}) \bm{1}_{K\SC} \right)^\T  \bm{1}_{K\SC} = \|\bm{\sfq}\|_1,
\end{split}
\end{equation}
where (a) makes use of the push-through identity, the diagonal structure of $\bar{\bm{\sfD}}(\bar{\bm{\sfT}})$ and $(\bm{\sfA\sfB})^\T = \bm{\sfB}^\T \bm{\sfA}^\T$. The same target SQINRs can be obtained in the UL and DL for equal sum power which establishes the UL-DL duality principle for MU-MIMO-OFDM systems with CEQ ADCs/DACs.

\section{UL-DL duality based  proposed solution} \label{sec:solution}
Based on the UL-DL duality result for MU-MIMO-OFDM with CEQ constraints established in Section \ref{sec:duality}, we extend the alternating minimization solution from \cite{pw2} to the max-min optimization problem presented in Section \ref{sec:problem}. Due to the high-dimensional nature of this solution which simultaneously deals with all active sub-carriers, we present a simplified solution that can be independently applied to each sub-carrier analogous to the flat fading case presented in \cite{pw2}.  We also comment on the convergence of the proposed algorithm and briefly justify the small angle approximation presented in Section \ref{sec:uqn}.

\vspace{-0.3cm}
\subsection{Optimal DL power allocation} \label{sec:power}
With the vectorized definitions of the DL and UL SQINR in (\ref{eq:downlinkSINR_covar2}) and (\ref{eq:uplinkSINR_covar4}), the solution to the optimal DL power problem (\ref{eq:DL_SINR2}) (and its UL analog) closely follows the solution for the flat fading case given in Section IV-B of \cite{pw2} with a few changes defined next. We define the extended DL power allocation vector $\bm{\sfq}^\star_\text{ ext} = \left[ \bm{\sfq}^\star \quad 1 \right]^\T$ and the positive extended DL coupling~matrix
\begin{equation}
  \bar{\bm{\Upsilon}}(\bar{\bm{\sfT}}^\star , P_{\BS}) = \begin{bmatrix}
    \bar{\bm{\sfD}}(\bar{\bm{\sfT}}^\star) \left( \bar{\bm{\Psi}}(\bar{\bm{\sfT}}^\star) + \bar{\bm{\Phi}}(\bar{\bm{\sfT}}^\star) \right) &  \frac{\sigma^2}{\zeta_b^2} \bar{\bm{\sfD}}( \bar{\bm{\sfT}}^\star) \bm{1}_{K\SC}  \\ 
        
   \frac{\bm{1}_{K\SC}^\T}{P_{\BS} \SC}  \bar{\bm{\sfD}}(\bar{\bm{\sfT}}^\star) \left( \bar{\bm{\Psi}}(\bar{\bm{\sfT}}^\star) + \bar{\bm{\Phi}}(\bar{\bm{\sfT}}^\star) \right) &  \frac{\sigma^2}{\zeta_b^2 P_{\BS} \SC} \bm{1}_{K\SC}^\T   \bar{\bm{\sfD}}( \bar{\bm{\sfT}}^\star) \bm{1}_{K\SC}
      \end{bmatrix}.
  \label{eq:psiMatrix}
\end{equation}
Following the development in \cite{pw2}, the solution to the DL power allocation problem (\ref{eq:DL_SINR2}) is 
\begin{equation} \label{eq:psiMatrix3}
\CDL_\text{opt}(P_{\BS} , \bar{\bm{\sfT}}^\star) = \frac{1}{ \lambda_\maxx \left( \bar{\bm{\Upsilon}}(\bar{\bm{\sfT}}^\star , P_{\BS}) \right) }.
\end{equation}
And the optimal DL power allocation vector $\bm{\sfq}^\star$ is given by the first $K\SC$ entries of the dominant eigenvector of $\bar{\bm{\Upsilon}}(\bar{\bm{\sfT}}^\star , P_{\BS})$ scaled such that the last entry equals 1. Similarly, the maximizer $\bm{\sfp}^\star$ of the UL version of the problem (\ref{eq:DL_SINR2}) is given by the first $K\SC$ entries of the dominant eigenvector $\bm{\sfp}^\star_\ext = \left[ \bm{\sfp}^\star \quad 1 \right]^\T$ (last entry scaled to 1) of the positive UL extended coupling matrix~defined~as 
\begin{equation}
  \bar{\bm{\Lambda}}(\bar{\bm{\sfT}}^\star , P_{\BS}) = \begin{bmatrix}
    \bar{\bm{\sfD}}(\bar{\bm{\sfT}}^\star) \left( \bar{\bm{\Psi}}(\bar{\bm{\sfT}}^\star) + \bar{\bm{\Phi}}(\bar{\bm{\sfT}}^\star) \right)^\T &  \frac{\sigma^2}{\zeta_b^2} \bar{\bm{\sfD}}( \bar{\bm{\sfT}}^\star) \bm{1}_{K\SC}  \\ 
        
   \frac{\bm{1}_{K\SC}^\T}{P_{\BS} \SC}  \bar{\bm{\sfD}}(\bar{\bm{\sfT}}^\star) \left( \bar{\bm{\Psi}}(\bar{\bm{\sfT}}^\star) + \bar{\bm{\Phi}}(\bar{\bm{\sfT}}^\star) \right)^\T &  \frac{ \sigma^2}{\zeta_b^2 P_{\BS} \SC} \bm{1}_{K\SC}^\T   \bar{\bm{\sfD}}( \bar{\bm{\sfT}}^\star) \bm{1}_{K\SC}
      \end{bmatrix}.
  \label{eq:gammaMatrix}
\end{equation}
And the optimal achieved to target SQINR ratio of the UL power allocation problem is given~by
\begin{equation} \label{eq:gammaMatrix2}
\CUL_\text{opt}(P_{\BS} , \bar{\bm{\sfT}}^\star) = \frac{1}{ \lambda_\maxx \left( \bar{\bm{\Lambda}}(\bar{\bm{\sfT}}^\star , P_{\BS}) \right) }.
\end{equation}
The solutions of the DL and UL power allocation problems are equal, i.e.\ $\CUL_\text{opt}(P_{\BS} , \bar{\bm{\sfT}}^\star) = \CDL_\text{opt}(P_{\BS} , \bar{\bm{\sfT}}^\star)$. This follows directly from the duality result in Theorem \ref{theorem:theorem1}. Hence, the same target SQINR set $\{\gamma_\kn\}$ (or a scalar multiple thereof) is achieved in both DL and UL using this power allocation procedure. We will use this observation~to~cast~the MU-MIMO-OFDM DL problem in terms of the corresponding UL problem for a more efficient~solution.

\vspace{-0.3cm}
\subsection{Joint power allocation and precoder design} \label{sec:joint} 
We now focus on the joint beamforming matrix and power allocation problem in (\ref{eq:DL_SINR}). The optimum solution to the power allocation problem (\ref{eq:DL_SINR2}) for a fixed BF matrix $\bar{\bm{\sfT}}^\star$ is given by the reciprocal of the dominant eigenvalue of the extended DL coupling matrix $\bar{\bm{\Upsilon}}(\bar{\bm{\sfT}}^\star , P_{\BS})$. The joint power and beamforming optimization problem (\ref{eq:DL_SINR}) can be equivalently stated as
\begin{equation} \label{eq:DL_SINR6}
\CDL_\text{opt}(P_{\BS}) = \frac{1}{\min_{\bar{\bm{\sfT}}} \lambda_{\max}\left( \bar{\bm{\Upsilon}}(\bar{\bm{\sfT}} , P_{\BS}) \right) }.
\end{equation} 
Making use of the duality result from Theorem \ref{theorem:theorem1}, we replace (the motivation of doing this will become clear later) the DL extended coupling matrix with the UL extended coupling matrix
\begin{equation} \label{eq:DL_SINR7}
\CDL_\text{opt}(P_{\BS}) = \frac{1}{ \min_{\bar{\bm{\sfT}}} \lambda_{\max}\left( \bar{\bm{\Lambda}}(\bar{\bm{\sfT}} , P_{\BS}) \right) }.
\end{equation} 
By the Perron-Frobenius theorem \cite{MUDL}, $\lambda_{\max}$ of the non-negative matrix $\bar{\bm{\Lambda}}(\bar{\bm{\sfT}} , P_{\BS})$ is
\begin{equation} \label{eq:DL_SINR8}
\lambda_{\max}\left( \bar{\bm{\Lambda}}(\bar{\bm{\sfT}} , P_{\BS}) \right) =   \max_{\bm{\sfx}>0} \min_{\bm{\sfy}>0} \frac{ \bm{\sfx}^\T \bar{\bm{\Lambda}}(\bar{\bm{\sfT}} , P_{\BS}) \bm{\sfy}  }{ \bm{\sfx}^\T \bm{\sfy} } .
\end{equation} 
Next, we define an intermediate cost function
\begin{equation} \label{eq:DL_SINR9}
 \bar{\lambda} \left( \bar{\bm{\sfT}}, P_{\BS}, \bm{\sfp}_\ext \right) = \max_{\bm{\sfx}>0}  \frac{ \bm{\sfx}^\T \bar{\bm{\Lambda}}(\bar{\bm{\sfT}} , P_{\BS}) \bm{\sfp}_\ext  }{ \bm{\sfx}^\T \bm{\sfp}_\ext },
\end{equation} 
which allows us to rewrite (\ref{eq:DL_SINR7}) as
\begin{equation} \label{eq:DL_SINR10}
\left( \CDL_\text{opt}(P_{\BS}) \right)^{-1} = \min_{\bar{\bm{\sfT}}} \min_{\bm{\sfp}_\ext > 0} \bar{\lambda} \left( \bar{\bm{\sfT}}, P_{\BS}, \bm{\sfp}_\ext \right).
\end{equation} 
Similar to \cite{pw2,MUDL}, we take a two step alternating minimization approach to solve (\ref{eq:DL_SINR10}). In the first step, we solve for the UL power allocation vector $\bm{\sfp}^\star$ for a fixed beamforming matrix. In the second step $\bm{\sfp}$ is held fixed while solving for $\bar{\bm{\sfT}}^\star$.

\subsubsection{Power allocation step}
Since the solution to the power allocation problem (\ref{eq:DL_SINR2}) was obtained by maximizing the minimum, it follows from (\ref{eq:gammaMatrix2}) that the optimal UL power allocation vector $\bm{\sfp}^\star_\ext$ minimizes the function $\bar{\lambda} \left( \bar{\bm{\sfT}}, P_{\BS}, \bm{\sfp}_\ext \right)$ for a fixed beamforming matrix $\bar{\bm{\sfT}}$. 



\subsubsection{Beamformer optimization step}
Next we fix the power allocation vector $\bm{\sfp}$ $\left(\text{with } \bm{\sfp}_\ext = \left[ \bm{\sfp} \quad 1 \right]^\T \right)$ and optimize the beamforming matrix $\bar{\bm{\sfT}}$ given by the problem
\begin{equation} \label{eq:P1}
\bar{\bm{\sfT}}^\star = \argmin_{\bar{\bm{\sfT}}} \bar{\lambda} \left( \bar{\bm{\sfT}}, P_{\BS}, \bm{\sfp}_\ext \right).
\end{equation}
Lemma \ref{lemma:lemma5} (follows from \cite[Lemma 4.3]{pw2}) helps break down (\ref{eq:P1}) into smaller decoupled sub-problems. 
\begin{lemma} \label{lemma:lemma5}
The cost function $\bar{\lambda} \left( \bar{\bm{\sfT}}, P_{\BS}, \bm{\sfp}_\ext \right)$ can equivalently be written as
\begin{equation}\label{eq:P2}
\max_{\bm{\sfx}>0}  \frac{ \bm{\sfx}^\T \bar{\bm{\Lambda}}(\bar{\bm{\sfT}}, P_{\BS}) \bm{\sfp}_\ext  }{ \bm{\sfx}^\T \bm{\sfp}_\ext } = \max_{ \substack{ {1 \leq k \leq K}\\ {1 \leq n \leq \SC} }} \frac{\gamma_\kn}{\gamma_\kn^{\text{UL}}\left(  \bm{\sft}_\kn ,  \bm{\sfp} \right) }.
\end{equation}
\end{lemma}
Lemma \ref{lemma:lemma5}  transforms the minimization over $\bar{\bm{\sfT}}$ into maximization over the UL SQINRs. Now we recall from Section \ref{sec:UL} that the UL SQINR for the $k^\thh$ user on the $n^\thh$ sub-carrier, $\gamma_\kn^{\text{UL}}\left(  \bm{\sft}_\kn ,  \bm{\sfp} \right)$, is a function of only the beamforming vector $\bm{\sft}_\kn$. Hence each individual SQINR term can be maximized independently of others and this is the primary reason for recasting the problem from DL to UL by making use of the UL-DL duality.
This is summarized in Corollary \ref{cor:cor2}.

\begin{corollary} \label{cor:cor2}
The solution to the problem (\ref{eq:P1}) is given by independent maximization of the $K\SC$ UL SQINRs $\gamma_\kn^{\text{UL}}\left(  \bm{\sft}_\kn ,  \bm{\sfp} \right)$ for $1 \leq k \leq K$ and $1 \leq n \leq \SC$.
\end{corollary}
Using the definition of the UL SQINR (\ref{eq:uplinkSINR_covar4}), the beamformer $\bm{\sft}_\kn^\star$ maximizing $\gamma_\kn^{\text{UL}}\left(  \bm{\sft}_\kn ,  \bm{\sfp} \right)$ is
\begin{equation} \label{eq:P3}
\bm{\sft}_\kn^\star = \argmax_{\bm{\sft}_\kn} \frac{\sfp_\kn \bm{\sft}_\kn^\T \bm{\sfR}_\kn \bm{\sft}_\kn^*}{\bm{\sft}_\kn^\T \bm{\sfS}_\kn(\bm{\sfp}) \bm{\sft}_\kn^*}, \quad  \textrm{s.t.} \| \bm{\sft}_\kn \|_2 = 1,
\end{equation} 
where
\begin{equation} \label{eq:DL_SINR20}
\bm{\sfS}_\kn( \bm{\sfp} ) =  \sum_{ \substack{i = 1\\ i \neq k}}^K \sfp_\inn \bm{\sfR}_\inn + \left( \frac{1}{\zeta_b^2} - 1 \right) \frac{1}{\SC} \sum_{i=1}^K \sum_{j=1}^{\SC} \sfp_\ijj \diag(\bm{\sfR}_\ijj)  + \frac{1}{\zeta_b^2} \sigma^2 \bm{\sfI}.
\end{equation} 
Since the matrices $\bm{\sfR}_\kn$ and $\bm{\sfS}_\kn$ are hermitian, the solution to (\ref{eq:P3}) is given by the dominant generalized eigenvector of the matrix pair $( \bm{\sfR}_\kn , \bm{\sfS}_\kn )$ for $1 \leq k \leq K$ and $1 \leq n \leq \SC$ \cite{MUDL}. 

The alternating minimization algorithm iterates between the power allocation step and BF optimization step till $ \lambda^{(t-1)}_{\max}\left( \bar{\bm{\Lambda}}( \bar{\bm{\sfT}}^{(t-1)}, P_{\BS}) \right) -  \lambda^{(t)}_{\max}\left( \bar{\bm{\Lambda}}(\bar{\bm{\sfT}}^{(t)}, P_{\BS}) \right) < \epsilon$. The superscript $(\cdot)^{(t)}$ denotes the iteration index and $\epsilon$ is a predefined constant used to stop the optimization procedure. After convergence, the DL power allocation vector $\bm{\sfq}^\star$ is calculated using the precoder matrix $\bm{\sfT}^\star$ obtained in the final iteration. The proposed solution is summarized in Algorithm~\ref{alg:alg1}. $\bm{\sfv}_{\max}(\cdot , \cdot)$ denotes the dominant generalized eigenvector of its arguments and $\lambda^{(t)}_{\max} \triangleq \lambda^{(t)}_{\max}\left( \bar{\bm{\Lambda}}(\bar{\bm{\sfT}}^{(t)}, P_{\BS}) \right)$.

\vspace{-0cm}
\begin{algorithm}[h]
\caption{Alternating minimization solution to (\ref{eq:DL_SINR}) }\label{alg:alg1}
 1) \text{Initialize:} $t = 0, \bm{\sfp}^{\star(0)} = \bm{0}_{K\SC} , P_{\BS} \SC , \epsilon$ \vspace{.1cm}  \\
 2) \textbf{while} $\lambda_{\max}^{(t-1)} - \lambda_{\max}^{(t)} \geq \epsilon$ \vspace{.05cm}    \\
 3) \hspace{0.5cm}  $\forall_\kn \text{ } {\bm{\sft}}_\kn^{\star(t)} = \bm{\sfv}_{\max}\left( \bm{\sfR}_\kn,  \bm{\sfS}_\kn(\bm{\sfp}^{\star(t-1)})\right)$ \vspace{0cm}    \\
 4) \hspace{0.5cm}  $\forall_\kn \text{ } {\bm{\sft}}_\kn^{\star(t)} = \bm{\sft}_\kn^{\star(t)} / \| \bm{\sft}_\kn^{\star(t)} \|_2 $ \hspace{1.5cm}  \vspace{0cm}    \\
 5) \hspace{0.5cm}  $\bar{\bm{\Lambda}}(\bar{\bm{\sfT}}^{\star(t)}, P_{\BS}) \bm{\sfp}_\text{ext}^{\star(t)} =  \lambda_{\max}^{(t)} \bm{\sfp}_\text{ext}^{\star(t)}$  \hspace{0.1cm} \vspace{0cm}\\
 6) \hspace{0.5cm}  $\bm{\sfp}^{\star(t)} = \bm{\sfp}_\text{ext}^{\star(t)}[1,\dots,K\SC] / \bm{\sfp}_\text{ext}^{\star(t)}[K\SC+1]$ \hspace{0.4cm}\vspace{.1cm}\\
8) \textbf{end} \vspace{.1cm}\\
9) $\bar{\bm{\Upsilon}}(\bar{\bm{\sfT}}^{\star(t)}, P_{\BS}) \bm{\sfq}_\text{ext}^{\star} =  \lambda_{\max}^{(t)} \bm{\sfq}_\text{ext}^{\star} $   \hspace{0.8cm} \vspace{.05cm} \\
10) $\bm{\sfq}^{\star(t)} = \bm{\sfq}_\text{ext}^{\star(t)}[1,\dots,K\SC] / \bm{\sfq}_\text{ext}^{\star(t)}[K\SC+1]$ \hspace{0.5cm}
\end{algorithm}
\vspace{-.3cm}

\vspace{-0.3cm}
\subsection{Convergence} \label{sec:convergence} 
Like the various methods in existing literature \cite{StuderSQUID,MSM1,MSM2,MSM3,magic,magic2,access}, our proposed method does not guarantee global optimality of the solution given by Algorithm \ref{alg:alg1} to the MU-MIMO-OFDM DL precoding problem which is NP-hard under CEQ constraints. Our simulations  indicate that the proposed algorithm typically converges within $2-5$ iterations. We now show that the proposed algorithm is indeed convergent  to some point in the solution space. We know from the precoder optimization step that $\bar{\bm{\sfT}}^{\star(t)}$ minimizes the cost function  $\bar{\lambda} \left( \bar{\bm{\sfT}}, P_{\BS},  \bm{\sfp}_\text{ext}^{\star(t-1)} \right)$ in the $t^\thh$~iteration
\begin{equation} \label{eq:conv1}
\bar{\lambda} \left( \bar{\bm{\sfT}}^{\star(t)}, P_{\BS},  \bm{\sfp}_\text{ext}^{\star(t-1)} \right) \leq \bar{\lambda} \left( \bar{\bm{\sfT}}^{\star(t-1)}, P_{\BS},  \bm{\sfp}_\text{ext}^{\star(t-1)} \right) = \lambda_{\max}^{(t-1)}.
\end{equation}
From the definition of  $ \lambda_{\max}^{(t)}$, we know that
\begin{equation} \label{eq:conv2}
\begin{split}
\lambda_{\max}^{(t)} &= \max_{\bm{\sfx}>0} \min_{\bm{\sfy}>0} \frac{ \bm{\sfx}^\T \bar{\bm{\Lambda}}(\bar{\bm{\sfT}}^{\star(t)}, P_{\BS}) \bm{\sfy}  }{ \bm{\sfx}^\T \bm{\sfy} } \leq \max_{\bm{\sfx}>0} \frac{ \bm{\sfx}^\T \bar{\bm{\Lambda}}(\bar{\bm{\sfT}}^{\star(t)}, P_{\BS}) \bm{\sfp}_\text{ext}^{\star(t-1)}  }{ \bm{\sfx}^\T \bm{\sfp}_\text{ext}^{\star(n)} } = \bar{\lambda} \left( \bar{\bm{\sfT}}^{\star(t)}, P_{\BS},  \bm{\sfp}_\text{ext}^{\star(t-1)} \right).
\end{split}
\end{equation}
It can be observed from (\ref{eq:conv1}) and (\ref{eq:conv2}) that the sequence $\lambda_{\max}^{(t)}$ is monotonically decreasing. Combining this behavior with the non-negativity of $\lambda_{\max}^{(t)}$ implies the existence of a limiting value $\lambda_{\max}^{(\infty)}$. The parameter $\epsilon$ controls how far the algorithm stops from this point in the solution~space.
 

\vspace{-0.3cm}
\subsection{Sub-carrier-wise algorithm} \label{sec:joint2} 
One big advantage of OFDM is the orthogonality of the sub-carriers which can each be treated independently of others. As seen in Section \ref{sec:systemModel}, this orthogonality is destroyed by the non-linear CEQ operation. Under the assumptions introduced at the beginning of Section \ref{sec:duality}, this loss of orthogonality shows up in the linearized SQINR expressions in the form of the quantization coupling matrix $\bar{\bm{\Phi}}(\bar{\bm{\sfT}})$ defined in (\ref{eq:PhiMatrix}). This is the only matrix which does not have a block-diagonal structure and couples all sub-carrier together by making the distortion vector depend on the beamformers of all $\SC$ sub-carrriers. This causes all matrices involved in Algorithm~\ref{alg:alg1} to have dimensions of $K\SC \times K\SC$ making it computationally expensive. We would like to approximately maintain the orthogonality of the OFDM sub-carriers and come up with an algorithm that can be applied to each sub-carrier individually without sacrificing the performance (shown in Section \ref{sec:results}). Towards this end, we look closely at the structure of the matrix $\bar{\bm{\Phi}}(\bar{\bm{\sfT}})$ and make an IID assumption about the involved variables that restores the sub-carrier orthogonality.

We want the quantization coupling matrix $\bar{\bm{\Phi}}(\bar{\bm{\sfT}})$ to have a block-diagonal structure similar to the MUI coupling matrix $\bar{\bm{\Psi}}(\bar{\bm{\sfT}})$. This will allow us to break all $K\SC \times K\SC$ matrices in Algorithm~\ref{alg:alg1} into $\SC$ $K \times K$ sub-matrices which can then be treated in a parallel manner. It can be seen that each entry of $\bar{\bm{\Phi}}(\bar{\bm{\sfT}})$ is a positively weighted Euclidean norm squared of one of the beamforming vectors $\bm{\sft}_\kn$ for $1 \leq k \leq K$ and $1 \leq n \leq \SC$. We know that $\| \bm{\sft}_\kn \|_2^2 = 1 \text{ } \forall_\kn$. The weighting is through the matrices $\tilde{\bm{\sfR}}_\kn = \frac{1}{\SC}\left( \frac{1}{\zeta_b^2} - 1 \right) \text{diag} \left(\bm{\sfR}_\kn \right)$. At this point we make the assumption that each entry of the beamforming vector $\bm{\sft}_\kn$ comes from some IID distribution with second moment equal to $\alpha^2$. Similarly each entry of the channel vector $\bm{\sfh}_\kn$ comes from some IID distribution with second moment $\beta^2$. Ignoring the scalar constant $\left( \frac{1}{\zeta_b^2} - 1 \right)$, each entry of the matrix $\bar{\bm{\Phi}}(\bar{\bm{\sfT}})$ equals $\frac{\NBS}{\SC}\alpha^2 \beta^2$ in expectation under the IID assumption. We now make the entries which are off each $K \times K$ block on the main diagonal equal to zero and scale up the entries of each $K \times K$ block by a factor of $\SC$. Now $\bar{\bm{\Phi}}(\bar{\bm{\sfT}})$ has a block-diagonal structure similar to $\bar{\bm{\Psi}}(\bar{\bm{\sfT}})$. The quantization coupling matrix for the $n^\thh$ sub-carrier, ${\bm{\Phi}}_n({\bm{\sfT}}_n)$, is defined in the same way as the MUI coupling matrix ${\bm{\Psi}}_n({\bm{\sfT}}_n)$ in (\ref{eq:PsiMatrix2}) with $\bm{\sfR}_\kn$ replaced by $\tilde{\bm{\sfR}}_\kn = \left( \frac{1}{\zeta_b^2} - 1 \right) \text{diag} \left(\bm{\sfR}_\kn \right)$. With this modification, the per sub-carrier algorithm is given in Algorithm~\ref{alg:alg2}. Other minor differences in Algorithm~\ref{alg:alg2} are listed below.
\begin{itemize}
\item Since Algorithm~\ref{alg:alg2} runs on each sub-carrier independently, the total BS DL power $P_\BS \SC$ has to be divided among the $\SC$ sub-carriers. We divide the power equally among the sub-carriers enforced by setting the power budget for each sub-carrier equal to $P_{\BS}$.
\item The extended DL and UL coupling matrices, ${\bm{\Upsilon}}({\bm{\sfT}}_n, {P}_{\BS})$ and ${\bm{\Lambda}}({\bm{\sfT}}_n, {P}_{\BS})$, are now defined separately for each sub-carrier using the $n^\thh$ SQINR matrix ${\bm{\sfD}}_n({\bm{\sfT}}_n)$, the $n^\thh$ MUI coupling matrix ${\bm{\Psi}}_n({\bm{\sfT}}_n)$ and the $n^\thh$ quantization coupling matrix ${\bm{\Phi}}_n({\bm{\sfT}}_n)$ by appropriately modifying their definitions in (\ref{eq:psiMatrix}) and (\ref{eq:gammaMatrix}) to not include the factor $\SC$.
\item The matrix $\bm{\sfS}_\kn$ is now defined as
\begin{equation} \label{eq:DL_SINR21}
\bm{\sfS}_\kn( \bm{\sfp}_n ) =  \sum_{ \substack{i = 1\\ i \neq k}}^K \sfp_\inn \bm{\sfR}_\inn + \left( \frac{1}{\zeta_b^2} - 1 \right) \sum_{i=1}^K \sfp_\inn \diag(\bm{\sfR}_\inn)  + \frac{1}{\zeta_b^2} \sigma^2 \bm{\sfI}.
\end{equation} 
\end{itemize}
We point out here that the IID assumption used to simplify Algorithm \ref{alg:alg1} is not mathematically accurate. Using (\ref{eq:P3}) and (\ref{eq:DL_SINR20}), it can be seen that the precoder $\bm{\sft}^\star_\kn$ does depend on $\bm{\sfR}_\ijj \forall_{i,j}$. For $j \neq n$, this dependence is weak because of the $1/\SC$ scaling. Our numerical results in Section \ref{sec:results} demonstrate that this assumption is fairly accurate and does not affect the performance.

 \vspace{-0.4cm}
\begin{algorithm}[b]
\caption{Per sub-carrier solution to (\ref{eq:DL_SINR}) }\label{alg:alg2}
 1) \textbf{for} $n = 1:\SC$ \vspace{.05cm}  (parallelizable)   \\
 2) \hspace{0.3cm} \text{Initialize:} $t = 0, \bm{\sfp}_n^{\star(0)} = \bm{0}_{K} , P_{\BS} , \epsilon$ \vspace{.05cm}  \\
 3) \hspace{0.3cm} \textbf{while} $\lambda_{\max}^{(t-1)} - \lambda_{\max}^{(t)} \geq \epsilon$ \vspace{.05cm}    \\
 4) \hspace{0.6cm}  $\forall_k \text{ } {\bm{\sft}}_\kn^{\star(t)} = \bm{\sfv}_{\max}\left( \bm{\sfR}_\kn,  \bm{\sfS}_\kn(\bm{\sfp}_n^{\star(t-1)})\right)$   \\
 5) \hspace{0.6cm}  $\forall_k \text{ } {\bm{\sft}}_\kn^{\star(t)} = \bm{\sft}_\kn^{\star(t)} / \| \bm{\sft}_\kn^{\star(t)} \|_2 $ \\
 6) \hspace{0.6cm}  ${\bm{\Lambda}}({\bm{\sfT}}_n^{\star(t)}, {P}_{\BS}) \bm{\sfp}_{n_\text{ext}}^{\star(t)} =  \lambda_{\max}^{(t)} \bm{\sfp}_{n_\text{ext}}^{\star(t)}$  \\
 7) \hspace{0.6cm}  $\bm{\sfp}_n^{\star(t)} = \bm{\sfp}_{n_\text{ext}}^{\star(t)}[1,\dots,K] / \bm{\sfp}_{n_\text{ext}}^{\star(t)}[K+1]$  \vspace{.05cm}\\
 8) \hspace{0.3cm} \textbf{end while} \vspace{0cm}\\
 9)  \hspace{0.3cm} ${\bm{\Upsilon}}({\bm{\sfT}}_n^{\star(t)}, {P}_{\BS}) \bm{\sfq}_{n_\text{ext}}^{\star} =  \lambda_{\max}^{(t)} \bm{\sfq}_{n_\text{ext}}^{\star} $  \vspace{.05cm} \\
10) \hspace{0.1cm} $\bm{\sfq}_n^{\star(t)} = \bm{\sfq}_{n_\text{ext}}^{\star(t)}[1,\dots,K] / \bm{\sfq}_{n_\text{ext}}^{\star(t)}[K+1]$\\
 11) \textbf{end for}
\end{algorithm}
 \vspace{-0cm}

\vspace{-0.3cm}
\subsection{Optimized dithering by dummy users} \label{sec:dummy} 
The UL-DL duality proof in Section \ref{sec:duality} relied on the quantization noise being uncorrelated resulting from the small angle approximation introduced in Section \ref{sec:uqn}. The small angle approximation is accurate when the off-diagonal elements of the covariance matrix of the signal before quantization are small compared to the diagonal entries. This is not true when the number of users is small (in DL) or the per-antenna SQINR is high (in UL). This makes the quantization noise correlated which in turn limits the achievable SQINR due to constructive interference. We ensure that this approximation remains true under all conditions by adding optimized dithering to the system. Dithering is introduced in the form of dummy users operating in the null space of the true users with their own individual SQINR constraints. The amount of dithering added is proportional to the power allocated to the dummy users which depends on their target SQINRs. We convert this into a scalar optimization problem by forcing the SQINR constraint of all dummy users on all sub-carriers to be the same. This problem is solved using a simple line search method by starting off from a small value for the dummy user SQINR constraint and then increasing it in small increments till the minimum of all true user SQINRs is increasing. We do not focus on this aspect of the problem in the results presented in this paper. We refer the reader to \cite{pw2} for a more detailed description of the benefits of adding dummy users~to~the~system.

\section{Results and discussion} \label{sec:results}
In this section, we present numerical results for the proposed algorithm and compare it with existing linear and non-linear precoding methods in terms of the achievable rate and coded BER.

\vspace{-0.3cm}
\subsection{Simulation setup}
We consider a setting where individual users are uniformly distributed (IID across realizations) in a $120^\circ$ sector around a BS (located at the origin 25 m above the ground) from a minimum distance of 50 m to a maximum distance of 150 m. We draw channel realizations from the 3GPP Urban-Macro line-of-sight and non-line-of-sight (3GPP 38.901 UMa LoS and 3GPP 38.901 UMa NLoS) channel models  implemented in Quadriga \cite{Quadriga}. For some of our results, we limit the minimum spatial separation (in degrees) between users. Quadriga implements underlying random variables in way that they are correlated across space and time. Limiting the separation between users hence limits the correlation between their channel coefficients. The channel coefficients and delays for each channel realization are converted to a complex baseband channel with $L = 8$ taps by sampling from a truncated sinc pulse. We empirically verified that $L = 8$ captures the delay spread of the channel realizations for our considered bandwidth. Next, we obtain the frequency domain channel for each sub-carrier from the complex baseband channel by taking its $\SC$-point FFT.  The important simulation parameters (unless stated otherwise) are given in Table \ref{table:table2}.

\begin {table}[h]
\begin{center}
\begin{tabular}{ | >{\centering\arraybackslash} m{3.6cm} | >{\centering\arraybackslash} m{4.8cm} | }
  \hline 			
  {Quadriga channel model} & {$\text{3GPP 38.901 UMa LoS / NLoS}$} \\  			
  \hline 
  { Number of antennas ${\BS}$} & 32 \\  			
  \hline 
  { $\text{Antenna element pattern}$} & $\text{0 dBi omni-directional}$ \\  			
  \hline 
  { Total transmit power $P_{\BS}$} & 40 dBm \\  			
  \hline 
  {Carrier frequency $f_c$} & 60 GHz \\  			
  \hline 
  {Bandwidth $B$} & 100 MHz (LoS) / 20 MHz (NLoS) \\  			
  \hline 
  {Number of subcarriers $\SC$} & 32 \\  			
  \hline 
  {Cyclic prefix length $\CP$} & 8 \\  			
  \hline 
  {SQINR constraint $\{\gamma_\kn\}$} & \{ 3 dB \} \\  			
  \hline 
  {CEQ resolution b } & \{ 2 , 3 , $\infty$\} \\  			
  \hline 
\end{tabular}
\end{center}
\vspace{-0.3cm}
\caption{Important simulation parameters.}
\vspace{-1.1cm}
\label{table:table2} 
\end{table}

\subsection{Benchmark strategies}
We compare with ZF precoding \cite{StuderOFDM} as a benchmark for our proposed technique with the DL precoding matrix given by $\bar{\bm{\sfT}} =  \bar{\bm{\sfH}}^\Her (\bar{\bm{\sfH}} \bar{\bm{{\sfH}}}^\Her)^{-1}$. We choose two ways to allocate the per-antenna power allocation matrix $\bar{\mathbf{Q}}_\PA$ in the CEQ system given by 

\begin{itemize}

\item \emph{ZF Opt-Pwr}: The power allocation vector $\bm{\sfq}$ is obtained using the optimal DL power allocation procedure described in \ref{sec:power} for $\bar{\bm{\sfT}}^\star = [ \bar{ \bm{\sft} }_{1,1} \dots \bar{ \bm{\sft} }_{K,\SC} ]$ with $\bar{ \bm{\sft} }_\kn= { \bm{\sft} }_\kn / \| { \bm{\sft} }_\kn \|_2$. The per-antenna power allocation is then given by $\QPA = \bm{\sfI}_{\SC} \otimes { \diag \left( \frac{1}{\SC} \sum_{n = 1}^{\SC} \bm{\sfT}_n^\star {\bm{\sfQ}_n} \bm{\sfQ}_n^\Her \bm{\sfT}_n^{\star\Her} \right)^{\frac{1}{2}}}$. 


\item \emph{ZF Equal-Pwr}: The BS divides the total transmit power equally across all the antennas and all samples within one OFDM symbol. The per-antenna power allocation is given by $\bar{\mathbf{Q}}_\PA = \diag \left( {\frac{P_{\BS}}{\NBS}} \bm{1}_{\NBS\SC} \right)^{\frac{1}{2}}$. This has been considered before in existing literature \cite{amodh,StuderOFDM}.
\end{itemize}
For a fair comparison, we also add dithering to ZF precoding by projecting Gaussian noise with variance  $\sigma_d^2$ onto the null space of the channel matrix $\bar{\bm{\sfH}}$ and adding it to the DL signal before the CEQ operation \cite{amodh}. Appropriate value for $\sigma_d^2$ is found using a simple line search method similar to what is described for the dummy users in Section \ref{sec:dummy}. We also compare with the unconstrained fully digital ZF and regularized ZF precoders (denoted by UnQZF and UnQRZF) to establish a baseline for all low-resolution algorithms.
 
As mentioned in Section \ref{sec:intro}, all non-linear methods \cite{StuderSQUID,MSM1,MSM2,MSM3,magic,magic2,access} perform roughly the same. We choose SQUID \cite{StuderSQUID} and MAGIQ \cite{magic2} as the representative non-linear methods to compare with the proposed solution. SQUID is based on Douglas-Rachford splitting of a squared $\ell_\infty$-norm relaxation of the symbol MMSE problem. MAGIQ is based on a coordinate wise minimization of the time domain MSE. We refer the reader to \cite{StuderSQUID} and \cite{magic2} for a more detailed description of the algorithms. The hyperparameters involved in the implementation of SQUID were chosen according to the guidelines given in \cite{StuderSQUID}. In the results that follow, the proposed Algorithm \ref{alg:alg1} and its per sub-carrier version Algorithm \ref{alg:alg2} are denoted as `Max-min' and `Max-min SC'.  `Max-min SC Equal-Pwr' denotes Algorithm \ref{alg:alg2} with equal per-antenna power allocation. 
 

\vspace{-0.3cm}
\subsection{SQINR results}\label{sec:SQINR}
We use the ergodic sum rate given by $\mathbb{E}\left[ \frac{1}{\SC} \sum_{k=1}^K \sum_{n=1}^{\SC} \text{log}_2 ( 1 + \gamma_\kn^{\text{DL}} ) \right]$ and the ergodic minimum rate given by $\mathbb{E}[ \min_{1\leq k \leq K}  \frac{1}{\SC} \sum_{n=1}^{\SC} \text{log}_2 ( 1 + \gamma_\kn^{\text{DL}} ) ]$ as the metrics of choice for our results. The expectation is computed by averaging across IID channel realizations each of which corresponds to an IID user location realization. $\gamma_\kn^{\text{DL}}$ is calculated using (\ref{eq:downlinkSINR_covar}) with the exact \emph{arcsine} law and without the small angle approximation. It should also be pointed out that our definitions of the sum and minimum rate do not account for the loss due to the CP. 

The ergodic sum rate is shown in Fig.~\ref{fig:meanCap} as a function of the number of active users. It can be observed that the proposed strategy and ZF (with optimal power allocation) perform similarly when the number of users is small. All variants of the proposed solution, however, perform increasingly better than ZF with the number of active users. For $K = 14$, the performance in terms of the sum rate differs by about 6-7 b/s/Hz depending on the CEQ resolution. Another important takeaway from Fig.~\ref{fig:meanCap} is that the per sub-carrier version of the proposed algorithm (Max-min SC) achieves the same performance as that of Algorithm \ref{alg:alg1} thus justifying the assumption made in Section \ref{sec:joint2}. Lastly, the performance difference between $b=3$ and $b=\infty$ is not significant for both ZF and the proposed solution.

\begin{figure}[t]
    	\begin{center}
    		\includegraphics[width=.5\textwidth,clip,keepaspectratio]{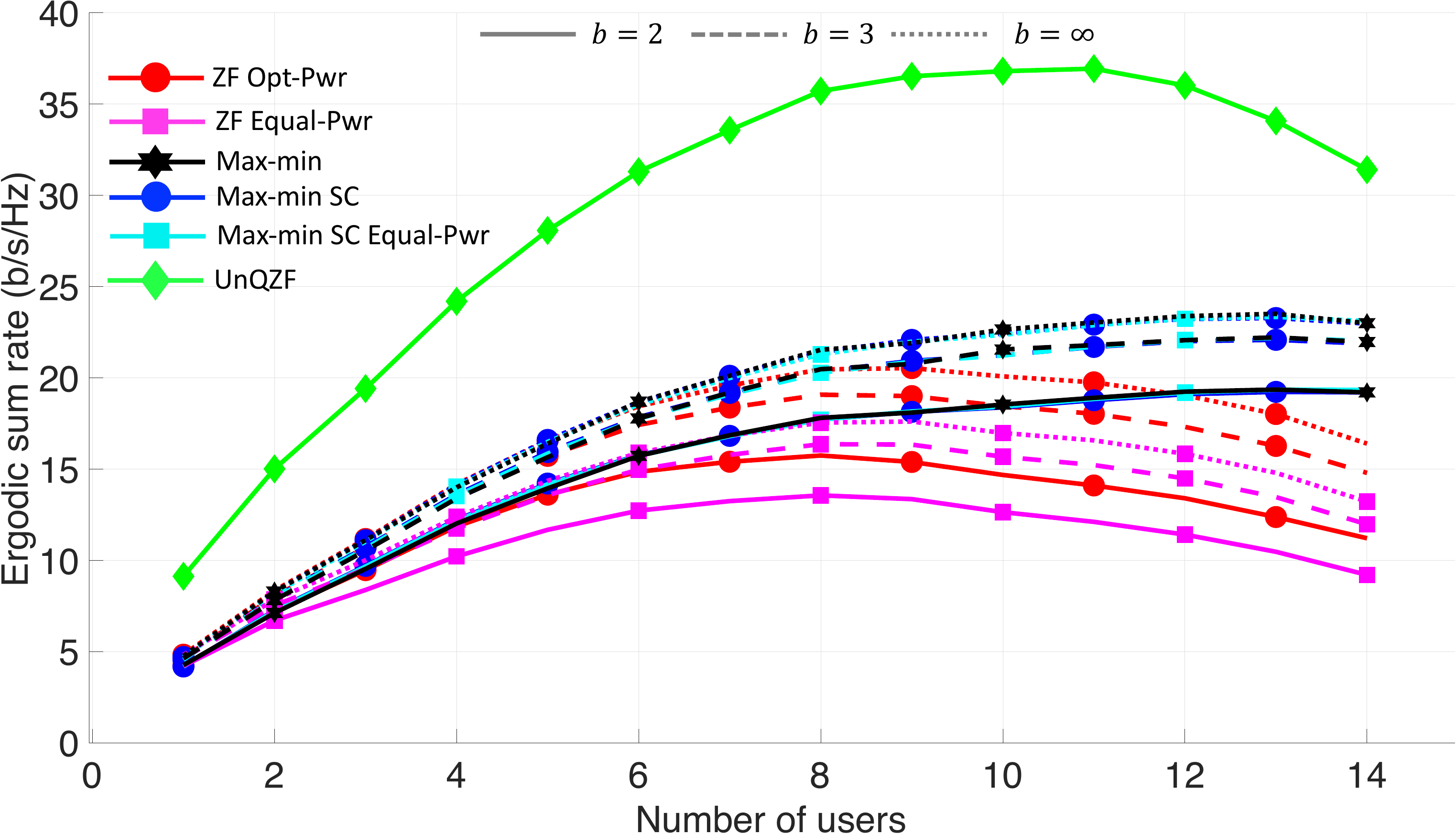}
    	\end{center}
	 \vspace{-.5cm}
    	\caption{Ergodic sum rate versus number of users. All versions of the proposed algorithm increasingly outperform ZF as the number of users increases. Max-min SC and Max-min SC Equal-Pwr achieve performance similar to Max-min. Lastly, the difference in performance between $b=3$ and $b=\infty$ is not significant.}
    	 \label{fig:meanCap}
	 \vspace{-.6cm}
\end{figure}

The linearity of the power amplifier over the dynamic range of the input signal is a critical issue especially for OFDM due to its high peak to average power ratio (PAPR). The design of linear amplifiers with a large dynamic range is further exacerbated at mmWave frequencies. The optimal power allocation described in Section \ref{sec:power} requires amplifiers which are linear over large bandwidths which contradicts the motivation of using low-resolution quantizers. In this context, an equal per-antenna power allocation is a useful solution to further reduce hardware complexity. With this design choice, amplifiers can be made to operate in their saturation region at a fixed power point without any back-off and further reduce the total power consumption. Fig. \ref{fig:meanCap} also illustrates the ergodic sum rate for the proposed solution and ZF with \emph{equal} per-antenna power allocation versus number of users. Looking at the {equal} per-antenna power allocation curves in Fig. \ref{fig:meanCap}, it can be observed that the performance is significantly deteriorated for ZF based precoding compared to optimal power allocation. The proposed solution on the other hand performs the same as optimal per-antenna power allocation. This is another advantage of the proposed method from a power amplifier and circuit design perspective.

\begin{figure}[t]
    	\begin{center}
    		\includegraphics[width=.5\textwidth,clip,keepaspectratio]{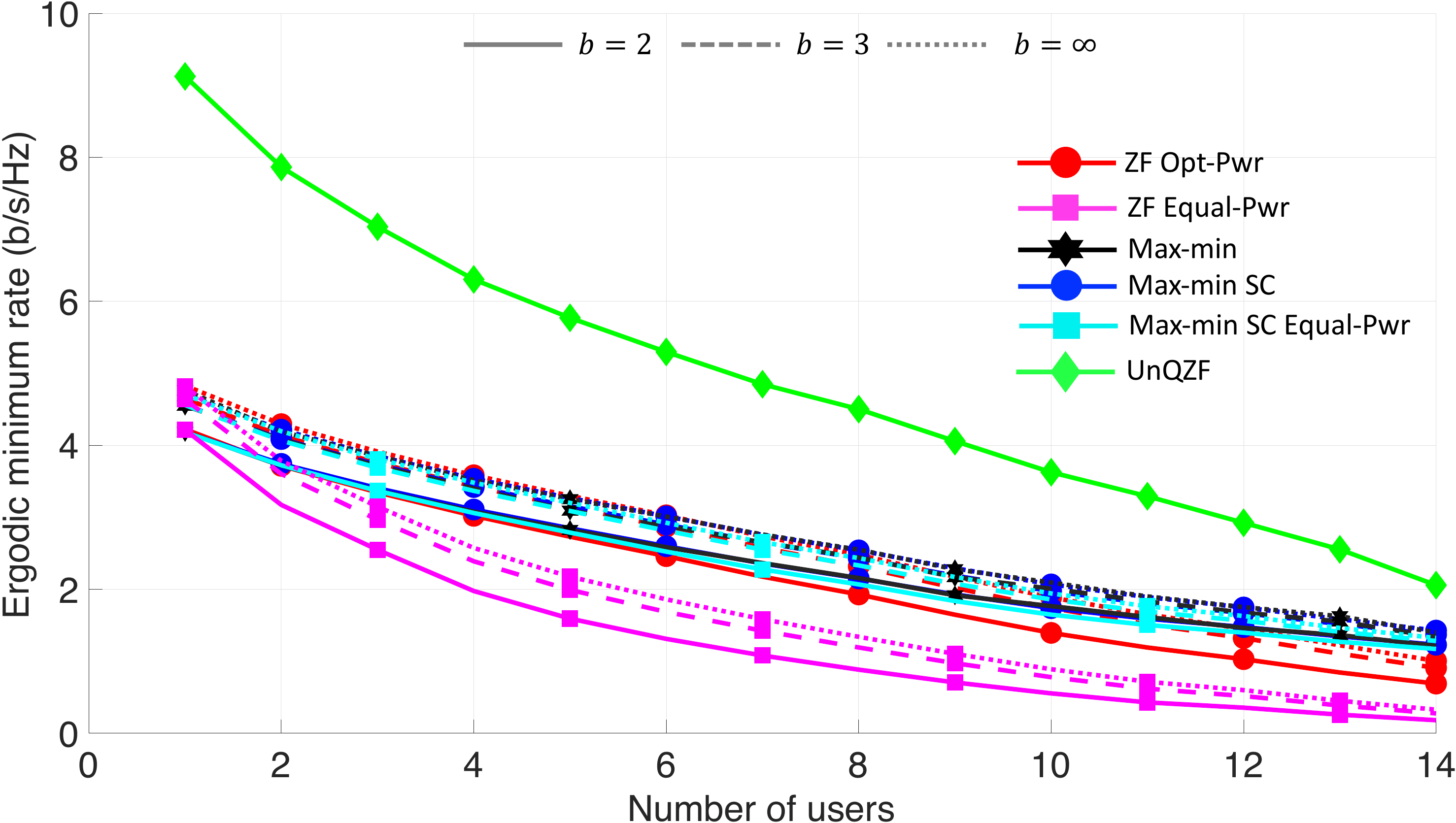}
    	\end{center}
	 \vspace{-.4cm}
    	\caption{Ergodic minimum rate  versus number of users. The performance of ZF deteriorates more as number of users increase. As observed in Fig. \ref{fig:meanCap}, Max-min SC and Max-min SC Equal-Pwr perform the same as Algorithm \ref{alg:alg1}. Furthermore, $b=3$ and $b=\infty$  achieve similar performance.}
    	 \label{fig:minCap}
	 \vspace{-.6cm}
\end{figure}

 The ergodic minimum rate is shown in Fig.~\ref{fig:minCap}. It can be seen that the performance of the ZF precoding (including UnQZF) deteriorates more than that of the proposed solution for larger number of users. For $K=14$, the two solutions differ by 0.5-1 b/s/Hz. The improvement in terms of the minimum rate might seem small but can be very important from an outage probability and fairness perspective. Similar to Fig.~\ref{fig:meanCap}, the per sub-carrier and equal per-antenna power allocation versions of the proposed solution perform the same as the version considering all sub-carriers together whereas the performance of ZF with equal per-antenna power allocation drops sharply. From here onwards, we are only going to consider the equal per-antenna power allocation per sub-carrier version (Max-min SC Equal-Pwr) of the proposed solution and ZF with equal per-antenna power due to their practical importance.



\vspace{-0.3cm}
\subsection{BER results} \label{sec:BER}
Now we look at the coded BER results for transmit symbols drawn from unit-norm normalized QPSK and 16-Quadrature Amplitude Modulation (16-QAM) constellations in an IID manner. For each channel realization, the BER is calculated by generating data bits that span 60 OFDM symbols. The data bits are encoded using a convolution encoder and then randomly interleaved across the sub-carriers. On the receive side, we first use the blind estimation method from \cite{hela} in  which a block of received symbols is used to estimate the appropriate scaling factor before sending the symbols to a max-log detector. The soft output is then fed into a max-log BCJR decoder made available by Christoph Studer in the process of his work in \cite{StuderSQUID} which we are also comparing against. The resulting BER is then further averaged over multiple channel realizations.

\begin{figure}[t]
    	\begin{center}
    		\includegraphics[width=.5\textwidth,clip,keepaspectratio]{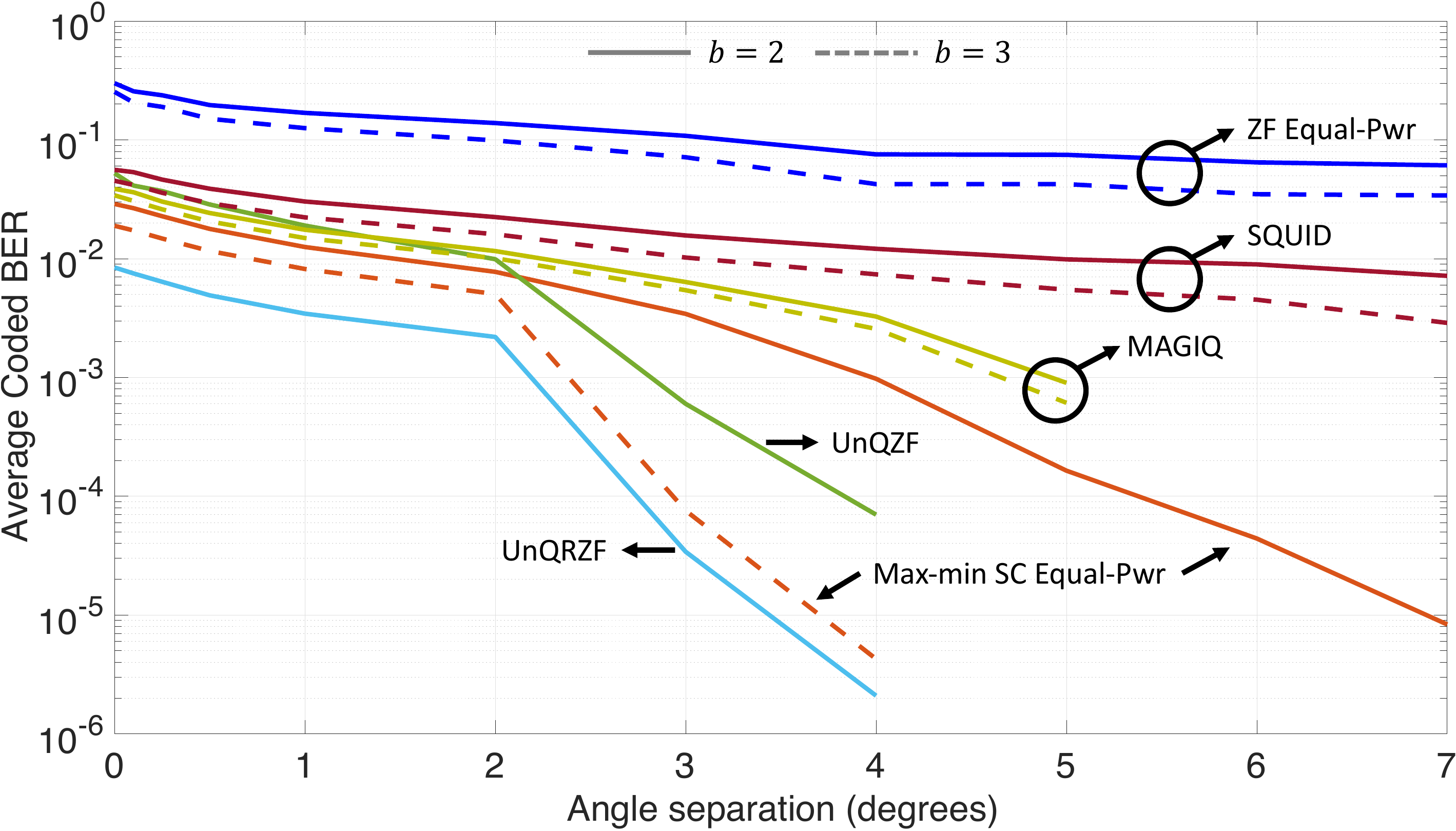}
    	\end{center}
	 \vspace{-.5cm}
    	\caption{Coded BER for QPSK constellation with a rate 1/2 convolution code for $K = 10$, $P_{\BS} = 36$ dBm and $\SG = 24$ against the minimum angle separation between users. The proposed solution outperforms all other techniques and even performs better than UnQZF for $b = 3$. }
    	 \label{fig:ber1}
	 \vspace{-.6cm}
\end{figure}

Fig. \ref{fig:ber1} illustrates the coded BER for the proposed algorithm and benchmark strategies for symbols drawn from the QPSK constellation with a rate 1/2 convolution code, $\SG = 24$, $K = 10$ and $P_{\BS} = 36$ dBm against the minimum spatial separation between active users. It can be seen that the proposed solution outperforms all CEQ precoding algorithms including the UnQZF (for $b = 3$). Furthermore, when the minimum separation between users is not limited, all algorithms (including the unquantized setting) deteriorate in performance due to the increased correlation.

Next, we look at the performance for a fixed minimum angle separation of $2^\circ$ as the ratio of number of BS antennas to the number of users is varied. The solid (dashed) set of lines in Fig. \ref{fig:ber2} illustrates coded BER against the number of BS antennas for $K =10$ ($K = 5$), QPSK constellation with a rate 1/2 convolution code, $\SG = 24$, $P_{\BS} = 40$ dBm and $b = 3$. It can be seen that the proposed solution achieves the best performance out of all CEQ precoding solutions. At relatively higher ratios of the number of BS antennas to the number of users (i.e. for $K = 10$), both SQUID and MAGIQ are not able to achieve acceptable values of coded BER ($\approx 10^{-4}$).

\begin{figure}[b]
	 \vspace{-.5cm}
    	\begin{center}
    		\includegraphics[width=.5\textwidth,clip,keepaspectratio]{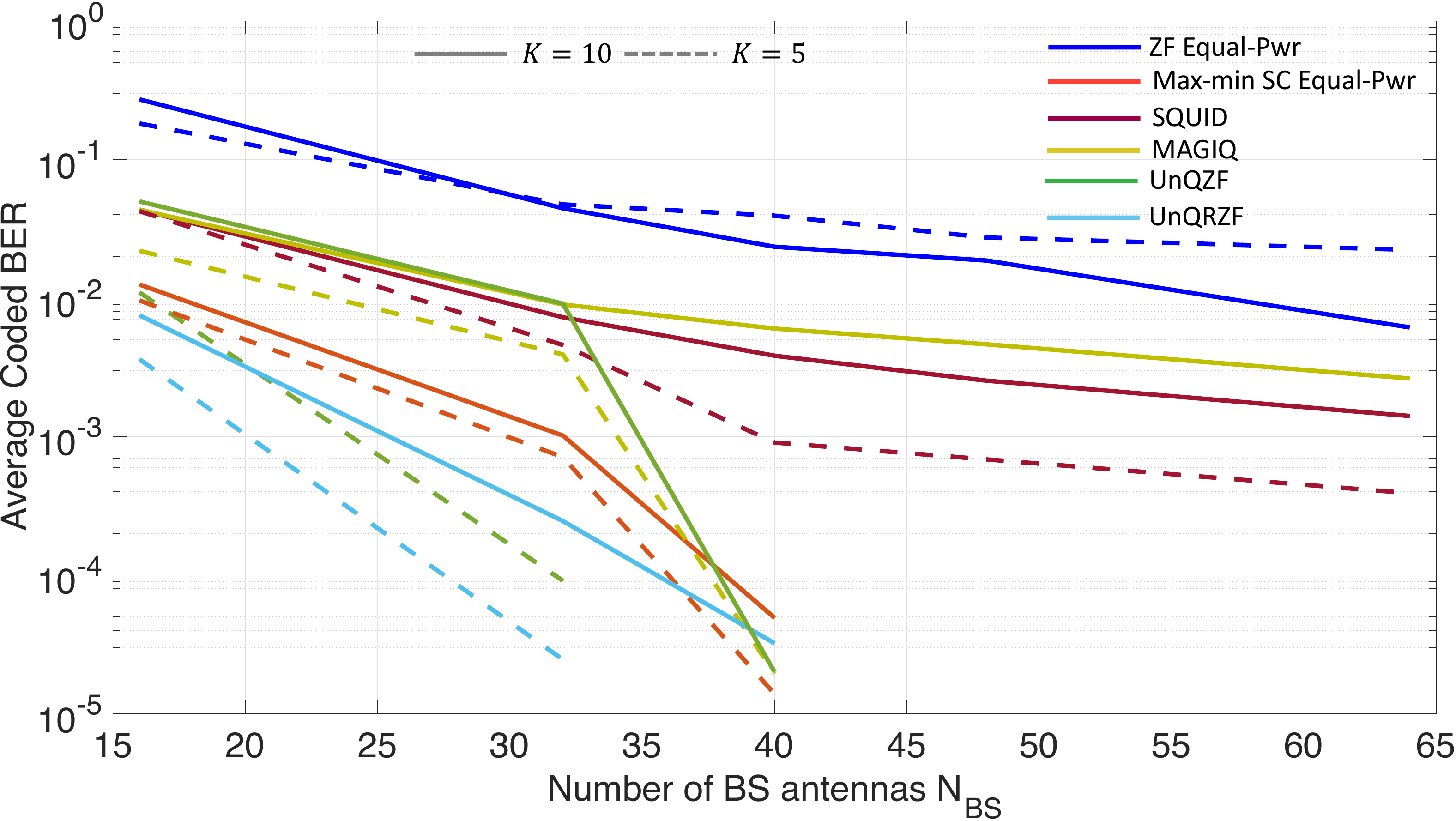}
    	\end{center}
	 \vspace{-.5cm}
    	\caption{Coded BER for $K = 5$ and $K = 10$ (for a fixed minimum angle separation of $2^\circ$) with QPSK constellation with a rate 1/2 convolution code, $\SG = 24$, $P_{\BS} = 40$ dBm and $b = 3$ versus the number of BS antennas. The proposed solution achieves the best performance. }
    	 \label{fig:ber2}
	 \vspace{-.6cm}
\end{figure}

\begin{figure}[t]
    	\begin{center}
    		\includegraphics[width=.5\textwidth,clip,keepaspectratio]{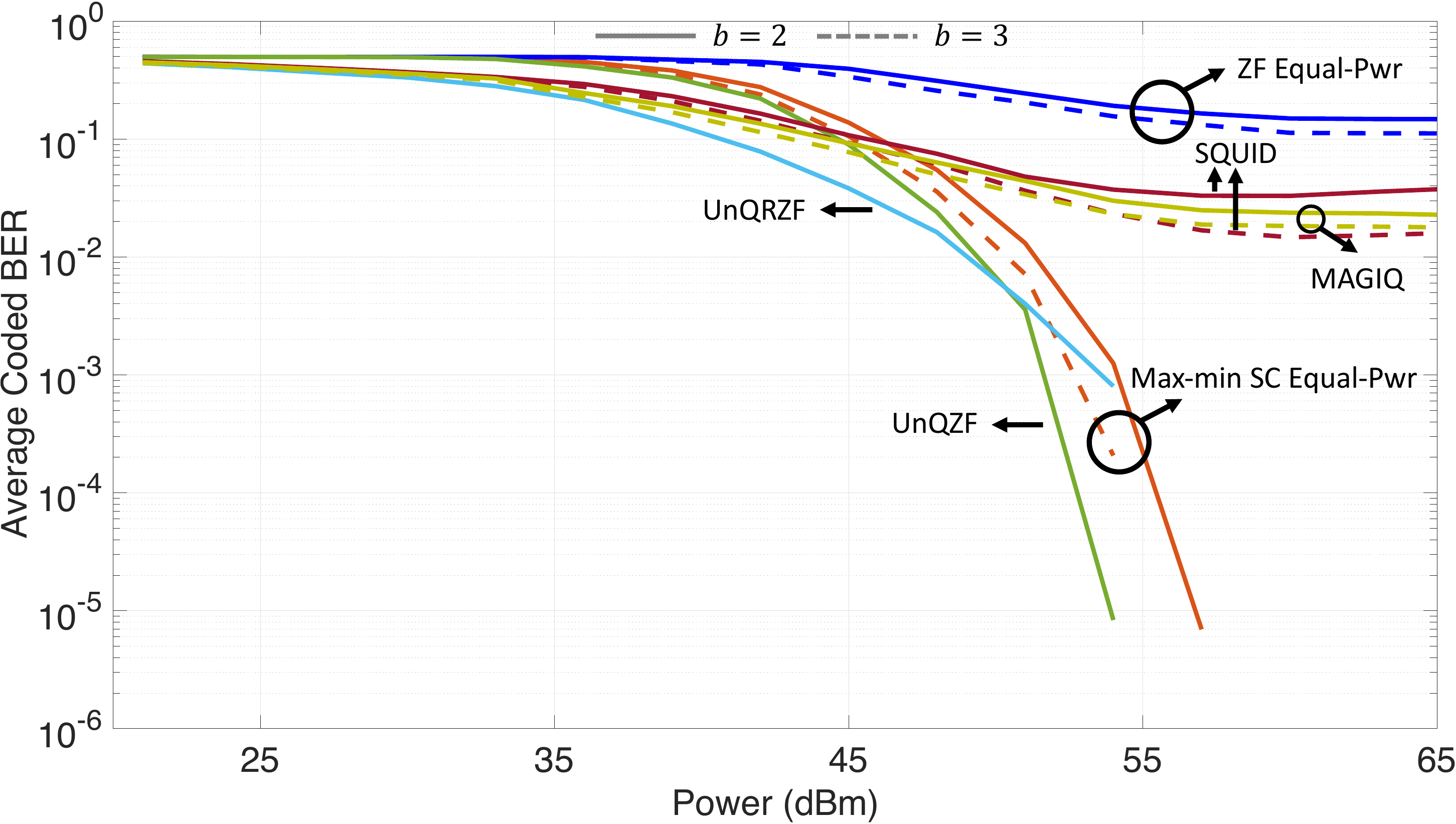}
    	\end{center}
	 \vspace{-.5cm}
    	\caption{Coded BER for QPSK constellation with a rate 1/2 convolution code for $K = 10$ and $\SG = 24$ against the transmit power $P_{\BS}$ for the 3GPP 38.901 UMa NLoS channel model. The proposed solution operates close to the unquantized benchmarks. SQUID and MAGIQ floor out at BERs close to $10^{-2}$ due to the relatively low number of BS antennas to number of users~ratio.}
    	 \label{fig:ber3}
	 \vspace{-.6cm}
\end{figure}

Next, we look at the performance for the 3GPP 38.901 UMa NLoS channel model as the transmit power is varied. Fig. \ref{fig:ber3} illustrates the uncoded BER for the proposed algorithm and benchmark strategies for symbols drawn from a QPSK constellation with a rate 1/2 convolution code, $K = 10$ and $\SG = 24$ versus the transmit power $P_{\BS}$. The first thing to observe is that the power value where the BER goes to 0 or stagnates ($\approx 51 $ dBm) is relatively higher compared to the previous results. This difference can be attributed to the NLoS channel model. We observed that the highest mode of the NLoS channel model generated by Quadriga was a factor of 10 or so less compared to the LoS channel model. Nevertheless, the proposed solution outperforms the non-linear algorithms and operates quite close to the unquantized benchmarks. SQUID and MAGIQ outperform the Max-min solution only at the lower end of the transmit power but that is not of interest due to the high BER. Furthermore, SQUID and MAGIQ floor out at BERs close to $10^{-2}$ due to the relatively low number of BS antennas to number of users~ratio.

In this paper, we assumed the availability of channel state information at the BS and did not explicitly account for the loss due to the channel estimation error. Similar to \cite{Studer,StuderOFDM,magic2}, we look at the performance of all algorithms as the normalized channel estimation error is varied from 0 to 1 to bridge this gap. Channel estimation with low-resolution ADCs is a completely separate topic with a rich existing literature \cite{CE1,CE2} and can not be addressed in this paper due to limited space. The dashed (solid) lines in Fig. \ref{fig:ber4} illustrate the coded BER for the proposed algorithm and benchmark strategies for symbols drawn from the QPSK (16QAM) constellation with a rate 1/2 convolution code, $K = 5$, $\SG = 16$, $P_{\BS} = 55$ dBm (40 dBm) and $b = 2$ ($b = 3$) for the NLoS (LoS) channel model. It can be observed that the proposed method achieves the best performance over a wide range of the normalized channel estimation error for both LoS and NLoS channel models with lower and higher order constellations.

\begin{figure}[t]
    	\begin{center}
    		\includegraphics[width=.5\textwidth,clip,keepaspectratio]{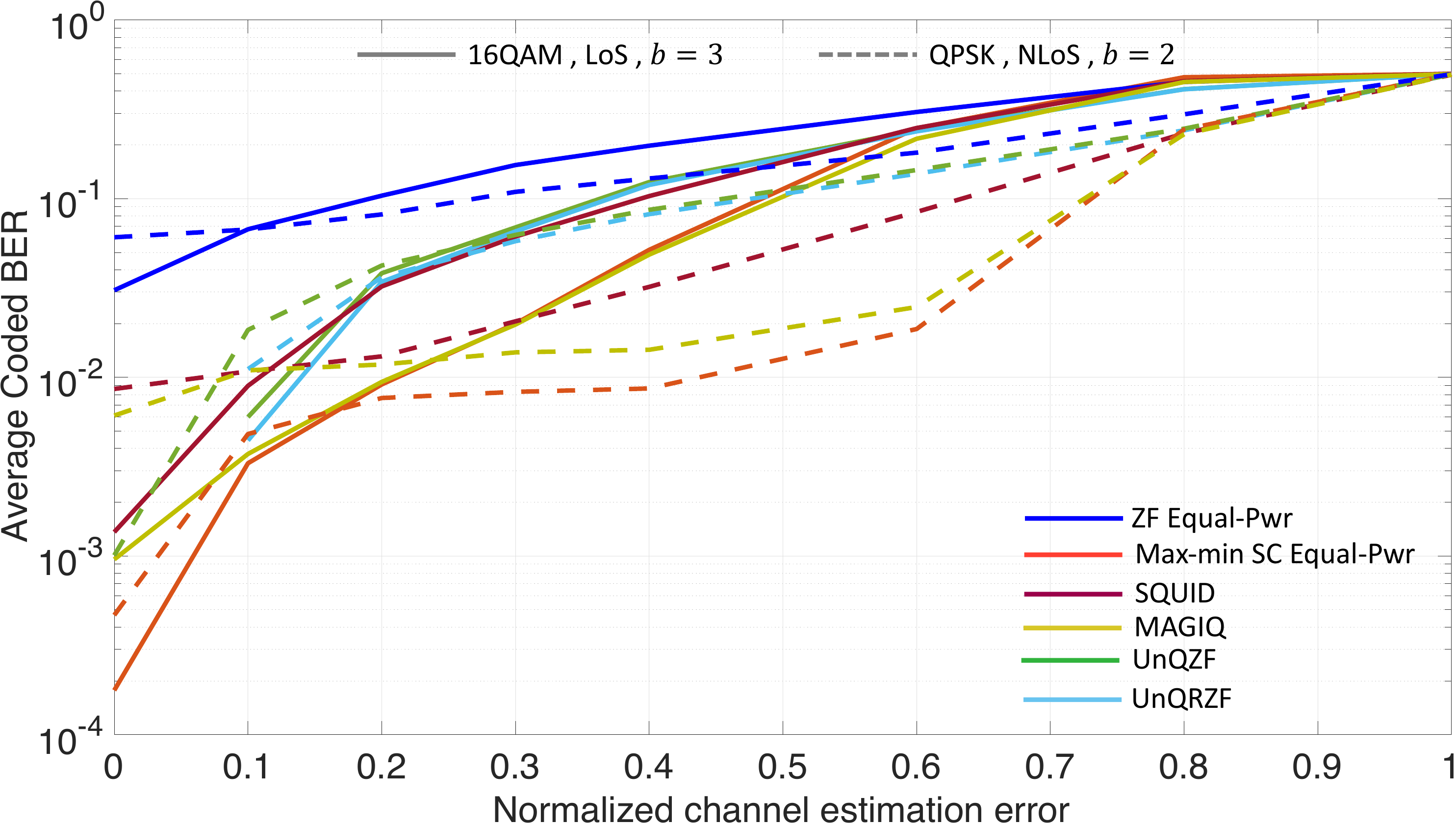}
    	\end{center}
	 \vspace{-.4cm}
    	\caption{Dashed (solid) lines illustrate coded BER for QPSK (16QAM) constellation with a rate 1/2 convolution code for $K = 5$, $\SG = 16$, $P_{\BS} = 55$ dBm (40 dBm) and $b = 2$ ($b = 3$) on the NLoS (LoS) channel model. This result demonstrates that the proposed solution outperforms existing algorithms over both LoS and NLoS channel models for lower and higher order constellations over a wide range of the normalized channel estimation error.}
    	 \label{fig:ber4}
	 \vspace{-.7cm}
\end{figure}

The results presented in Section \ref{sec:SQINR} and \ref{sec:BER} demonstrate that the proposed solution improves considerably on existing linear precoding techniques and even outperforms the benchmark non-linear precoding methods \cite{StuderSQUID,magic2}. This is in contrast to the results published in existing literature \cite{StuderSQUID,MSM1,MSM2,MSM3,magic,magic2,access} which (with the exception of one result in \cite{magic2}) have been obtained on channels with IID Gaussian entries. Our experiments on IID Gaussian channels (not presented in this paper) confirm the results published in prior work. Various factors, however, limit the performance of the non-linear precoding methods when using the realistic channel models presented in this work. The self-correlation of a user's channel entries across the antenna elements adversely affects all methods as also observed in \cite{magic2}. The cross-correlation between channels of different users is another limiting factor that seems to affect the non-linear methods more especially at lower number of BS antennas to users ratio as seen in Fig. \ref{fig:ber1} and \ref{fig:ber2}. Another factor that limits the performance of non-linear methods is the mismatch between channels of different users. A user with a weaker channel is adversely affected a lot more by \cite{StuderSQUID,magic2} bringing down the average performance significantly whereas max-min tries to maximize the performance of the weak user.

 The flexibility to assign per user/sub-carrier SQINRs makes the proposed solution even more attractive. These constraints were all set equal in the results presented in this manuscript. Further improvement might be possible by tweaking these constraints in favor of users/sub-carriers with better channel quality. For example, a lot of power might be wasted if two users cause significant interference to each other. In that setting, the SQINR constraint of one of the users can be reduced lowering its achievable SQINR/BER but improving the SQINR/BER performance of all other users. This is a scheduling problem that can be dealt with by the system operator at a higher level and then integrated into the presented framework using the individual SQINR constraints.

\section{Conclusion} \label{sec:conc}
In this paper, we presented a linear precoding based solution to the MU-MIMO-OFDM DL precoding problem under CEQ DAC constraints at the BS and per-user SQINR constraints. Our proposed solution, based on UL-DL duality, maximized the minimum ratio of the achieved SQINR to target SQINR over all sub-carrier of all users. We further reduced the complexity of the proposed algorithm by parallelizing it over the individual sub-carriers. Our results in terms of the ergodic sum and minimum rate showed that the proposed solution outperforms existing linear precoding strategies. Furthermore, we showed that the max-min solution performs better than the non-linear methods in terms of coded BER for the channel models considered in this~paper. 

The analysis carried out in this paper complements our prior work by generalizing the UL-DL duality principle under 1-bit hardware constraints from flat fading channels to frequency selective channels for CEQs. The key insight in both these results was that the quantization noise has to be uncorrelated. This was enforced by adding optimized dithering to the system in the form of dummy users which operate in the null space of the active users. Our future work in this direction will incorporate ideas from this paper, such as optimized dithering and per-user SQINR constraints, to improve the performance of non-linear algorithms. Another interesting line of work is to conduct a detailed study about the various factors mentioned in this paper that limit the performance of non-linear methods. A detailed evaluation of the proposed and existing approaches from a circuits implementation point of view is another exciting avenue of work.

\vspace{-0.4cm}
\bibliographystyle{IEEEbib}
\bibliography{ref}

\begin{thebibliography}{10}

\bibitem{SP}
R.~W. {Heath}, N.~{Gonzalez-Prelcic}, S.~{Rangan}, W.~{Roh}, and A.~M.
  {Sayeed},
\newblock ``An overview of signal processing techniques for {mmWave MIMO}
  systems,''
\newblock {\em IEEE Journal of Selected Topics in Signal Processing}, vol. 10,
  no. 3, pp. 436--453, 2016.

\bibitem{NYU}
P.~Skrimponis, N.~Hosseinzadeh, A.~Khalili, E.~Erkip, M.~J.~W. Rodwell, J.~F.
  Buckwalter, and S.~Rangan,
\newblock ``Towards energy efficient mobile wireless receivers above 100
  {GHz},''
\newblock {\em IEEE Access}, vol. 9, pp. 20704--20716, 2021.

\bibitem{powerStuder}
P.~Skrimponis, S.~Dutta, M.~Mezzavilla, S.~Rangan, S.~H. Mirfarshbafan,
  C.~Studer, J.~Buckwalter, and M.~Rodwell,
\newblock ``Power consumption analysis for mobile {MmWave} and {Sub-THz}
  receivers,''
\newblock in {\em Proc. of 2020 6G SUMMIT}, 2020, pp. 1--5.

\bibitem{StuderOFDMconf}
S.~Jacobsson, G.~Durisi, M.~Coldrey, and C.~Studer,
\newblock ``Massive {MU-MIMO-OFDM} downlink with one-bit {DACs} and linear
  precoding,''
\newblock in {\em Proc. of {GLOBECOM} 2017}, 2017, pp. 1--6.

\bibitem{StuderOFDM}
S.~Jacobsson, G.~Durisi, M.~Coldrey, and C.~Studer,
\newblock ``Linear precoding with low-resolution {DACs} for massive
  {MU-MIMO-OFDM} downlink,''
\newblock {\em IEEE Trans. on Wireless Communications}, vol. 18, no. 3, pp.
  1595--1609, 2019.

\bibitem{amine2009}
A.~Mezghani, R.~Ghiat, and J.~A. Nossek,
\newblock ``Transmit processing with low resolution {D/A}-converters,''
\newblock in {\em Proc. of 2009 16th IEEE ICECS}, 2009, pp. 683--686.

\bibitem{amine2016}
O.~B. Usman, H.~Jedda, A.~Mezghani, and J.~A. Nossek,
\newblock ``{MMSE} precoder for massive {MIMO} using 1-bit quantization,''
\newblock in {\em Proc. of 2016 IEEE ICASSP}, 2016, pp. 3381--3385.

\bibitem{StuderSQUID}
S.~Jacobsson, O.~Castañeda, C.~Jeon, G.~Durisi, and C.~Studer,
\newblock ``Nonlinear precoding for phase-quantized constant-envelope massive
  {MU-MIMO-OFDM},''
\newblock in {\em Proc. of 2018 25th {ICT}}, 2018, pp. 367--372.

\bibitem{MSM1}
H.~Jedda, A.~Mezghani, J.~A. Nossek, and A.~L. Swindlehurst,
\newblock ``Massive {MIMO} downlink 1-bit precoding for frequency selective
  channels,''
\newblock in {\em Proc. of 2017 IEEE 7th International Workshop on {CAMSAP}},
  2017, pp. 1--5.

\bibitem{MSM2}
H.~Jedda and J.~A. Nossek,
\newblock ``Quantized constant envelope precoding for frequency selective
  channels,''
\newblock in {\em Proc. of 2018 {IEEE SSP}}, 2018, pp. 213--217.

\bibitem{MSM3}
F.~Askerbeyli, H.~Jedda, and J.~A. Nossek,
\newblock ``1-bit precoding in massive {MU-MISO-OFDM} downlink with linear
  programming,''
\newblock in {\em Proc. of 23rd International ITG WSA}, 2019, pp. 1--5.

\bibitem{magic}
A.~Nedelcu, F.~Steiner, M.~Staudacher, G.~Kramer, W.~Zirwas, R.~S. Ganesan,
  P.~Baracca, and S.~Wesemann,
\newblock ``Quantized precoding for multi-antenna downlink channels with
  {MAGIQ},''
\newblock in {\em Proc. of WSA 2018}, 2018, pp. 1--8.

\bibitem{magic2}
A.~S. Nedelcu, F.~Steiner, and G.~Kramer,
\newblock ``Low-resolution precoding for multi-antenna downlink channels and
  {OFDM},''
\newblock {\em Entropy}, vol. 24, no. 4, pp. 504, Apr 2022.

\bibitem{access}
C.~G. Tsinos, S.~Domouchtsidis, S.~Chatzinotas, and B.~Ottersten,
\newblock ``Symbol level precoding with low resolution {DACs} for constant
  envelope {OFDM MU-MIMO} systems,''
\newblock {\em IEEE Access}, vol. 8, pp. 12856--12866, 2020.

\bibitem{oob}
R.~Okawa and Y.~Sanada,
\newblock ``Power-based criteria for signal reconstruction using 1-bit
  resolution {DACs} in massive {MU-MIMO OFDM} downlink,''
\newblock {\em IEICE Trans. on Communications}, vol. E104.B, no. 10, pp.
  1299--1306, 2021.

\bibitem{oob2}
T.~Yamakado, R.~Okawa, and Y.~Sanada,
\newblock ``Quantized precoding for out-of-band radiation reduction in massive
  {MU-MIMO-OFDM},''
\newblock in {\em Proc. of 2021 IEEE 94th VTC2021-Fall}, 2021, pp. 1--5.

\bibitem{amodh}
A.~K. {Saxena}, A.~{Mezghani}, and R.~W. {Heath},
\newblock ``Linear {CE} and 1-bit quantized precoding with optimized
  dithering,''
\newblock {\em {IEEE} Open Journal of Signal Processing}, vol. 1, pp. 310--325,
  2020.

\bibitem{pw}
K.~U. Mazher, A.~Mezghani, and R.~W. Heath,
\newblock ``Multi-user downlink beamforming using uplink downlink duality with
  1-bit converters,''
\newblock in {\em Proc. of 2021 IEEE SPAWC}, 2021, pp. 1--5.

\bibitem{pw2}
K.~U. Mazher, A.~Mezghani, and R.~W. Heath,
\newblock ``Multi-user downlink beamforming using uplink downlink duality with
  1-bit converters for flat fading channels,''
\newblock {\em preprint, arXiv:2206.14427}, 2022.

\bibitem{Studer}
S.~{Jacobsson}, G.~{Durisi}, M.~{Coldrey}, T.~{Goldstein}, and C.~{Studer},
\newblock ``Quantized precoding for massive {MU-MIMO},''
\newblock {\em IEEE Trans. on Communications}, vol. 65, no. 11, pp. 4670--4684,
  2017.

\bibitem{polar}
J.~Groe,
\newblock ``Polar transmitters for wireless communications,''
\newblock {\em IEEE Comm. Magazine}, vol. 45, no. 9, pp. 58--63, 2007.

\bibitem{SE}
Y.~Li, C.~Tao, A.~Mezghani, A.~L. Swindlehurst, G.~Seco-Granados, and L.~Liu,
\newblock ``Optimal design of energy and spectral efficiency tradeoff in
  one-bit massive {MIMO} systems,''
\newblock {\em preprint, arXiv:1612.03271}, 2016.

\bibitem{statistical}
H.~Jedda and J.~A. Nossek,
\newblock ``On the statistical properties of constant envelope quantizers,''
\newblock {\em IEEE Wireless Communications Letters}, vol. 7, no. 6, pp.
  1006--1009, 2018.

\bibitem{MUDL}
M.~{Schubert} and H.~{Boche},
\newblock ``Solution of the multiuser downlink beamforming problem with
  individual {SINR} constraints,''
\newblock {\em IEEE Trans. on Vehicular Technology}, vol. 53, no. 1, pp.
  18--28, 2004.

\bibitem{Quadriga}
S.~{Jaeckel}, L.~{Raschkowski}, K.~{Börner}, and L.~{Thiele},
\newblock ``Quadriga: A {3-D} multi-cell channel model with time evolution for
  enabling virtual field trials,''
\newblock {\em IEEE Trans. on Antennas and Propagation}, vol. 62, no. 6, pp.
  3242--3256, June 2014.

\bibitem{hela}
H.~Jedda, A.~Mezghani, A.~L. Swindlehurst, and J.~A. Nossek,
\newblock ``Quantized constant envelope precoding with {PSK and QAM}
  signaling,''
\newblock {\em IEEE Trans. on Wireless Communications}, vol. 17, no. 12, pp.
  8022--8034, 2018.

\bibitem{CE1}
H.~Wang, W.~Shih, C.~Wen, and S.~Jin,
\newblock ``Reliable {OFDM} receiver with ultra-low resolution {ADC},''
\newblock {\em IEEE Trans. on Communications}, vol. 67, no. 5, pp. 3566--3579,
  2019.

\bibitem{CE2}
X.~Cheng, B.~Xia, K.~Xu, and S.~Li,
\newblock ``Bayesian channel estimation and data detection in oversampled
  {OFDM} receiver with low-resolution {ADC},''
\newblock {\em IEEE Trans. on Wireless Comm.}, vol. 20, no. 9, pp. 5558--5571,
  2021.

\end{thebibliography}

\newpage

\section{Supplementary Results} \label{sec:suppResults}
In this supplementary note, we present numerical results similar to those presented in Section \ref{sec:results} but for a larger system. The simulation setup and comparison strategies/metrics stay the same. In this set of results, we only compare with the non-linear algorithm SQUID from \cite{StuderSQUID}. The other non-linear comparison algorithm from the original manuscript, MAGIQ \cite{magic2}, takes an infeasible amount of time to run for this larger system. We, however, expect its performance to be similar to SQUID as observed from the results presented in Section \ref{sec:results} and as demonstrated in \cite{magic2}. 

\vspace{-.3cm}
\subsection{Simulation setup}
The important simulation parameters (unless stated otherwise) are given in Table \ref{table:table2Supp}.

\begin {table}[h]
\begin{center}
\begin{tabular}{ | >{\centering\arraybackslash} m{3.6cm} | >{\centering\arraybackslash} m{4.8cm} | }
  \hline 			
  {Quadriga channel model} & {$\text{3GPP 38.901 UMa LoS / NLoS}$} \\  			
  \hline 
  { Number of antennas ${\BS}$} & 64 \\  			
  \hline 
  { $\text{Antenna element pattern}$} & $\text{0 dBi omni-directional}$ \\  			
  \hline 
  { Total transmit power $P_{\BS}$} & 35 dBm \\  			
  \hline 
  {Carrier frequency $f_c$} & 60 GHz \\  			
  \hline 
  {Bandwidth $B$} & 100 MHz (LoS) / 50 MHz (NLoS) \\  			
  \hline 
  {Number of subcarriers $\SC$} & 128 \\  			
  \hline 
  {Cyclic prefix length $\CP$} & 32 \\  			
  \hline 
  {SQINR constraint $\{\gamma_\kn\}$} & \{ 3 dB \} \\  			
  \hline 
  {CEQ resolution b } & \{ 2 , 3 \} \\  			
  \hline 
\end{tabular}
\end{center}
\vspace{-.3cm}
\caption{Important simulation parameters.}
\label{table:table2Supp} 
\end{table}

\vspace{-.5cm}
\subsection{Benchmark strategies}
As mentioned before, we will compare with the same strategies as the original paper with the exception of MAGIQ \cite{magic2}.

\subsection{SQINR results}\label{sec:SQINRSupp}

The ergodic sum rate is shown in Fig.~\ref{fig:meanCapSupp} as a function of the number of active users. It can be observed that the gap between proposed strategy and ZF is small when the number of users is small. The proposed solution, however, performs increasingly better than ZF with the number of active users. For $K = 14$, the performance in terms of the sum rate differs by about 2-4 b/s/Hz depending on the CEQ resolution.

\begin{figure}[h]
    	\begin{center}
    		\includegraphics[width=.5\textwidth,clip,keepaspectratio]{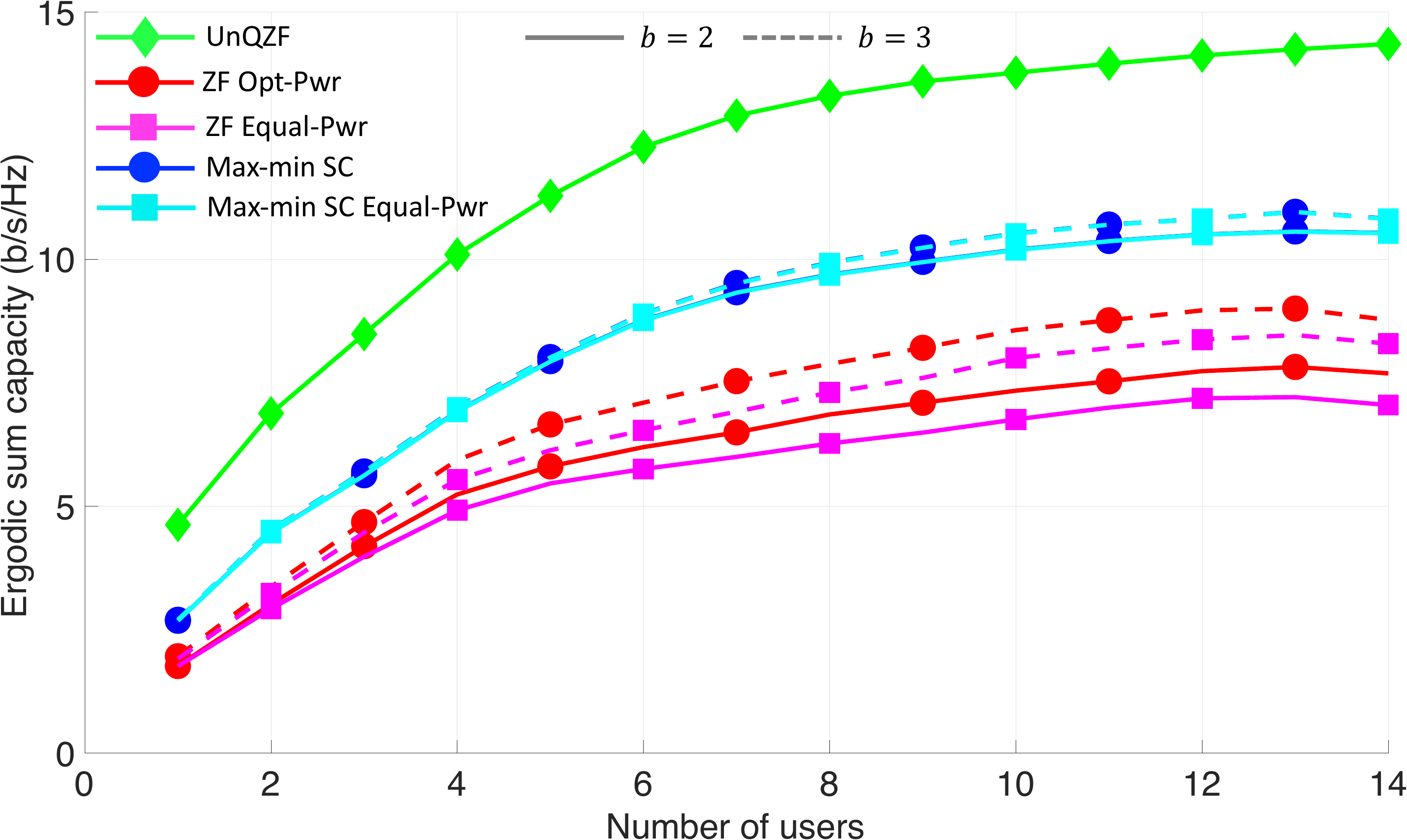}
    	\end{center}
	 \vspace{-.3cm}
    	\caption{Ergodic sum rate versus number of users. The proposed algorithm increasingly outperforms ZF as the number of users increases. Max-min SC and Max-min SC Equal-Pwr achieve similar performance.}
    	 \label{fig:meanCapSupp}
	 \vspace{-.4cm}
\end{figure}

Fig. \ref{fig:meanCapSupp} also illustrates the ergodic sum rate for the proposed solution and ZF with \emph{equal} per-antenna power allocation versus number of users. Looking at the {equal} per-antenna power allocation curves in Fig. \ref{fig:meanCapSupp}, it can be observed that the performance is deteriorated for ZF based precoding compared to optimal power allocation. The proposed solution on the other hand performs the same as optimal per-antenna power allocation. This is another advantage of the proposed method from a power amplifier and circuit design perspective.

\begin{figure}[h]
    	\begin{center}
    		\includegraphics[width=.5\textwidth,clip,keepaspectratio]{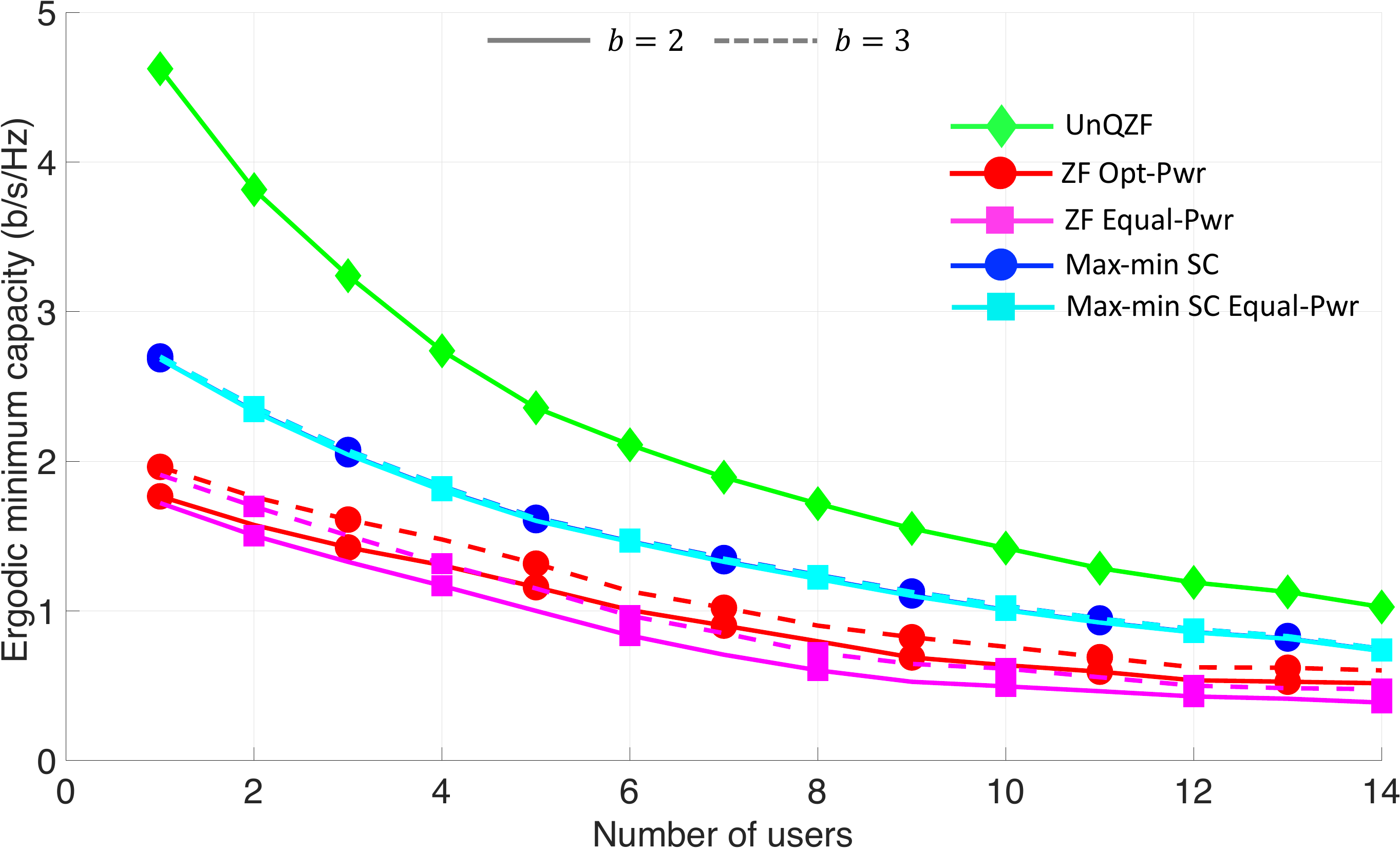}
    	\end{center}
	 \vspace{-.3cm}
    	\caption{Ergodic minimum rate  versus number of users. The performance of ZF deteriorates more as number of users increase.}
    	 \label{fig:minCapSupp}
	 \vspace{-.4cm}
\end{figure}

 The ergodic minimum rate is shown in Fig.~\ref{fig:minCapSupp}. It can be seen that the performance of the ZF precoding deteriorates more than that of the proposed solution for larger number of users. For $K=14$, the two solutions differ by 0.25-0.5 b/s/Hz. The improvement in terms of the minimum rate might seem small but can be very important from an outage probability and fairness perspective. Similar to Fig.~\ref{fig:meanCapSupp}, the optimal per-antenna and equal per-antenna power allocation versions of the proposed solution perform the same. The performance of ZF with equal per-antenna power allocation suffers more degradation. 
 


\vspace{-0.3cm}
\subsection{BER results} \label{sec:BERSupp}

Fig. \ref{fig:ber1Supp} illustrates the coded BER for the proposed algorithm and benchmark strategies for symbols drawn from the 16QAM constellation with a rate 1/2 convolution code, $\SG = 96$, $K = 10$, $P_{\BS} = 52$ dBm, and $b = 3$ against the minimum spatial separation between active users. It can be seen that the proposed solution outperforms all ZF based precoding as well as the non-linear algorithm SQUID. Furthermore, when the minimum separation between users is not limited, all algorithms (including the unquantized setting) deteriorate in performance due to the increased correlation among the user channels.

\begin{figure}[t]
    	\begin{center}
    		\includegraphics[width=.5\textwidth,clip,keepaspectratio]{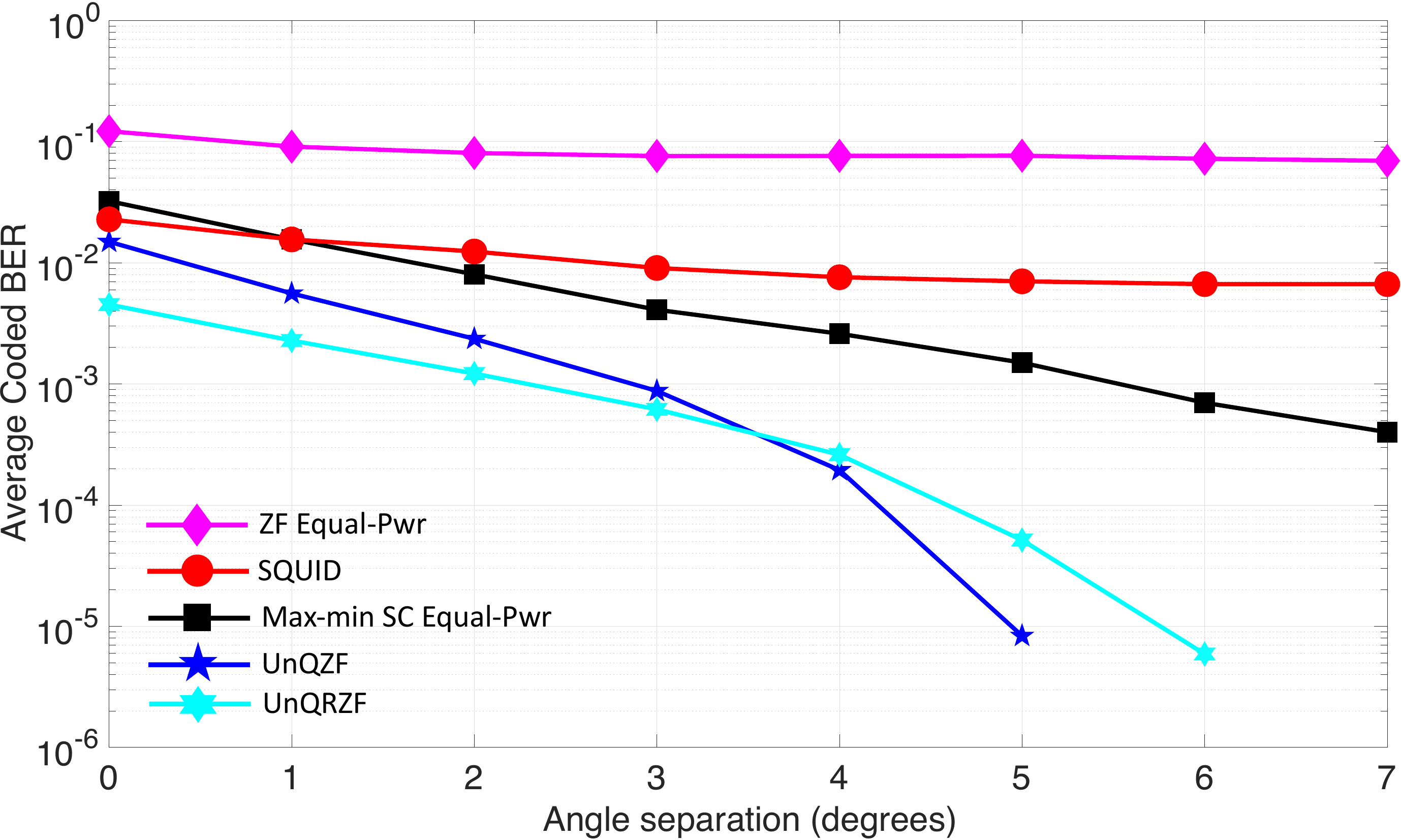}
    	\end{center}
	 \vspace{-.3cm}
    	\caption{Coded BER for 16QAM constellation with a rate 1/2 convolution code for $K = 10$, $P_{\BS} = 42$ dBm, $\SG = 96$, and $b = 3$ against the minimum angle separation between users. The proposed solution outperforms all other techniques. }
    	 \label{fig:ber1Supp}
	 \vspace{-.4cm}
\end{figure}

\begin{figure}[t]
    	\begin{center}
    		\includegraphics[width=.5\textwidth,clip,keepaspectratio]{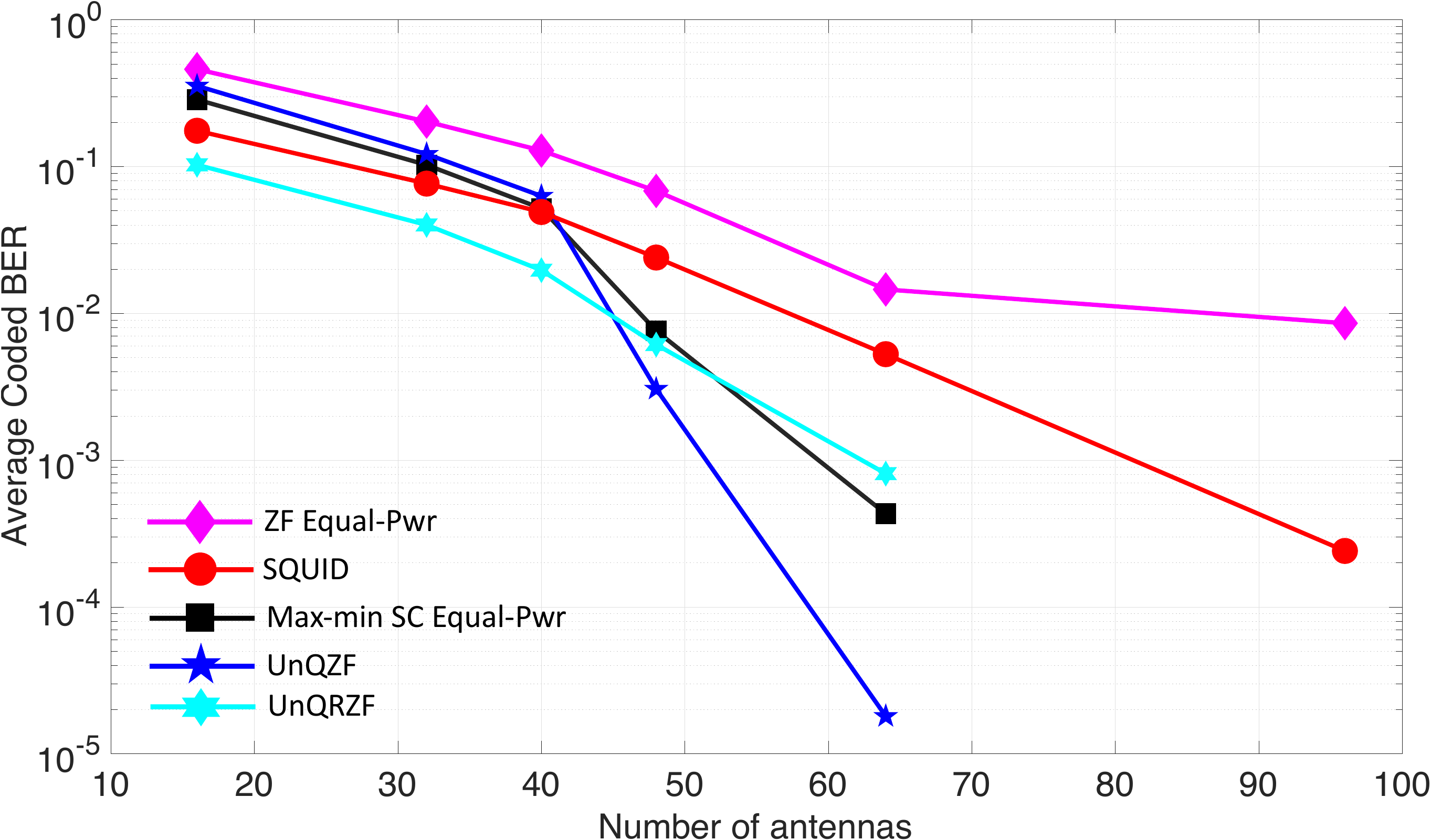}
    	\end{center}
	 \vspace{-.3cm}
    	\caption{Coded BER for $K = 10$ (for a fixed minimum angle separation of $2^\circ$) with QPSK constellation with a rate 1/2 convolution code, $\SG = 96$, $P_{\BS} = 36$ dBm and $b = 3$ versus the number of BS antennas. The proposed solution achieves the best performance. }
    	 \label{fig:ber2Supp}
	 \vspace{-.4cm}
\end{figure}

Next, we look at the performance for a fixed minimum angle separation of $2^\circ$ as the ratio of number of BS antennas to the number of users is varied. Fig. \ref{fig:ber2Supp} illustrates coded BER against the number of BS antennas for $K =10$, QPSK constellation with a rate 1/2 convolution code, $\SG = 96$, $P_{\BS} = 36$ dBm and $b = 3$. It can be seen that the proposed solution achieves the best performance out of all CEQ precoding solutions. At very high ratio of number of BS antennas to the number of users (i.e. for small number of BS antennas), SQUID does slightly better than the proposed solution but those numbers are not meaningful due to the coded BER being greater than 0.1 for all strategies being compared (including the unquantized setting).

Next, we look at the performance as the transmit power is varied. Fig. \ref{fig:ber3Supp} illustrates the uncoded BER for the 3GPP 38.901 UMa LoS channel model  for symbols drawn from a QPSK constellation with a rate 1/2 convolution code, $K = 7$, $\SG = 96$, and $b = 2$ versus the transmit power $P_{\BS}$. It can be observed that the proposed solution operates quite close to the unquantized setting and outperforms both SQUID and ZF based CEQ precoding.

\begin{figure}[h]
    	\begin{center}
    		\includegraphics[width=.5\textwidth,clip,keepaspectratio]{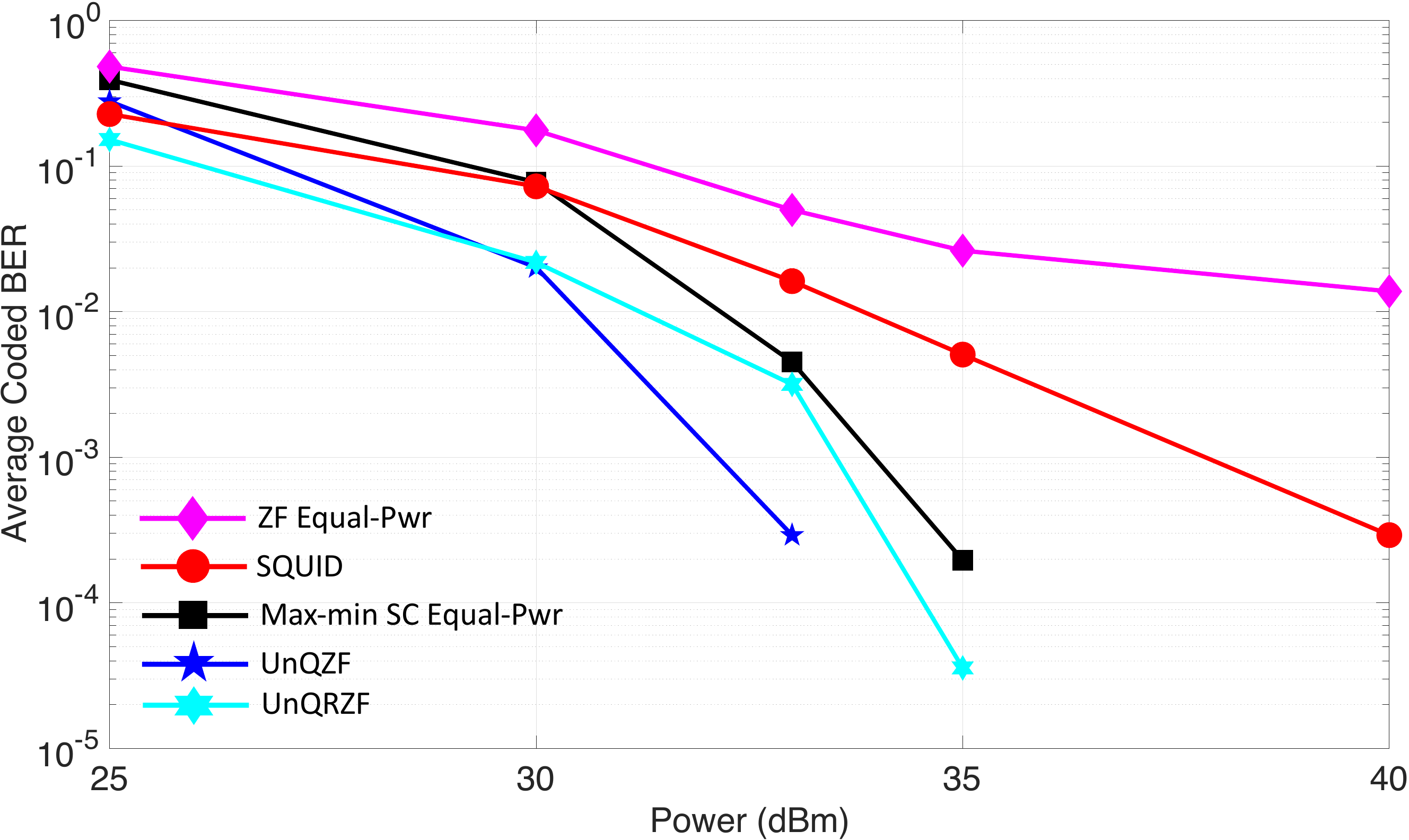}
    	\end{center}
	 \vspace{-.3cm}
    	\caption{Coded BER for QPSK constellation with a rate 1/2 convolution code for $K = 7$, $\SG = 96$, and $b = 2$ against the transmit power $P_{\BS}$ for the 3GPP 38.901 UMa LoS channel model. The proposed solution operates close to the unquantized benchmarks outperforming SQUID and ZF.}
    	 \label{fig:ber3Supp}
	 \vspace{-.4cm}
\end{figure}

\begin{figure}[t]
    	\begin{center}
    		\includegraphics[width=.5\textwidth,clip,keepaspectratio]{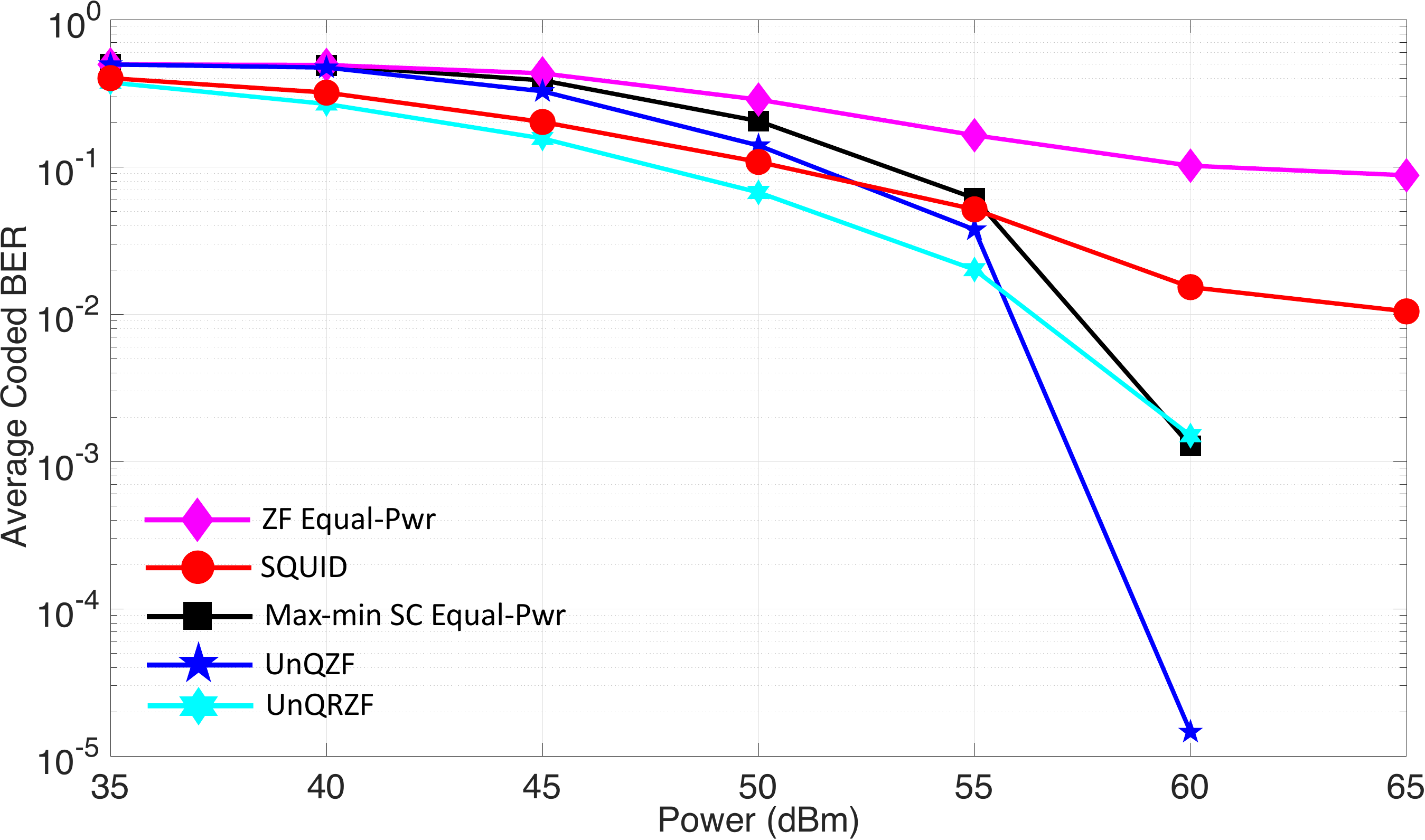}
    	\end{center}
	 \vspace{-.3cm}
    	\caption{Coded BER for QPSK constellation with a rate 1/2 convolution code for $K = 10$, $\SG = 96$ and $b = 2$ against the transmit power $P_{\BS}$ for the 3GPP 38.901 UMa NLoS channel model. The proposed solution operates close to the unquantized benchmarks. SQUID and ZF floor out at BERs close to $10^{-2}$.}
    	 \label{fig:ber32Supp}
	 \vspace{-.4cm}
\end{figure}

Fig. \ref{fig:ber32Supp} illustrates the uncoded BER for the 3GPP 38.901 UMa NLoS channel model  for symbols drawn from a QPSK constellation with a rate 1/2 convolution code, $K = 10$, $\SG = 96$, and $b = 2$ versus the transmit power $P_{\BS}$.
The first thing to observe is that the power value where the BER goes to 0 or stagnates (in case of SQUID and ZF $\approx 60 $ dBm) is relatively higher compared to the previous results. This difference can be attributed to the NLoS channel model. We observed that the highest mode of the NLoS channel model generated by Quadriga was a factor of 10 or so less compared to the LoS channel model. Nevertheless, the proposed solution achieves the performance closet to the unquantized precoders. SQUID outperforms the max-min solution only at the lower end of the transmit power but that is not of interest due to the high BER. Furthermore, SQUID floors out at BERs close to $10^{-2}$.

This last set of results in Fig. \ref{fig:ber4Supp} and Fig. \ref{fig:ber42Supp} illustrates the performance of the proposed solution and the benchmark strategies as the normalized channel estimation error is varied from 0 to 1. Fig. \ref{fig:ber4Supp} plots the coded BER for symbols drawn from the 16QAM constellation with a rate 1/2 convolution code, $K = 5$, $\SG = 64$, $P_{\BS} = 40$ dBm, and $b = 3$ for the LoS channel model. Fig. \ref{fig:ber42Supp} plots the coded BER for symbols drawn from the QPSK constellation with a rate 1/2 convolution code, $K = 5$, $\SG = 64$, $P_{\BS} = 51$ dBm, and $b = 3$ for the NLoS channel model. It can be observed from Fig. \ref{fig:ber4Supp} and Fig. \ref{fig:ber42Supp} that the proposed method achieves the best performance over a wide range of the normalized channel estimation error for both LoS and NLoS channel models with lower and higher order constellations. The proposed solution also outperforms the unquantized precoders when the CSI is not perfect demonstrating its robustness to channel estimation errors.

\begin{figure}[t]
    	\begin{center}
    		\includegraphics[width=.5\textwidth,clip,keepaspectratio]{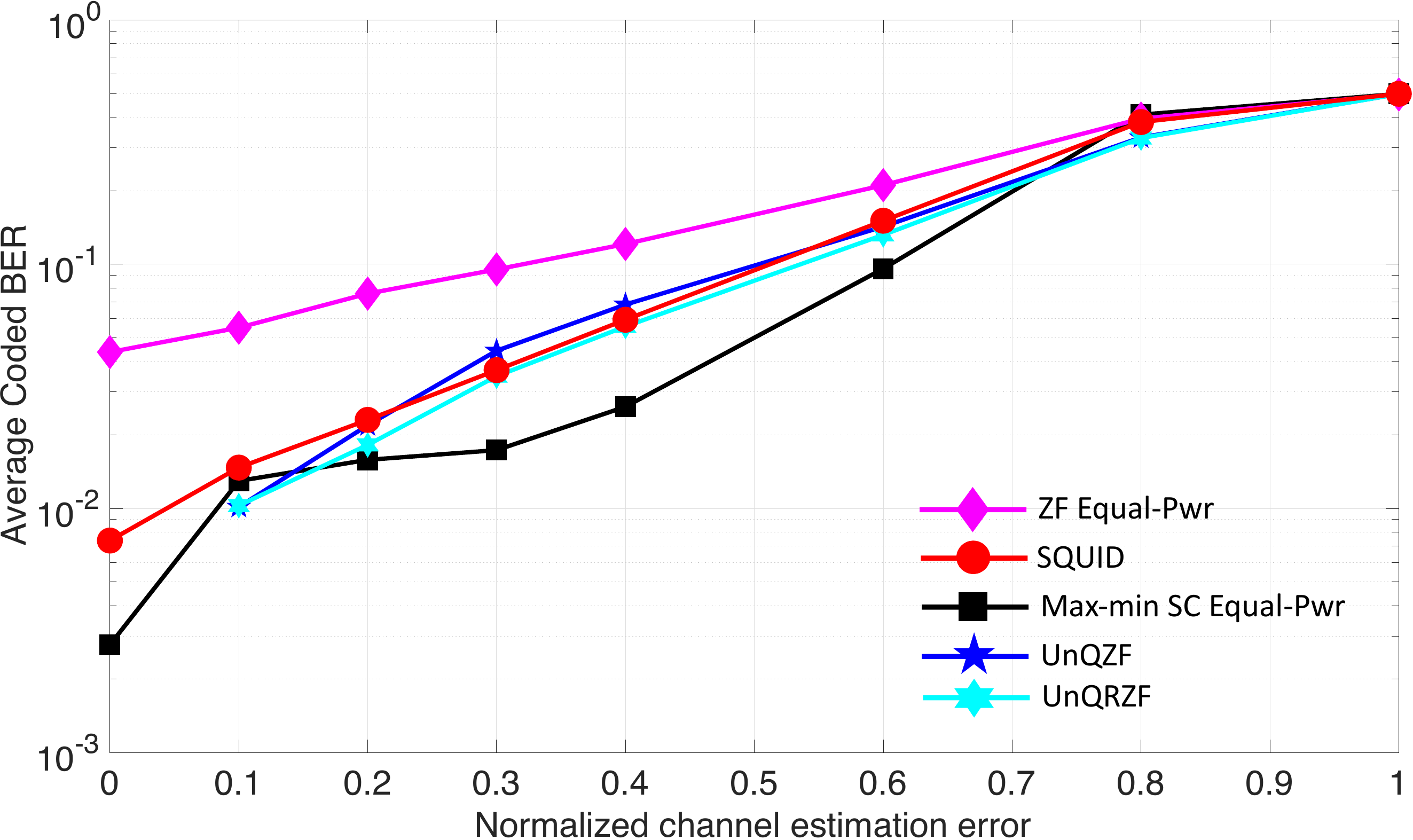}
    	\end{center}
	 \vspace{-.3cm}
    	\caption{Coded BER for 16QAM constellation with a rate 1/2 convolution code for $K = 5$, $\SG = 64$, $P_{\BS} = 40$ dBm and $b = 3$ on the LoS channel model. This result demonstrates that the proposed solution outperforms existing algorithms for higher order constellations over a wide range of the normalized channel estimation error.}
    	 \label{fig:ber4Supp}
	 \vspace{-.4cm}
\end{figure}

\begin{figure}[t]
    	\begin{center}
    		\includegraphics[width=.5\textwidth,clip,keepaspectratio]{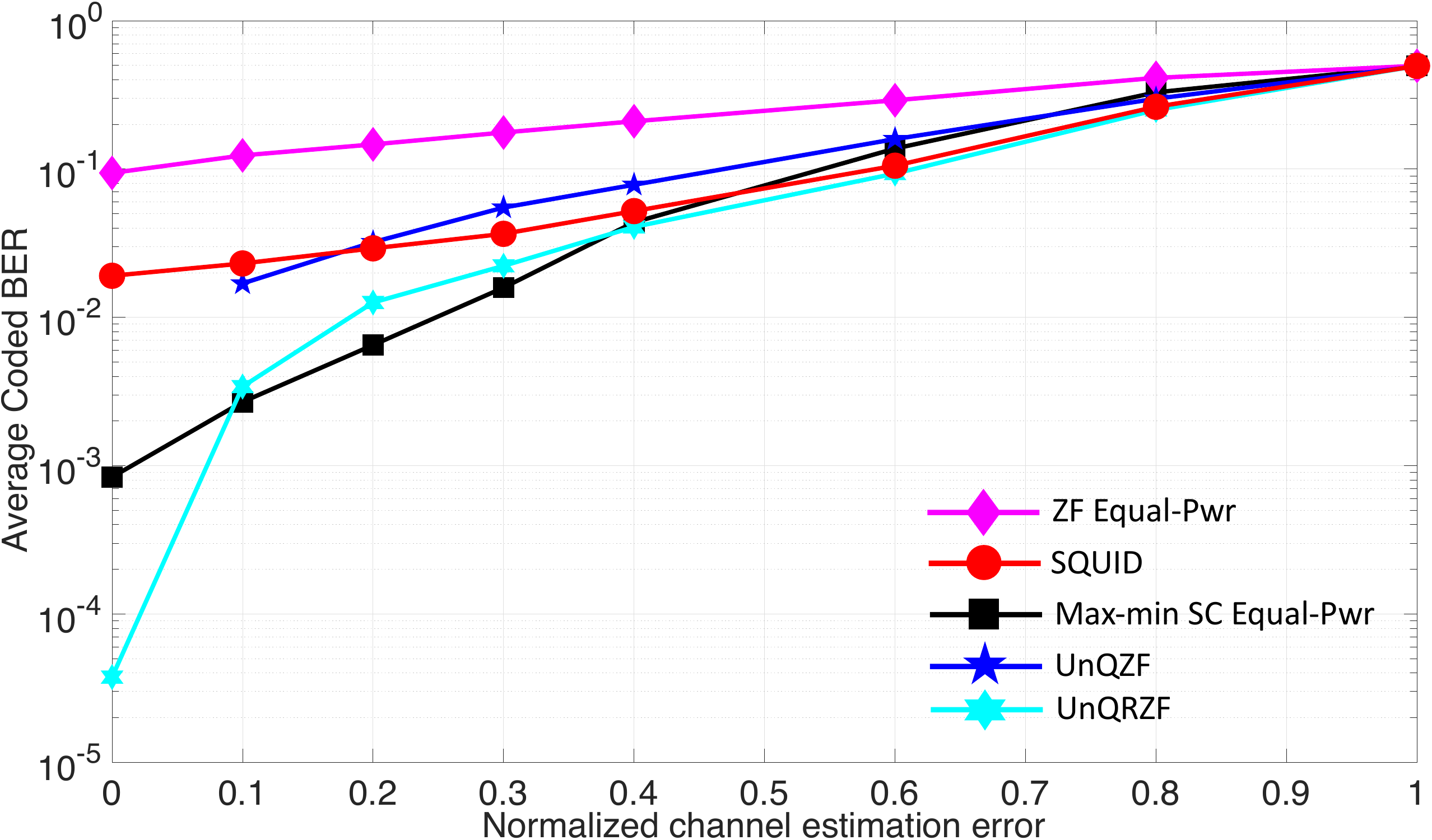}
    	\end{center}
	 \vspace{-.3cm}
    	\caption{Coded BER for QPSK constellation with a rate 1/2 convolution code for $K = 5$, $\SG = 64$, $P_{\BS} = 51$ dBm, and $b = 2$ on the NLoS channel model. This result demonstrates that the proposed solution outperforms existing algorithms over NLoS channel models as well over a wide range of the normalized channel estimation error.}
    	 \label{fig:ber42Supp}
	 \vspace{-.4cm}
\end{figure}

The results presented in this note supplement the results presented in Section \ref{sec:results} of the original manuscript and demonstrate that the proposed max-min formulation and the resulting solution significantly outperforms existing precoding strategies over a wide range of system parameters.

\newpage

\section{Supplementary Appendices}

In this supplementary note, we give the proofs of Lemma \ref{lemma:lemma1} and Lemma \ref{lemma:lemma5} for ease of reviewers. These proofs closely follow the proofs of \cite[Lemma 3.2 and Lemma 3.3]{pw2} and \cite[Lemma 4.3]{pw2} in philosophy with a few changes accounting for the OFDM signal model and CEQs.

\appendices
\section{Proof of lemma \ref{lemma:lemma1}}\label{sec:A}
By dropping the noise term on the RHS of (\ref{eq:downlinkSINR_covar2}), let us define the signal-to-quantization-plus-interference ratio (SQIR) for the $n^\thh$ sub-carrier of the $k^\thh$ user as 
\begin{equation}
\begin{split}
\hat{\gamma}_\kn^{\text{DL}}(\bar{\bm{\sfT}} , \bm{\sfq}) = \frac{ \sfq_\kn {\bm{\sft}}_\kn^\T {\bm{\sfR}}_\kn {\bm{\sft}}_\kn^*}{  \sum_{ \substack{i = 1\\ i \neq k}}^K \sfq_\inn {\bm{\sft}}_\inn^\T  {\bm{\sfR}}_\kn {\bm{\sft}}_\inn^*  + \left( \frac{1}{\zeta_b^2} - 1 \right) \tr \left(  \frac{1}{\SC}  \sum_{i = 1}^{K} \sum_{j = 1}^{\SC}  \sfq_\ijj \bm{\sft}_\ijj^\T \diag \left( {\bm{\sfR}}_\kn^\ast \right)  \bm{\sft}_\ijj^* \right)}
\label{eq:A1}
\end{split}
\end{equation}
It can be observed from (\ref{eq:A1}) that the SQIR is a constant function of scalar multiples of the DL power allocation vector $\bm{\sfq}$ i.e.
\begin{equation}
\hat{\gamma}_\kn^{\text{DL}}(\bar{\bm{\sfT}} , \lambda \bm{\sfq}) = \hat{\gamma}_\kn^{\text{DL}}(\bar{\bm{\sfT}} , \bm{\sfq}),
\label{eq:A2}
\end{equation}
for all positive $\lambda$. It can also be observed from (\ref{eq:downlinkSINR_covar2}) that the DL SQINR, ${\gamma}_\kn^{\text{DL}}(\bar{\bm{\sfT}} , \bm{\sfq}) $, is a monotonically increasing function of scalar multiples of the DL power allocation vector $\bm{\sfq}$ i.e.
\begin{equation}
{\gamma}_\kn^{\text{DL}}(\bar{\bm{\sfT}} , \lambda \bm{\sfq})  > {\gamma}_\kn^{\text{DL}}(\bar{\bm{\sfT}} , \bm{\sfq}) ,
\label{eq:A3}
\end{equation}
for $\lambda > 1$. Furthermore by comparing (\ref{eq:downlinkSINR_covar2}) and (\ref{eq:A1}), it can be seen that
\begin{equation}
\lim_{||\bm{\sfq}||_2 \to \infty } {\gamma}_\kn^{\text{DL}}(\bar{\bm{\sfT}} , \bm{\sfq}) = \lim_{||\bm{\sfq}||_2 \to \infty } \hat{\gamma}_\kn^{\text{DL}}(\bar{\bm{\sfT}} , \bm{\sfq}).
\label{eq:A4}
\end{equation}
For a target DL SQINR set $\{ \gamma_\kn \}$ to be feasible
\begin{equation} \label{eq:A5}
\begin{split}
1 &\leq \min_{\substack{ {1\leq k \leq K} \\ {1 \leq n \leq \SC}}}   \frac{{\gamma}_\kn^{\text{DL}}( \bar{\bm{\sfT}} , \bm{\sfq})}{\gamma_\kn} \\
   & \overset{(\ref{eq:A3})}{<} \max_{||\bm{\sfq}||_2 \to \infty } \left( \min_{\substack{ {1\leq k \leq K} \\ {1 \leq n \leq \SC}}}   \frac{{\gamma}_\kn^{\text{DL}}( \bar{\bm{\sfT}} , \bm{\sfq})}{\gamma_\kn} \right) \triangleq R^\star.
\end{split}
\end{equation}
Making use of (\ref{eq:A2}) and (\ref{eq:A4}), the upper bound (\ref{eq:A5}) is equivalently given by
\begin{equation} \label{eq:A6}
\begin{split}
R^\star & \overset{(\ref{eq:A4})}{=} \max_{||\bm{\sfq}||_2 \to \infty } \left( \min_{\substack{ {1\leq k \leq K} \\ {1 \leq n \leq \SC}}}   \frac{\hat{\gamma}_\kn^{\text{DL}}( \bar{\bm{\sfT}} , \bm{\sfq})}{\gamma_\kn} \right) \\
& \overset{(\ref{eq:A2})}{=} \max_{||\bm{\sfq}||_2 = 1 } \left(  \min_{\substack{ {1\leq k \leq K} \\ {1 \leq n \leq \SC}}}   \frac{\hat{\gamma}_\kn^{\text{DL}}( \bar{\bm{\sfT}} , \bm{\sfq})}{\gamma_\kn} \right).
\end{split}
\end{equation}
The solution to the optimization problem (\ref{eq:A6}) results in \emph{equal} achieved SQIR to target SQIR ratio for all $K \SC$ channels given by
\begin{equation} \label{eq:A7}
R^\star = \frac{\hat{\gamma}_{1,1}^{\text{DL}}( \bar{\bm{\sfT}} , \bm{\sfq}^\star)}{\gamma_{1,1}} = \dots = \frac{\hat{\gamma}_{K,\SC}^{\text{DL}}( \bar{\bm{\sfT}} , \bm{\sfq}^\star)}{\gamma_{K,\SC}},
\end{equation}
where $\bm{\sfq}^\star$ is the power allocation vector which solves (\ref{eq:A6}). This claim is proved in Appendix \ref{sec:C}. The $K \SC$ equations in (\ref{eq:A7}) can be written in matrix form as 
\begin{equation} \label{eq:A8}
\bm{\sfq}^\star \frac{1}{R^\star} =\bar{\bm{\sfD}}(\bar{\bm{\sfT}})  \bar{\bm{\Psi}}(\bar{\bm{\sfT}}) \bm{\sfq}^\star + \bar{\bm{\sfD}}(\bar{\bm{\sfT}})  \bar{\bm{\Phi}}(\bar{\bm{\sfT}})\bm{\sfq}^\star.
\end{equation}
It can be observed from (\ref{eq:A8}) that the achieved SQIR to target SQIR balance value, $R^\star$, equals the reciprocal of an eigenvalue of the matrix $\bar{\bm{\sfD}}(\bar{\bm{\sfT}})  \bar{\bm{\Psi}}(\bar{\bm{\sfT}}) + \bar{\bm{\sfD}}(\bar{\bm{\sfT}})  \bar{\bm{\Phi}}(\bar{\bm{\sfT}})$ and the optimal power allocation vector is given by the corresponding eigenvector. It is also known from Perron-Frobenius theory \cite{MUDL,pw2} that the optimal eigenvalue/eigenvector pair correspond to the maximal eigenvalue of the non-negative matrix $\bar{\bm{\sfD}}(\bar{\bm{\sfT}})  \bar{\bm{\Psi}}(\bar{\bm{\sfT}}) + \bar{\bm{\sfD}}(\bar{\bm{\sfT}})  \bar{\bm{\Phi}}(\bar{\bm{\sfT}})$. Hence
\begin{equation} \label{eq:A9}
\lambda_\maxx \left( \bar{\bm{\sfD}}(\bar{\bm{\sfT}})  \bar{\bm{\Psi}}(\bar{\bm{\sfT}}) + \bar{\bm{\sfD}}(\bar{\bm{\sfT}})  \bar{\bm{\Phi}}(\bar{\bm{\sfT}}) \right) = \frac{1}{R^\star} \overset{(\ref{eq:A5})}{<} 1.
\end{equation}
This establishes that for any feasible target SQINR set $\{ \gamma_\kn \}$, $\lambda_\maxx \left( \bar{\bm{\sfD}}(\bar{\bm{\sfT}})  \bar{\bm{\Psi}}(\bar{\bm{\sfT}}) + \bar{\bm{\sfD}}(\bar{\bm{\sfT}})  \bar{\bm{\Phi}}(\bar{\bm{\sfT}}) \right) < 1$. 

Next, assume that the matrix $\left(  \bm{\sfI}_{K\SC} -  \bar{\bm{\sfD}}(\bar{\bm{\sfT}})  \bar{\bm{\Psi}}(\bar{\bm{\sfT}}) - \bar{\bm{\sfD}}(\bar{\bm{\sfT}})  \bar{\bm{\Phi}}(\bar{\bm{\sfT}}) \right)$ is not invertible. This must mean that for some vector $\bm{\sfb}$
\begin{equation} \label{eq:B1}
\begin{split}
\left(  \bm{\sfI}_{K\SC} -  \bar{\bm{\sfD}}(\bar{\bm{\sfT}})  \bar{\bm{\Psi}}(\bar{\bm{\sfT}}) - \bar{\bm{\sfD}}(\bar{\bm{\sfT}})  \bar{\bm{\Phi}}(\bar{\bm{\sfT}}) \right) \bm{\sfb} &= \bm{0}_{K\SC} \\
\Rightarrow  \left( \bar{\bm{\sfD}}(\bar{\bm{\sfT}})  \bar{\bm{\Psi}}(\bar{\bm{\sfT}}) +  \bar{\bm{\sfD}}(\bar{\bm{\sfT}})  \bar{\bm{\Phi}}(\bar{\bm{\sfT}}) \right) \bm{\sfb}  &=  \bm{\sfb}.
\end{split}
\end{equation}
\sloppy This implies that the matrix $\bar{\bm{\sfD}}(\bar{\bm{\sfT}})  \bar{\bm{\Psi}}(\bar{\bm{\sfT}}) + \bar{\bm{\sfD}}(\bar{\bm{\sfT}})  \bar{\bm{\Phi}}(\bar{\bm{\sfT}})$ has an eigenvalue equal to 1. We know from (\ref{eq:A9}) that $\lambda_\maxx \left( \bar{\bm{\sfD}}(\bar{\bm{\sfT}})  \bar{\bm{\Psi}}(\bar{\bm{\sfT}}) + \bar{\bm{\sfD}}(\bar{\bm{\sfT}})  \bar{\bm{\Phi}}(\bar{\bm{\sfT}}) \right) < 1$. Hence this is a contradiction and the matrix $\left(  \bm{\sfI}_{K\SC} -  \bar{\bm{\sfD}}(\bar{\bm{\sfT}})  \bar{\bm{\Psi}}(\bar{\bm{\sfT}}) - \bar{\bm{\sfD}}(\bar{\bm{\sfT}})  \bar{\bm{\Phi}}(\bar{\bm{\sfT}}) \right)$ is invertible for any feasible target DL SQINR set $\{ \gamma_\kn \}$.

\section{Proof of (\ref{eq:A7}) }\label{sec:C}
Let $(i,j)$ be the user-subcarrier index such that 
\begin{equation} \label{eq:C1}
\frac{\hat{\gamma}_\ijj^{\text{DL}}( \bar{\bm{\sfT}} , \bm{\sfq^\star})}{\gamma_\ijj} > R^\star =  \min_{\substack{ {1\leq k \leq K} \\ {1 \leq n \leq \SC}}}   \frac{\hat{\gamma}_\kn^{\text{DL}}( \bar{\bm{\sfT}} , \bm{\sfq}^\star )}{\gamma_\kn} .
\end{equation}
It can be seen from (\ref{eq:A1}) that the DL SQIR $\hat{\gamma}_\kn^{\text{DL}}(\hat{\bm{\sfT}} , \bm{\sfq})$ is an increasing function of $\sfq_\kn$ and a decreasing function of $\sfq_{\ell,m}$ for $\ell \neq k$ and $n \neq m$. The power allocated to the $(i,j)$ user-subcarrier pair, $\sfq_\ijj$, can be decreased without reducing the objective function $\min_{\substack{ {1\leq k \leq K} , {1 \leq n \leq \SC}}}   \frac{\hat{\gamma}_\kn^{\text{DL}}( \bar{\bm{\sfT}} , \bm{\sfq}^\star )}{\gamma_\kn}$. This excess power can then be allocated to the $(\ell,m)$ user-subcarrier pair whose achieved to target SQIR ratio equals $R^\star$ thus resulting in a larger optimum value of the objective function $\min_{\substack{ {1\leq k \leq K} , {1 \leq n \leq \SC}}}   \frac{\hat{\gamma}_\kn^{\text{DL}}( \bar{\bm{\sfT}} , \bm{\sfq}^\star )}{\gamma_\kn}$. Consequently, the initial assumption was a contradiction and all $\SC$ sub-carriers of the $K$ users achieve the same achieved SQIR to target SQIR ratio.

\section{Proof of lemma \ref{lemma:lemma5}}\label{sec:D}
It was shown in \cite{MUDL} that for any positive $N$-dimensional vectors $\bm{\sfb}$ and $\bm{\sfc}$
\begin{equation}  \label{eq:D1}
\max_{\bm{\sfx}} \frac{ \bm{\sfx}^\T \bm{\sfb} }{ \bm{\sfx}^\T \bm{\sfc} } = \max_{ 1 \leq n \leq N } \frac{ \bm{\sfb}_n }{ \bm{\sfc}_n }.
\end{equation}
Using (\ref{eq:D1}) and the non-negativity of $\bar{\bm{\Lambda}}(\bar{\bm{\sfT}} , P_{\BS}) \bm{\sfp}_\ext $, it follows that
\begin{equation} \label{eq:D2}
 \bar{\lambda} \left( \bar{\bm{\sfT}}, P_{\BS}, \bm{\sfp}_\ext \right) = \max_{1 \leq n \leq K\SC +1}  \frac{ \bm{\sfee}_n^\T \bar{\bm{\Lambda}}(\bar{\bm{T}} , P_{\BS}) \bm{\sfp}_\ext  }{ \bm{\sfee}_n^\T \bm{\sfp}_\ext }.
\end{equation} 
Using (\ref{eq:uplinkSINR_covar5}) and (\ref{eq:gammaMatrix}), the first $K \SC$ equations in (\ref{eq:D2}) can be written as
\begin{equation} \label{eq:D3}
\max_{1 \leq n \leq K\SC}  \frac{ \bm{\sfee}_n^\T \bar{\bm{\Lambda}}(\bar{\bm{\sfT}} , P_{\BS}) \bm{\sfp}_\ext   }{ \bm{\sfee}_n^\T \bm{\sfp}_\ext } = \max_{1 \leq n \leq K\SC} \frac{\gamma_\kn}{\gamma_\kn^{\text{UL}}\left(  \bm{\sft}_\kn ,  \bm{\sfp} \right) }.
\end{equation}
It also follows from (\ref{eq:uplinkSINR_covar5}) and (\ref{eq:gammaMatrix}) that
\begin{equation} \label{eq:D4}
\begin{split}
\frac{ \bm{\sfee}_{K\SC+1}^\T \bar{\bm{\Lambda}}(\bar{\bm{\sfT}} , P_{\BS}) \bm{\sfp}_\ext   }{ \bm{\sfee}_{K\SC+1}^\T \bm{\sfp}_\ext } & = \frac{1}{P_\BS \SC} \sum_{k=1}^K \sum_{n=1}^{\SC}  \frac{\bm{\sfp}_\kn \gamma_\kn}{\gamma_\kn^{\text{UL}}\left(  \bm{\sft}_\kn ,  \bm{\sfp} \right) }\\
& \overset{(a)}{\leq}  \left( \max_{\substack{ {1\leq k \leq K} \\ {1 \leq n \leq \SC}}}\frac{\gamma_\kn}{\gamma_\kn^{\text{UL}}\left(  \bm{\sft}_\kn ,  \bm{\sfp} \right) } \right) \frac{1}{P_\BS \SC} \sum_{k=1}^K  \sum_{n=1}^{\SC} \bm{\sfp}_\kn \\
& =  \max_{\substack{ {1\leq k \leq K} \\ {1 \leq n \leq \SC}}} \frac{\gamma_\kn}{\gamma_\kn^{\text{UL}}\left(  \bm{\sft}_\kn ,  \bm{\sfp} \right) }.
\end{split}
\end{equation}
(a) follows because the max is greater than the average. This shows that the $(K\SC+1)^\thh$ equation in (\ref{eq:D2}) is smaller than or equal to the first $K\SC$ equations. Hence
\begin{equation} \label{eq:D5}
 \bar{\lambda} \left( \bar{\bm{\sfT}}, P_{\BS}, \bm{\sfp}_\ext \right) = \max_{\substack{ {1\leq k \leq K} \\ {1 \leq n \leq \SC}}}\frac{\gamma_\kn}{\gamma_\kn^{\text{UL}}\left(  \bm{\sft}_\kn ,  \bm{\sfp} \right) }.
\end{equation}

\end{document}